\newif\ifanon
\anonfalse 

\ifanon
\documentclass[acmsmall,anonymous,review,nonacm]{acmart}	
\else
\documentclass[acmsmall,nonacm]{acmart}
\fi

\makeatletter
\@ifundefined{lhs2tex.lhs2tex.sty.read}%
  {\@namedef{lhs2tex.lhs2tex.sty.read}{}%
   \newcommand\SkipToFmtEnd{}%
   \newcommand\EndFmtInput{}%
   \long\def\SkipToFmtEnd#1\EndFmtInput{}%
  }\SkipToFmtEnd

\newcommand\ReadOnlyOnce[1]{\@ifundefined{#1}{\@namedef{#1}{}}\SkipToFmtEnd}
\usepackage{newtxmath}
\usepackage{stmaryrd}
\DeclareFontFamily{OT1}{cmtex}{}
\DeclareFontShape{OT1}{cmtex}{m}{n}
  {<5><6><7><8>cmtex8
   <9>cmtex9
   <10><10.95><12><14.4><17.28><20.74><24.88>cmtex10}{}
\DeclareFontShape{OT1}{cmtex}{m}{it}
  {<-> ssub * cmtt/m/it}{}

\DeclareFontShape{OT1}{cmtt}{bx}{n}
  {<5><6><7><8>cmtt8
   <9>cmbtt9
   <10><10.95><12><14.4><17.28><20.74><24.88>cmbtt10}{}
\DeclareFontShape{OT1}{cmtex}{bx}{n}
  {<-> ssub * cmtt/bx/n}{}

\newcommand{\Conid}[1]{\mathit{#1}}
\newcommand{\Varid}[1]{\mathit{#1}}
\newcommand{\anonymous}{\kern0.06em \vbox{\hrule\@width.5em}}

\usepackage{polytable}

\@ifundefined{mathindent}%
  {\newdimen\mathindent\mathindent\leftmargini}%
  {}%

\def\resethooks{%
  \global\let\SaveRestoreHook\empty
  \global\let\ColumnHook\empty}
\newcommand*{\savecolumns}[1][default]%
  {\g@addto@macro\SaveRestoreHook{\savecolumns[#1]}}
\newcommand*{\restorecolumns}[1][default]%
  {\g@addto@macro\SaveRestoreHook{\restorecolumns[#1]}}
\newcommand*{\aligncolumn}[2]%
  {\g@addto@macro\ColumnHook{\column{#1}{#2}}}

\resethooks

\newcommand{\onelinecommentchars}{\quad-{}- }
\newcommand{\commentbeginchars}{\enskip\{-}
\newcommand{\commentendchars}{-\}\enskip}

\newcommand{\visiblecomments}{%
  \let\onelinecomment=\onelinecommentchars
  \let\commentbegin=\commentbeginchars
  \let\commentend=\commentendchars}

\newcommand{\invisiblecomments}{%
  \let\onelinecomment=\empty
  \let\commentbegin=\empty
  \let\commentend=\empty}

\visiblecomments

\newlength{\blanklineskip}
\newcommand{\hsindent}[1]{\quad}
\let\hspre\empty
\let\hspost\empty

\EndFmtInput
\makeatother
\ReadOnlyOnce{polycode.fmt}%
\makeatletter

\newcommand{\hsnewpar}[1]%
  {{\parskip=0pt\parindent=0pt\par\vskip #1\noindent}}

\newcommand{\hscodestyle}{}

\newcommand{\sethscode}[1]%
  {\expandafter\let\expandafter\hscode\csname #1\endcsname
   \expandafter\let\expandafter\endhscode\csname end#1\endcsname}

  {\par\noindent
   \advance\leftskip\mathindent
   \hscodestyle
   \let\\=\@normalcr
   \let\hspre\(\let\hspost\)%
   \pboxed}%
  {\endpboxed\)%
   \par\noindent
   \ignorespacesafterend}

  {\hsnewpar\abovedisplayskip
   \advance\leftskip\mathindent
   \hscodestyle
   \let\hspre\(\let\hspost\)%
   \pboxed}%
  {\endpboxed%
   \hsnewpar\belowdisplayskip
   \ignorespacesafterend}

  {\hsnewpar\abovedisplayskip
   \advance\leftskip\mathindent
   \hscodestyle
   \let\\=\@normalcr
   \(\pboxed}%
  {\endpboxed\)%
   \hsnewpar\belowdisplayskip
   \ignorespacesafterend}

\newcommand{\plainhs}{\sethscode{plainhscode}}

\plainhs

  {\hsnewpar\abovedisplayskip
   \advance\leftskip\mathindent
   \hscodestyle
   \let\\=\@normalcr
   \(\parray}%
  {\endparray\)%
   \hsnewpar\belowdisplayskip
   \ignorespacesafterend}

  {\parray}{\endparray}

  {\(\parray}{\endparray\)}

\def\codeframewidth{\arrayrulewidth}
\RequirePackage{calc}

  {\parskip=\abovedisplayskip\par\noindent
   \hscodestyle
   \arrayrulewidth=\codeframewidth
   \tabular{@{}|p{\linewidth-2\arraycolsep-2\arrayrulewidth-2pt}|@{}}%
   \hline\framedhslinecorrect\\{-1.5ex}%
   \let\endoflinesave=\\
   \let\\=\@normalcr
   \(\pboxed}%
  {\endpboxed\)%
   \framedhslinecorrect\endoflinesave{.5ex}\hline
   \endtabular
   \parskip=\belowdisplayskip\par\noindent
   \ignorespacesafterend}

\newcommand{\framedhslinecorrect}[2]%
  {#1[#2]}

  {\(\def\column##1##2{}%
   \let\>\undefined\let\<\undefined\let\\\undefined
   \newcommand\>[1][]{}\newcommand\<[1][]{}\newcommand\\[1][]{}%
   \def\fromto##1##2##3{##3}%
   }{\) }%

  {\let\orighscode=\hscode
   \let\origendhscode=\endhscode
   \def\endhscode{\def\hscode{\endgroup\def\@currenvir{hscode}\\}\begingroup}
   \orighscode\def\hscode{\endgroup\def\@currenvir{hscode}}}%
  {\origendhscode
   \global\let\hscode=\orighscode
   \global\let\endhscode=\origendhscode}%

\makeatother
\EndFmtInput

\usepackage{amsmath,amsfonts,amsthm,stmaryrd}
\usepackage{xspace}
\usepackage{mathpartir}
\usepackage{natbib}
\usepackage[shortlabels,inline]{enumitem}
\usepackage{hyperref}
\usepackage{tikz}
\usetikzlibrary{positioning,calc,shapes.multipart,arrows.meta}
\usepackage{bussproofs}
\EnableBpAbbreviations 
\usepackage{color}
\usepackage{algorithmicx}
\usepackage{algpseudocode}
\usepackage{algorithm}
\usepackage{epigraph}
\usepackage{suffix}
\usepackage{mathtools}
\usepackage{verbatim}
\usepackage{ifthen}
\usepackage{subcaption}
\usepackage{paracol}
\usepackage{array}
\usepackage[table]{xcolor}
\usepackage{colortbl}

\newcommand{\card}[1]{

}
\newcommand\Rizzo{\textsf{Rizzo}\xspace}
\newcommand\now{\textit{now}\xspace}
\newcommand\earlier{\textit{earlier}\xspace}
\newcommand\chans{\ensuremath{\sym{Chan}}\xspace}
\newcommand\locs{\ensuremath{\sym{Loc}}\xspace}

\newcommand\nin{\not\in}

\newcommand\cl[2][]{\ifthenelse{\equal{#1}{}}{\sym{cl}\!\left(#2\right)}{\sym{cl}_{#1}\left(#2\right)}}
\newcommand\ticked[3][\kappa]{\sym{ticked}^{#1}_{#2}\left(#3\right)}

\newcommand\ignore[1]{}

\newcommand\sym[1]{\mathsf{#1}}
\newcommand\con[1]{\mathtt{#1}}

\newcommand\pair[2]{\left(#1,#2\right)}
\newcommand\set[1]{\left\{#1\right\}}
\newcommand{\setcom}[2]{\set{#1\left\vert\vphantom{#1}\,#2\right.}}

\newcommand{\tick}{\checkmark}
\newcommand{\sep}{\mathbin{\!{/}\!}}
\newcommand\emptyheap{{\cdot}}

\newcommand{\T}[1]{\,\langle#1\rangle\,}

\newcommand\dom[1]{\sym{dom}\left(#1\right)}

\newcommand\heaple{\sqsubseteq}
\newcommand\heapge{\sqsupseteq}

\newcommand\allocate[1]{\sym{alloc}\left(#1\right)}

\newcommand{\interm}{\sym{in}}
\newcommand{\caseterm}[5]{\sym{case}\, #1 \, \sym{of}\, \interm_1\,#2 . #3 ; \interm_2\,#4 . #5}

\newcommand\unit{()}

\newcommand\fmap[1]{\sym{fmap}_{#1}}
\newcommand\lfp{\mathit{lfp}}

\newcommand\fix{\sym{fix}}

\newcommand\chan[1][A]{\sym{chan}_{#1}}
\newcommand\wait{\sym{wait}}

\newcommand\select{\sym{select}}
\newcommand\delay{\sym{delay}}

\newcommand\never{\sym{never}}

\newcommand{\sig}[3][U]{#2 \,\langle #1 \rangle\, #3}
\newcommand{\head}{\sym{head}}
\newcommand{\tail}{\sym{tail}}
\newcommand{\watch}{\sym{watch}}
\newcommand{\cons}[1][]{\sym{cons}_{#1}}
\newcommand{\appA}{\varoast}
\newcommand{\appAE}{\varogreaterthan}
\newcommand{\appIE}{\rhd}

\newcommand{\updated}{\top}
\newcommand{\unchanged}{\bot}
\newcommand\delaySig[1]{\mathit{sig}\!\left[#1\right]}

\newcommand\rec[4][]{\rec*_{#1}(#2.#3,#4)}
\WithSuffix\newcommand\rec*{\sym{rec}}

\newcommand\Delay{{\bigcirc}}
\newcommand\DelayE{\ensuremath{\mathrlap{\hspace{0.23em}{\scalebox{0.75}{$\exists$}}}{\Delay}}}
\WithSuffix\newcommand\DelayE*{\ensuremath{\mathrlap{\hspace{0.21em}{\scalebox{0.6}{$\exists$}}}{\Delay}}}
\newcommand\DelayA{\ensuremath{\mathrlap{\hspace{0.25em}{\scalebox{0.75}{$\forall$}}}{\Delay}}}
\WithSuffix\newcommand\DelayA*{\ensuremath{\mathrlap{\hspace{0.19em}{\scalebox{0.6}{$\forall$}}}{\Delay}}}

\newcommand\Sig{\sym{Sig}}

\newcommand\List{\sym{List}}
\newcommand\Chan{\sym{Chan}}

\newcommand\Nat{\sym{Nat}}

\newcommand\Unit{\sym{1}}

\newcommand\nats{{\mathbb N}}

\newcommand{\heaptype}[1]{\left|#1\right|}

\newcommand{\hinterp}[1]{\mathcal{H}(#1)}

\newcommand\vinterp[3][\rho]{%
  \mathcal{V}^{#1}\llbracket #2\rrbracket(#3)
}
\newcommand\tinterp[3][\rho]{%
  \mathcal{T}^{#1}\llbracket #2\rrbracket(#3)
}
\newcommand\cinterp[2]{\mathcal{C}\llbracket #1\rrbracket(#2)}

\newcommand{\sigtype}[1]{: \Sig\,#1}
\newcommand{\chantype}[1]{: \Chan\,#1}
\newcommand{\synmap}[1]{#1^{\sym{syn}}}

\newcommand\rhosyn{\synmap\rho}
\newcommand\rhosubs{[\rho]}

\newcommand\isenv[1]{\vdash #1 : \sym{env}}
\newcommand\isev[2][\Delta]{\vdash_{#1} #2 : \sym{event}}
\newcommand\type{\sym{type}}

\newcommand\isheap[2][\Delta]{\vdash_{#1} #2 : \sym{now}}
\newcommand\isearlier[3][\Delta]{\vdash^{#2}_{#1} #3 : \sym{earlier}}
\newcommand\hastype[4][\Delta]{#2 \vdash_{#1} #3 : #4}
\newcommand\Hastype[3]{#1 \Vdash #2 : #3}
\WithSuffix\newcommand\hastype*[1][H]{
  \def\hastypeH{#1}%
  \hastypeI
}
\newcommand\hastypeI[4][\Delta]{%
#2\vdash^{\hastypeH}_{#1} #3 : #4}

\newcommand\istype[2]{#1 \vdash #2: \type}

\newcommand\evalAdv[5][e]{ \state{#2;#3} \Searrow^{#1}
  \state{#4;#5}}

\newcommand\reacts[1]{\leadsto_{#1}}

\renewcommand\state[1]{\left\langle #1 \right\rangle}
\newcommand\heval[4]{ \state{#1;#2} \Downarrow
 \state{#3;#4}}

\newcommand\gcstep[1]{\underset{\sym{gc}}{\stackrel{#1}{\Longrightarrow}}}
\WithSuffix\newcommand\gcstep*{\underset{\sym{gc}}{\Longrightarrow}}
\newcommand\initstep{\fullstep{\sym{init}}}
\newcommand\fullstep[1]{\stackrel{#1}{\Longrightarrow}}

\newcommand\partstep[1]{\stackrel{#1}{\Longmapsto}}
\WithSuffix\newcommand\partstep*[1]{\stackrel{#1}{\Longmapsto^*}}

\startPage{1}

\begin{document}

\title{Simple Modal Types for Functional Reactive Programming}

\author{Patrick Bahr}
\authorsaddresses{}
\affiliation{
  \institution{IT University of Copenhagen}            
  \country{Denmark}                    
}

\begin{abstract}
Functional reactive programming (FRP) is a declarative programming
paradigm for implementing reactive programs at a high level of
abstraction. It applies functional programming principles to construct
and manipulate time-varying values, also known as \emph{signals}.
However, for this programming paradigm to work in practice, an FRP
language must ensure that programs are \emph{causal},
\emph{productive}, and free of \emph{space leaks}. Over the past
fifteen years, several modal type systems to enforce these operational
properties have been developed.

We present a new FRP language with a significantly simplified modal
type system that imposes fewer restrictions than previous modal FRP
languages while still guaranteeing the central operational properties
of causality, productivity, and absence of space leaks. The key
enabling idea is to define the language's semantics so that it is
impossible to inspect a signal's past. In particular, a signal only
stores its current value and how to compute its future values, but it
cannot retain its past values. As a result, we obtain a language with
a simpler modal type system that is also more expressive and supports
a more modular programming style compared to
previous modal FRP calculi without space leaks. While the central idea
applies to both synchronous and asynchronous FRP, we focus here on the
more challenging asynchronous case.
\end{abstract}

\maketitle

\section{Introduction}
\label{sec:introduction}

Functional reactive programming
(FRP)~\citep{elliott97FunctionalReactiveAnimation,wan00FunctionalReactiveProgramming}
is a declarative programming paradigm for implementing reactive programs.
In turn, a reactive program is an indefinitely running piece of software that
interacts with its environment by receiving input from the
environment and, in return, producing output that is sent back to the
environment. This class of software is ubiquitous -- ranging from
graphical user interfaces and servers to safety-critical control
software for components in aircraft and power plants.

The central abstraction employed by FRP to model the interaction with
the environment is the notion of a
\emph{signal}~\citep{courtney01Fruit} (or
\emph{behaviour}~\citep{elliott97FunctionalReactiveAnimation}), which
represents a time-varying value. FRP uses functional programming
principles to construct, consume, and manipulate such signals. This
approach affords a high level of abstraction and modularity, along with
equational reasoning principles.

However, this high level of abstraction makes it difficult to compile
FRP programs into effective and efficient low-level software. To allow
such low-level implementations in general, the FRP language must guarantee three
properties: \emph{productivity}, that is, after receiving an input, the
program must produce its response in finite time; \emph{causality},
that is, at any time, the output may depend only on current or
past inputs; and \emph{no space leaks}, that is, the program may not retain
input values in memory indefinitely. Reactive programs that do not
satisfy these properties will eventually grind to a halt and thus
become unresponsive.

Several approaches have been proposed to guarantee some or all of
these operational properties (cf.\ section~\ref{sec:related-work}). In
this article, we are particularly interested in a type-based approach,
which allows programmers direct access to signals but carefully
controls that access using modal types inspired by linear temporal
logic~\citep{pnueli77TemporalLogicPrograms}. Originally proposed by
\citet{krishnaswami11UltrametricSemanticsReactive}, this approach has
spawned several FRP calculi that use modal types to guarantee
productivity, causality, and the absence of space
leaks~\citep{krishnaswami13HigherorderFunctionalReactive,bahr19SimplyRaTTFitchstyle,bahr22ModalFRPAll}.
The type systems of these calculi feature two type modalities: the
\emph{later} modality $\Delay$ to classify data that is available in the next
time step and the \emph{stable} modality $\Box$ to classify data that can be
kept in memory across time steps without causing space leaks.

\citet{bahr23AsynchronousModalFRP} have shown that this approach
extends to \emph{asynchronous} FRP, where there is no global
clock and where signals thus may update asynchronously from one
another. Asynchronous FRP provides a more suitable programming model
for systems that lack a uniform global clock such as
GUIs~\citep{graulund21AdjointReactiveGUI}. To capture this
asynchronicity with a modal type system,
\citeauthor{bahr23AsynchronousModalFRP}'s Async RaTT calculus replaces
the synchronous later modality \ensuremath{\Delay} with an asynchronous later modality
\ensuremath{\DelayE}. While a value of type \ensuremath{\Delay\Conid{A}} is a delayed computation that
produces a value of type \ensuremath{\Conid{A}} at the next tick of the global clock, a
value of type \ensuremath{\DelayE\;\Conid{A}} is a pair \ensuremath{(\theta,\Varid{f})} consisting of a clock
\ensuremath{\theta} and a delayed computation \ensuremath{\Varid{f}} that produces a value of type
\ensuremath{\Conid{A}} as soon as \ensuremath{\theta} ticks. \emph{Asynchronous} signals are then
represented by the guarded recursive type $\ensuremath{\Sig\;\Conid{A}} \cong \ensuremath{\Conid{A}\times\DelayE\;(\Sig\;\Conid{A})}$. That is, a value of type $\ensuremath{\Sig\;\Conid{A}}$ is of the form \ensuremath{(\Varid{v},(\theta,\Varid{f}))} consisting of the signal's current value $v : A$, a clock
\ensuremath{\theta}, and a delayed computation \ensuremath{\Varid{f}} that produces a new value of
type \ensuremath{\Sig\;\Conid{A}} as soon as \ensuremath{\theta} ticks. Consequently, we can
intuitively think of an asynchronous signal in Async RaTT as a
(possibly infinite) sequence $v_0 \T{\theta_0} v_1 \T{\theta_1}
\ldots$ consisting of lazily computed clocks $\theta_i$ and values
$v_i$ of type \ensuremath{\Conid{A}}, each of which only becomes available after the
preceding clocks $\theta_0, \ldots, \theta_{i-1}$ have ticked \emph{in
succession}. For example, $\theta_2$ and $v_2$ are available after
first $\theta_0$ and then $\theta_1$ have ticked.

In this article, we present a new asynchronous FRP language, called
\Rizzo, that uses the same representation of asynchronous signals as
Async RaTT~\citep{bahr23AsynchronousModalFRP}, i.e.\ $\ensuremath{\Sig\;\Conid{A}} \cong \ensuremath{\Conid{A}\times\DelayE\;(\Sig\;\Conid{A})}$. However, instead of using a sophisticated Fitch-style
type system~\citep{clouston18FitchStyleModalLambda} and an additional
type modality \ensuremath{\Box} to rule out space leaks, we define the semantics
of \Rizzo to rule out space leaks \emph{by construction}. This allows
us to give \Rizzo a much simpler type system and to extend its
expressiveness compared to Async RaTT -- all while maintaining the three
crucial operational properties. The key idea is that when evaluating a
signal \ensuremath{\Varid{s}\mathbin{:}\Sig\;\Conid{A}} in response to a sequence of events
$\tau$, we only keep the current value of the signal and discard all
its history. That is, given a signal \ensuremath{\Varid{s}} that denotes a possibly
infinite sequence $v_0 \T{\theta_0} v_1 \T{\theta_1} \ldots$ and a
finite sequence of events $\tau$ that causes the clocks $\theta_0,
\ldots, \theta_{i-1}$ to tick (but not $\theta_i$), the evaluation of
\ensuremath{\Varid{s}} in response to $\tau$, written $s \reacts\tau
v$, produces a signal \ensuremath{\Varid{v}\mathbin{:}\Sig\;\Conid{A}} that denotes the sequence $v_i
\T{\theta_i} v_{i+1} \T{\theta_{i+1}} \ldots$, i.e.\ forgetting all
its prior history from before $\theta_{i-1}$ ticked. 

This forgetful semantics applies to the language as a whole: A term \ensuremath{\Varid{t}\mathbin{:}\Conid{A}} evaluates to a value \ensuremath{\Varid{v}\mathbin{:}\Conid{A}} in response to
$\tau$, written $t \reacts\tau v$. For example, we may have a GUI
program \ensuremath{\Varid{t}\mathbin{:}\Conid{Widget}} that implements the tree structure describing the
GUI (such as the DOM of a website), so that if $t \reacts\tau v$, then
the value $v$ is the state of the GUI after it has received the events
$\tau$. Importantly, we give \Rizzo an operational semantics that will
compute $t \reacts\tau v$ \emph{incrementally}. That is, after we
computed $t \reacts\tau v$ and then received an additional event $e$,
we can compute $t \reacts{\tau,e} v'$ from $v$ and $e$ alone.
Moreover, this update of $v$ to $v'$ is performed \emph{in place},
only replacing the parts of $v$ that have changed.

The main technical results of this article are the metatheoretical
properties of \Rizzo: We prove that the language is productive,
causal, and free of space leaks. The proof uses a combination of
a logical relations argument and a type preservation argument.
In addition, this article presents multiple examples that demonstrate
the expressiveness and the simplicity of \Rizzo compared to previous
work. The proofs of the metatheoretical results and all \Rizzo examples in
this article have been fully formalised using the \emph{Lean} theorem
prover~\citep{moura21LeanTheoremProver}. \ifanon%
This Lean formalisation is
available in the supplementary material.
\else%
This Lean formalisation is
available in the supplementary material~\citep{supplementary}. 
\fi%

\paragraph{Overview}

We introduce \Rizzo in section~\ref{sec:rizzo} and demonstrate its
expressiveness in section~\ref{sec:example}.
Section~\ref{sec:operational-semantics} presents the operational
semantics of \Rizzo and gives a precise account of the main technical
results. Section~\ref{sec:metatheory} sketches the proof of the main
results. Finally, section~\ref{sec:related-work} discusses related
work, and section~\ref{sec:conclusion} discusses conclusions and
future work.

\section{Introduction to \Rizzo}
\label{sec:rizzo}

We give an overview of the \Rizzo language, its type system, and an
informal account of its semantics. For a complete specification of its
syntax and type system, we refer the reader to
Figs.~\ref{fig:syntax}, \ref{fig:wftype}, and \ref{fig:typing}. The
formal semantics is presented in detail in
section~\ref{sec:operational-semantics}.

The language is based on the simply typed lambda calculus with
product, sum, and inductive types. To account for well-formed
inductive types of the form $ \mu \alpha.F$, we have an explicit type
formation judgement $\istype{\Phi}{F}$ defined in
Fig.~\ref{fig:wftype}. It states that $F$ is a well-formed type with
free type variables drawn from $\Phi$, and we call $F$ a \emph{closed
type} if $\istype{}{F}$. By convention, we use $A, B, C,$ and $D$ to
range over closed types, while $F$ and $G$ range over possibly open
types. We look more closely at the type formation rules in
section~\ref{sec:recursion} when we discuss recursive types. The
typing judgement $\hastype{\Gamma}{t}{A}$ defined in
Fig.~\ref{fig:typing} takes a term $t$, a closed type $A$, a
\emph{typing context} $\Gamma$ consisting of variable type assignments
of the form $x:B$, and a \emph{channel context} $\Delta$ consisting of
channel type assignments of the form $\kappa\chantype B$. We discuss
channels in more detail in section~\ref{sec:delayed-computations}
below. The term language defined in Fig.~\ref{fig:syntax} also
includes locations $l \in \locs$, but these are not typable by the
typing judgement and are only relevant for the operational semantics
discussed in section~\ref{sec:operational-semantics}.

\begin{figure}[t]
  \small
  \[
    \arraycolsep=2pt
    \begin{array}{llcl}
      \multicolumn{4}{l}{\text{Channels }\; \kappa \in  \chans, \hspace{1cm} \text{Locations }\; l \in \locs}\\
      \text{Events} &e&::=& \kappa \mapsto t\\
      \text{Types} &A,B,F, G &::= & \alpha \mid \Unit \mid F \times G
                                 \mid F + G \mid F \to G \mid
                                 \DelayE F\mid \DelayA F \mid
                                 \Sig\,F\mid\Chan\, F
      \mid\mu\alpha.F\\
      \text{Values} &v,w &::= &\unit \mid \lambda x.t \mid
                                  \pair{v}{w} \mid \interm_i\, v \mid
                                  v ::_A w \mid \kappa\mid \wait\,\kappa\mid \watch\,v \mid v \appAE w \mid\delay\, t\\
                                  &&\mid &\never\mid\select\,v\,w\mid\cons[\mu\alpha.F]\,v\\
      \text{Terms} &s,t &::= & v \mid x \mid \rec{x}{s}{t} \mid 
                                 \pair{s}{t} \mid \interm_i\, t
                                 \mid \pi_i\,t \mid t_1t_2  \mid
                                 \chan \mid
   \caseterm{t}{x}{t_1}{x}{t_2}\mid l\\
   &&\mid & s ::_A t \mid \select\, s\,t \mid s \appAE t \mid s \appA t\mid \wait\,t\mid\watch\,t \mid \fix\,x.t \mid \head\,t\mid\tail\,t \mid \cons[\mu\alpha.F]\,t
    \end{array}
  \]
  \vspace{-1em}
\caption{Syntax.}
\label{fig:syntax}
\end{figure}

\begin{figure}[t]
  \small
  \mprset{andskip=1em}
  \mprset{sep=1em}
  \begin{mathpar}
  \inferrule*%
  {\istype{\Phi}{F}}%
  {\istype{\Phi}{\Sig\,F}}%
  \and
  \inferrule*%
  {\istype{\Phi}{F} \\ \istype{\Phi}{G}}%
  {\istype{\Phi}{F \times G}}%
  \and
  \inferrule*%
  {\istype{\Phi}{F} \\ \istype{\Phi}{G}}%
  {\istype{\Phi}{F + G}}%
  \and
  \inferrule*%
  {\istype{}{F} \\ \istype{\Phi}{G}}%
  {\istype{\Phi}{F \to G}}%
  \and
    \inferrule*%
  {~}%
  {\istype{\Phi}{\Unit}}
  \and
  \inferrule*%
  {\istype{}{F}}%
  {\istype{\Phi}{\DelayE F}}%
  \and
  \inferrule*%
  {\istype{}{F}}%
  {\istype{\Phi}{\DelayA F}}%
  \and
  \inferrule*%
  {\istype{}{F}}%
  {\istype{\Phi}{\Chan\, F}}%
  \and
  \inferrule*%
  {\istype{\Phi,\alpha}{F}}%
  {\istype{\Phi}{\mu\,\alpha. F}}%
  \and
  \inferrule*%
  {\alpha \in \Phi}%
  {\istype{\Phi}{\alpha}}%
  \end{mathpar}
  \vspace{-1.5em}
  \caption{Type formation rules.}
  \label{fig:wftype}
\end{figure}

\begin{figure}[th]
  \small
  \mprset{andskip=0.9em}
  \mprset{sep=1em}
  \begin{mathpar}
  \inferrule*%
  {x:A \in \Gamma}%
  {\hastype{\Gamma}{x}{A}}%
  \and%
  \inferrule* {~}%
  {\hastype{\Gamma}{\unit}{\Unit}}%
  \and%
  \inferrule*%
  {\hastype{\Gamma,x:A}{t}{B}}%
  {\hastype{\Gamma}{\lambda x.t}{A \to B}}%
  \and%
  \inferrule*%
  {\hastype{\Gamma}{t}{A \to B} \\
    \hastype{\Gamma}{t'}{A}}%
  {\hastype{\Gamma}{t\,t'}{B}}%
  \and%
  \inferrule*%
  {\hastype{\Gamma}{t}{A_i}}
  {\hastype{\Gamma}{\interm_i\, t}{A_1 + A_2}}%
  \and%
  \inferrule*%
  {\hastype{\Gamma,x: A_i}{t_i}{B} \\
    \hastype{\Gamma}{t}{A_1 + A_2}}
  {\hastype{\Gamma}{\caseterm {t}{x}{t_1}{x}{t_2}}{B}}%
  \and
  \inferrule*%
  {\hastype{\Gamma}{t}{A} \\
    \hastype{\Gamma}{t'}{B}}%
  {\hastype{\Gamma}{\pair{t}{t'}}{A \times B}}%
  \and%
  \inferrule*%
  {\hastype{\Gamma}{t}{A_1 \times A_2}}
  {\hastype{\Gamma}{\pi_i\,t}{A_i}}%
  \and
    \inferrule*
  {\hastype{\Gamma}{t}{F[\mu\,\alpha.F/\alpha]}}
  {\hastype{\Gamma}{\cons[\mu\alpha.F]\,t}{\mu\,\alpha.F}}
  \and%
  \inferrule*%
  {\hastype{\Gamma,x:F[(\mu\alpha.F)\times A/\alpha]}{s}{A} \\ \hastype{\Gamma}{t}{\mu\alpha.F}}
  {\hastype{\Gamma}{\rec{x}{s}{t}}{A}}%
  \and%
  \inferrule*%
  {\kappa \chantype A \in \Delta}%
  {\hastype{\Gamma}{\kappa}{\Chan\,A}}%
  \and
  \inferrule*%
  {~}%
  {\hastype{\Gamma}{\chan}{\Chan\, A}}%
  \and%
  \inferrule*%
  {\hastype{\Gamma}{t}{A}}%
  {\hastype{\Gamma}{\delay\,t}{\DelayA A}}%
  \and%
  \inferrule*%
  {\hastype{\Gamma}{t}{\DelayA\,(A \to B)} \\
    \hastype{\Gamma}{t'}{\DelayA\, A}}%
  {\hastype{\Gamma}{t \appA t'}{\DelayA\,B}}%
  \and%
  \inferrule*%
  {\hastype{\Gamma}{t}{\DelayA\,(A \to B)} \\
    \hastype{\Gamma}{t'}{\DelayE\, A}}%
  {\hastype{\Gamma}{t \appAE t'}{\DelayE\,B}}%
    \and
  \inferrule*%
  {~}%
  {\hastype{\Gamma}{\never}{\DelayE A}}%
  \and%
  \inferrule*%
  {\hastype{\Gamma}{t}{\Chan\, A}}%
  {\hastype{\Gamma}{\wait\,t}{\DelayE A}}%
  \and
  \inferrule*%
  {\hastype{\Gamma}{t}{\Sig\, (A + 1)}}%
  {\hastype{\Gamma}{\watch\,t}{\DelayE A}}%
  \and
  \inferrule*%
  {\hastype{\Gamma}{t_1}{\DelayE A_1}
    \\\hastype{\Gamma}{t_2}{\DelayE A_2}}%
  {\hastype{\Gamma}{\select\,t_1\,t_2}{\DelayE ((A_1 + A_2) + (A_1 \times A_2))}}%
  \and
    \inferrule*
  {\hastype{\Gamma,x : \DelayA A}{t}{A}}
  {\hastype{\Gamma}{\fix\,x. t}{A}}
  \and%
  \inferrule*
  {\hastype{\Gamma}{t}{\Sig\,A}}%
  {\hastype{\Gamma}{\head\, t}{A}}%
  \and%
  \inferrule*
  {\hastype{\Gamma}{t}{\Sig\,A}}%
  {\hastype{\Gamma}{\tail\, t}{\DelayE\,(\Sig\,A)}}%
  \and%
  \inferrule*
  {\hastype{\Gamma}{s}{A}\\\hastype{\Gamma}{t}{\DelayE\,(\Sig\,A)}}%
  {\hastype{\Gamma}{s ::_A t}{\Sig\,A}}%
\end{mathpar}
\vspace{-1.5em}
\caption{Typing rules. We use $A, B$ to range over closed types as defined in Fig.~\ref{fig:wftype}.}
\label{fig:typing}
\end{figure}

\subsection{Signals and Delayed Computations}
\label{sec:signals}

\newcolumntype{C}{>{\centering\arraybackslash}p{0.79cm}}
\definecolor{lightgray}{gray}{0.9}
\newcolumntype{G}{>{\columncolor{lightgray}\centering\arraybackslash}p{0.79cm}}
\begin{figure}
 \begin{tabular}{r @{\hskip 1em} G C G C G C @{\hskip 1em} c}
$s_1 =$ & $1$ & & $2$ & $3$ & $4$ & & $\cdots$ \\
$s_2 =$ & 6 & 7 & 8 & & & 9 & $\cdots$ \\[0.4em]
$\ensuremath{\Varid{s}_{\mathrm{3}}\mathrel{=}\tail\;\Varid{s}_{\mathrm{1}}} =$ &  & & $2$ & $3$ & $4$ & & $\cdots$ \\
$\ensuremath{\Varid{map}\;(\lambda \Varid{x}.\Varid{x}\mathbin{+}\mathrm{1})\;\Varid{s}_{\mathrm{1}}} =$& $2$ & & $3$ & $4$ & $5$ & & $\cdots$ \\
$\ensuremath{\Varid{sample}\;\Varid{s}_{\mathrm{1}}\;\Varid{s}_{\mathrm{2}}} =$
 & $(1,6)$ & & $(2,8)$ & $(3,8)$ & $(4,8)$ & & $\cdots$ \\
$\ensuremath{\Varid{const}\;\mathrm{42}} =$& $42$ & & & & &  \\
$\ensuremath{\Varid{switch}\;\Varid{s}_{\mathrm{2}}\;\Varid{s}_{\mathrm{3}}} =$ 
 & $6$ & $7$ & $2$ & $3$ & $4$ & & $\cdots$\\
$\ensuremath{\Varid{mkSig}\;(\wait\;\kappa)} =$&  & & $7$ & $8$ &  & & $\cdots$ \\
$\ensuremath{\Varid{filter}\;\Varid{isEven}\;\Varid{s}_{\mathrm{3}}} =$ 
 & & & $2$ & & $4$ & & $\cdots$ \\
$\ensuremath{\Varid{scan}\;(\lambda \Varid{m}\;\Varid{n}.\Varid{m}\mathbin{+}\Varid{n})\;\mathrm{0}\;\Varid{s}_{\mathrm{1}}} =$& $1$ & & $3$ & $6$ & $10$ & & $\cdots$ \\
$\ensuremath{\Varid{scanAwait}\;(\lambda \Varid{m}\;\Varid{n}.\Varid{m}\mathbin{+}\Varid{n})\;\mathrm{0}\;\Varid{s}_{\mathrm{3}}} =$& $0$ & & $2$ & $5$ & $9$ & & $\cdots$ \\
$\ensuremath{\Varid{interleave}\;(\lambda \Varid{x}\;\Varid{y}.\Varid{x})\;\Varid{s}_{\mathrm{3}}\;(\tail\;\Varid{s}_{\mathrm{2}})} =$
 & & $7$ & $2$ & $3$ & $4$ & $9$ & $\cdots$ \\
$\Varid{zip}\, s_1\, s_2 =$ 
 & $(1,6)$ & $(1,7)$ & $(2,8)$ & $(3,8)$ & $(4,8)$ & $(4,9)$ & $\cdots$
\\[-0.4ex]&
\multicolumn{7}{l}{
\begin{tikzpicture}[baseline]
  \node at (-0.1,0) {};
  \draw[->] (0,0) -- (6.9,0);
  \draw (0,0.1) -- (0,-0.1);
    
  \node[below] at (0,-0.05) {\small$0$};
  \node[below] at (6.85,-0.05) {\small time};
  \foreach \i in {0,...,4} {
    \pgfmathsetlengthmacro{\x}{1.17cm + 1.15cm*\i}
    \draw (\x,0.1) -- (\x,-0.1);
    \node[below] at (\x,-0.1) {$e_{\i}$};
  }
\end{tikzpicture}
}
\end{tabular}
\caption{Example signals (and delayed signals) produced by signal
 combinators when receiving event sequence $\tau =
 e_0,e_1,e_2,e_3,e_4$. For \ensuremath{\Varid{mkSig}\;(\wait\;\kappa)}, we assume that $\kappa :
 \ensuremath{\Chan\;\Nat} \in \Delta$, $e_1 = \kappa \mapsto 7$, and $e_2 = \kappa
 \mapsto 8$, whereas the other events do not concern $\kappa$.}
\label{fig:example-sigs}
\end{figure}

Similarly to \citet{bahr23AsynchronousModalFRP}, \Rizzo features two later
modalities to capture asynchronicity: an \emph{existential} later
modality \ensuremath{\DelayE} and a \emph{universal} later modality \ensuremath{\DelayA}. Intuitively
speaking, a value of type \ensuremath{\DelayE\;\Conid{A}} is a pair $\ensuremath{(\theta,\Varid{f})}$ consisting of
a \emph{clock} \ensuremath{\theta} and a delayed computation \ensuremath{\Varid{f}} that will
produce a value of type \ensuremath{\Conid{A}} as soon as \ensuremath{\theta} \emph{ticks}. Given a
value \ensuremath{\Varid{v}\mathbin{:}\DelayE\;\Conid{A}}, we write \ensuremath{\cl{\Varid{v}}} to refer to the clock of \ensuremath{\Varid{v}}. By
contrast, a value of type \ensuremath{\DelayA\;\Conid{A}} is a delayed computation that will
produce a value of type \ensuremath{\Conid{A}} whenever \emph{any} clock ticks. 

The universal later modality provides the interface of an applicative
functor~\citep{mcbride08ApplicativeProgrammingEffects} made up of the
introduction form $\ensuremath{\delay\mathbin{:}\Conid{A}\to \DelayA\;\Conid{A}}$ and the applicative action
operator $\ensuremath{\appA\mathbin{:}\DelayA\;(\Conid{A}\to \Conid{B})\to \DelayA\;\Conid{A}\to \DelayA\;\Conid{B}}$. That is, we can delay
any value into the future and apply a delayed function to a delayed
argument to obtain its delayed result. In addition, \ensuremath{\DelayA} also features
the operator \ensuremath{\appAE\mathbin{:}\DelayA\;(\Conid{A}\to \Conid{B})\to \DelayE\;\Conid{A}\to \DelayE\;\Conid{B}}, which is a variant of
\ensuremath{\appA} that interacts with the existential later modality. This
interaction between the two later modalities is possible since a
universally delayed function $\ensuremath{\Varid{f}\mathbin{:}\DelayA\;(\Conid{A}\to \Conid{B})}$ is available at any
time in the future. In particular, it is available when the clock
\ensuremath{\cl{\Varid{v}}} of an existentially delayed value \ensuremath{\Varid{v}\mathbin{:}\DelayE\;\Conid{A}} ticks so that \ensuremath{\Varid{v}}
produces a value of type \ensuremath{\Conid{A}}. Therefore, \ensuremath{\Varid{f}\appAE\Varid{v}} produces a value
whenever \ensuremath{\Varid{v}} does, i.e.\ \ensuremath{\cl{\Varid{f}\appAE\Varid{v}}\mathrel{=}\cl{\Varid{v}}}.

With the help of these ingredients, we can define a functorial action
operator \ensuremath{\appIE}, which we will use extensively in examples. It applies a function to
an existentially delayed computation:
\begin{hscode}\SaveRestoreHook
\column{B}{@{}>{\hspre}l<{\hspost}@{}}%
\column{E}{@{}>{\hspre}l<{\hspost}@{}}%
\>[B]{}\anonymous \appIE\anonymous \mathbin{:}(\Conid{A}\to \Conid{B})\to \DelayE\;\Conid{A}\to \DelayE\;\Conid{B}{}\<[E]%
\\
\>[B]{}\Varid{f}\appIE\Varid{x}\mathrel{=}\delay\;\Varid{f}\appAE\Varid{x}{}\<[E]%
\ColumnHook
\end{hscode}\resethooks

The type of signals producing values of type \ensuremath{\Conid{A}} is defined by guarded
recursion as $\ensuremath{\Sig\;\Conid{A}} \cong \ensuremath{\Conid{A}\times\DelayE\;(\Sig\;\Conid{A})}$. Concretely, signals are
constructed using the \ensuremath{\mathbin{::}_{\Conid{A}}} operator of type \ensuremath{\Conid{A}\to \DelayE\;(\Sig\;\Conid{A})\to \Sig\;\Conid{A}} so that the signal \ensuremath{\Varid{v}\mathbin{::}_{\Conid{A}}\Varid{w}} has the current value \ensuremath{\Varid{v}\mathbin{:}\Conid{A}} and will
update as soon as \ensuremath{\cl{\Varid{w}}} ticks. We often elide the subscript \ensuremath{\Conid{A}} from
\ensuremath{\mathbin{::}_{\Conid{A}}} and simply write \ensuremath{\Varid{v}\mathbin{::}\Varid{w}} when \ensuremath{\Conid{A}} can be inferred from the
context. Since we can think of a value \ensuremath{\Varid{w}} of type \ensuremath{\DelayE\;(\Sig\;\Conid{A})} as a
pair \ensuremath{(\theta,\Varid{f})} consisting of a clock \ensuremath{\theta} and a delayed
computation \ensuremath{\Varid{f}} that produces a new value of type \ensuremath{\Sig\;\Conid{A}} whenever
\ensuremath{\theta} ticks, we can in turn think of a signal $\ensuremath{\Varid{s}\mathbin{:}\Sig\;\Conid{A}}$ as a
sequence $v_0 \T{\theta_0} v_1 \T{\theta_1} \ldots$ consisting of
lazily computed values $v_i$ of type \ensuremath{\Conid{A}} separated by clocks $\theta_i$. After the
clocks $\theta_0,\dots, \theta_{i-1}$ have ticked in succession, the
signal has the value $v_i$ and waits for a tick on $\theta_i$. The
current value and the future of a signal can be accessed via $\ensuremath{\head\mathbin{:}\Sig\;\Conid{A}\to \Conid{A}}$ and $\ensuremath{\tail\mathbin{:}\Sig\;\Conid{A}\to \DelayE\;(\Sig\;\Conid{A})}$, respectively.
However, we will usually use pattern matching syntax instead. For
example, the \ensuremath{\Varid{map}} function can be implemented as follows:
\begin{hscode}\SaveRestoreHook
\column{B}{@{}>{\hspre}l<{\hspost}@{}}%
\column{E}{@{}>{\hspre}l<{\hspost}@{}}%
\>[B]{}\Varid{map}\mathbin{:}(\Conid{A}\to \Conid{B})\to \Sig\;\Conid{A}\to \Sig\;\Conid{B}{}\<[E]%
\\
\>[B]{}\Varid{map}\;\Varid{f}\;(\Varid{x}\mathbin{::}\Varid{xs})\mathrel{=}\Varid{f}\;\Varid{x}\mathbin{::}(\Varid{map}\;\Varid{f}\appIE\Varid{xs}){}\<[E]%
\ColumnHook
\end{hscode}\resethooks
This is a \emph{guarded recursive definition}, as the recursive call
\ensuremath{\Varid{map}\;\Varid{f}} is delayed by the \ensuremath{\appIE} operator. We show how the syntactic
sugar of such guarded recursive function definitions is translated
into the guarded fixed point combinator \ensuremath{\fix} in
section~\ref{sec:recursion}.

Fig.~\ref{fig:example-sigs} illustrates how signal combinators such as
\ensuremath{\Varid{map}} work by showing the asynchronous timing of different (delayed)
signals. Each row represents a signal (of type \ensuremath{\Sig\;\Conid{A}}) or a delayed
signal (of type \ensuremath{\DelayE\;(\Sig\;\Conid{A})} and thus with no initial value at time
$0$). An entry $v$ in a row indicates that the (delayed) signal
updates its value to $v$ at that time. To illustrate these (delayed)
signals, we assume that we have received a sequence of events $\tau =
e_0,e_1,e_2,e_3,e_4$, whose occurrences are shown at the bottom of
Fig.~\ref{fig:example-sigs}. The figure illustrates the relative
timing of the two signals $s_1 : \ensuremath{\Sig\;\Nat}$ and $s_2 : \ensuremath{\Sig\;\Nat}$,
which intuitively denote the following sequences of values and
clocks:%
\footnote{More accurately, the sequence $v_0 \T{\theta_0} v_1
 \T{\theta_1} \ldots$ denoted by a signal \ensuremath{\Varid{s}} is w.r.t.\ some given
 sequence of events $(e_i)_{i\in\nats}$. Specifically, each $v_i$ and
 $\theta_i$ are obtained from $s \reacts{e_0,\dots, e_n} v_i :: w_i$
 with $\theta_i = \cl{w_i}$ for some prefix $e_0,\dots, e_n$ of
 $(e_i)_{i\in\nats}$.}
\[
s_1 \approx 1 \T{\theta_0} 2 \T{\theta_1} 3 \T{\theta_2} 4 \T{\theta_3}
\ldots\quad \text{ and } \quad s_2 \approx 6 \T{\theta'_0} 7 \T{\theta'_1} 8 \T{\theta'_2} 9 \T{\theta'_3}
\ldots
\]
Fig.~\ref{fig:example-sigs} tells us that when event $e_0$ occurs,
$s_2$ produces a new value $7$, whereas $s_1$ does not. That is, $e_0$
causes $\theta'_0$ to tick, but not $\theta_0$. Similarly, $e_1$
causes both $\theta_0$ and $\theta'_1$ to tick simultaneously, which
in turn causes $s_1$ and $s_2$ to produce new values $2$ and $8$,
respectively. 
Given a term $t$ and a sequence of events $\tau$, we write $t
\reacts\tau v$ for the evaluation of $t$ to value $v$ in response to
$\tau$. For example, given the sequence of events $e_0,e_1,e_2$ from
Fig.~\ref{fig:example-sigs}, we have that $s_1 \reacts{e_0,e_1,e_2} 3
:: d_1$ and $s_2 \reacts{e_0,e_1,e_2} 8 :: d_2$, for some $d_1, d_2 :
\ensuremath{\DelayE\;(\Sig\;\Nat)}$ with $\cl{d_1} = \theta_2$ and $\cl{d_2} = \theta'_2$.
We call $\reacts\tau$ the \emph{reactive evaluation semantics} of
\Rizzo and define it formally in
section~\ref{sec:operational-semantics}.

Applying \ensuremath{\Varid{map}\;(\lambda \Varid{x}.\Varid{x}\mathbin{+}\mathrm{1})} to $s_1$ produces a signal intuitively
denoted as follows:
\[
\ensuremath{\Varid{map}\;(\lambda \Varid{x}.\Varid{x}\mathbin{+}\mathrm{1})\;\Varid{s}_{\mathrm{1}}} \approx 2 \T{\theta_0} 3 \T{\theta_1} 4 \T{\theta_2} 5 \T{\theta_3}
\ldots
\]
Because \ensuremath{\cl{\Varid{map}\;\Varid{f}\appIE\Varid{xs}}\mathrel{=}\cl{\Varid{xs}}}, \ensuremath{\Varid{map}} produces a signal with the
same clocks as the argument signal. For example, $\ensuremath{\Varid{map}\;(\lambda \Varid{x}.\Varid{x}\mathbin{+}\mathrm{1})\;\Varid{s}_{\mathrm{1}}} \reacts{e_0,e_1,e_2} 4 :: d$ for some $d : \ensuremath{\DelayE\;(\Sig\;\Nat)}$
with $\cl{d} = \theta_2$.

Using the \ensuremath{\Varid{map}} combinator, we can implement a sampling combinator that
takes two signals and samples the value of the second signal each time
the first signal updates:
\begin{hscode}\SaveRestoreHook
\column{B}{@{}>{\hspre}l<{\hspost}@{}}%
\column{E}{@{}>{\hspre}l<{\hspost}@{}}%
\>[B]{}\Varid{sample}\mathbin{:}\Sig\;\Conid{A}\to \Sig\;\Conid{B}\to \Sig\;(\Conid{A}\times\Conid{B}){}\<[E]%
\\
\>[B]{}\Varid{sample}\;\Varid{xs}\;\Varid{ys}\mathrel{=}\Varid{map}\;(\lambda \Varid{x}.(\Varid{x},\head\;\Varid{ys}))\;\Varid{xs}{}\<[E]%
\ColumnHook
\end{hscode}\resethooks
Each time $xs$ produces a new value $v$, the function \ensuremath{(\lambda \Varid{x}.(\Varid{x},\head\;\Varid{ys}))} is applied to $v$ so that the resulting signal produces a value
obtained by evaluating \ensuremath{(\Varid{v},\head\;\Varid{ys})}. Let's consider the example
\ensuremath{\Varid{sample}\;\Varid{s}_{\mathrm{1}}\;\Varid{s}_{\mathrm{2}}}, shown in Fig.~\ref{fig:example-sigs}. Recall that
$e_2$ causes the clock $\theta_1$ of \ensuremath{\Varid{s}_{\mathrm{1}}} to tick, and that $\ensuremath{\Varid{s}_{\mathrm{1}}}
\reacts{e_0,e_1,e_2} 3 :: d_1$ and $s_2 \reacts{e_0,e_1,e_2} 8 ::
d_2$. Hence, when $e_2$ occurs, \ensuremath{\Varid{sample}\;\Varid{s}_{\mathrm{1}}\;\Varid{s}_{\mathrm{2}}} will sample \ensuremath{\Varid{s}_{\mathrm{2}}}, which
has been evaluated to $8 :: d_2$. Therefore, $\ensuremath{\Varid{sample}\;\Varid{s}_{\mathrm{1}}\;\Varid{s}_{\mathrm{2}}} \reacts{e_0,e_1,e_2}
(3,8) :: d'$ for some $d'$ with $\cl{d'} = \theta_2$.

\subsection{Constructing Delayed Computations}
\label{sec:delayed-computations}

In addition to \ensuremath{\delay}, \ensuremath{\appA}, and \ensuremath{\appAE}, \Rizzo features four more
primitives to construct delayed computations, namely \ensuremath{\never}, \ensuremath{\select},
\ensuremath{\wait}, and \ensuremath{\watch}. We discuss each of them in turn below and
illustrate them with signal combinators and accompanying
examples in Fig.~\ref{fig:example-sigs}.

The primitive \ensuremath{\never\mathbin{:}\DelayE\;\Conid{A}} allows us to construct delayed
computations with a clock that will never tick. This is useful for
constructing constant signals:
\begin{hscode}\SaveRestoreHook
\column{B}{@{}>{\hspre}l<{\hspost}@{}}%
\column{E}{@{}>{\hspre}l<{\hspost}@{}}%
\>[B]{}\Varid{const}\mathbin{:}\Conid{A}\to \Sig\;\Conid{A}{}\<[E]%
\\
\>[B]{}\Varid{const}\;\Varid{x}\mathrel{=}\Varid{x}\mathbin{::}\never{}\<[E]%
\ColumnHook
\end{hscode}\resethooks
That is, intuitively speaking, \ensuremath{\Varid{const}\;\Varid{v}} denotes a finite sequence $v
\T{\theta}$ with a clock $\theta$ that never ticks.

The primitive \ensuremath{\select} takes two delayed computations \ensuremath{\Varid{v}\mathbin{:}\DelayE\;\Conid{A},\Varid{w}\mathbin{:}\DelayE\;\Conid{B}} and produces a delayed computation of type \ensuremath{\DelayE\;((\Conid{A}\mathbin{+}\Conid{B})\mathbin{+}(\Conid{A}\times\Conid{B}))} with a clock that ticks whenever \ensuremath{\cl{\Varid{v}}} or \ensuremath{\cl{\Varid{w}}} ticks. If
\ensuremath{\cl{\Varid{v}}} ticks first, then \ensuremath{\select\;\Varid{v}\;\Varid{w}} produces the value produced by
$v$. Likewise, if \ensuremath{\cl{\Varid{w}}} ticks first, then \ensuremath{\select\;\Varid{v}\;\Varid{w}} produces the
value produced by $w$. If both \ensuremath{\cl{\Varid{v}}} and \ensuremath{\cl{\Varid{w}}} tick
simultaneously, then \ensuremath{\select\;\Varid{v}\;\Varid{w}} produces a pair of values with each
component drawn from the corresponding result of $v$ and $w$. The sum
type \ensuremath{(\Conid{A}\mathbin{+}\Conid{B})\mathbin{+}(\Conid{A}\times\Conid{B})} witnesses these three contingencies.

To make the type of \ensuremath{\select} more readable, we use the shorthand \ensuremath{\Conid{Select}\;\Conid{A}\;\Conid{B}} for the type \ensuremath{(\Conid{A}\mathbin{+}\Conid{B})\mathbin{+}(\Conid{A}\times\Conid{B})} as well as the shorthands \ensuremath{\con{left}\;\Varid{s}}, \ensuremath{\con{right}\;\Varid{t}}, and \ensuremath{\con{both}\;\Varid{s}\;\Varid{t}} for \ensuremath{\interm_1\;(\interm_1\;\Varid{s})}, \ensuremath{\interm_1\;(\interm_2\;\Varid{t})}, and
\ensuremath{\interm_2\;(\Varid{s},\Varid{t})}, respectively. We can use \ensuremath{\select} to construct dynamic
dataflows. The prototypical signal combinator with dynamic dataflow is
\ensuremath{\Varid{switch}}, which constructs a signal that initially behaves like the first
signal it is given but then switches its behaviour to the second --
delayed -- signal as soon as that second signal arrives:

\begin{hscode}\SaveRestoreHook
\column{B}{@{}>{\hspre}l<{\hspost}@{}}%
\column{3}{@{}>{\hspre}l<{\hspost}@{}}%
\column{10}{@{}>{\hspre}l<{\hspost}@{}}%
\column{22}{@{}>{\hspre}l<{\hspost}@{}}%
\column{25}{@{}>{\hspre}l<{\hspost}@{}}%
\column{29}{@{}>{\hspre}c<{\hspost}@{}}%
\column{29E}{@{}l@{}}%
\column{32}{@{}>{\hspre}l<{\hspost}@{}}%
\column{E}{@{}>{\hspre}l<{\hspost}@{}}%
\>[B]{}\Varid{switch}\mathbin{:}\Sig\;\Conid{A}\to \DelayE\;(\Sig\;\Conid{A})\to \Sig\;\Conid{A}{}\<[E]%
\\
\>[B]{}\Varid{switch}\;(\Varid{x}\mathbin{::}\Varid{xs})\;\Varid{d}\mathrel{=}\Varid{x}\mathbin{::}(\Varid{cont}\appIE\select\;\Varid{xs}\;\Varid{d}){}\<[E]%
\\
\>[B]{}\hsindent{3}{}\<[3]%
\>[3]{}\mathbf{where}\;{}\<[10]%
\>[10]{}\Varid{cont}\mathbin{:}\Conid{Select}\;(\Sig\;\Conid{A})\;(\Sig\;\Conid{A})\to \Sig\;\Conid{A}{}\<[E]%
\\
\>[10]{}\Varid{cont}\;(\con{left}\;\Varid{xs'}{}\<[29]%
\>[29]{}){}\<[29E]%
\>[32]{}\mathrel{=}\Varid{switch}\;\Varid{xs'}\;\Varid{d}{}\<[E]%
\\
\>[10]{}\Varid{cont}\;(\con{right}\;{}\<[25]%
\>[25]{}\Varid{d'}{}\<[29]%
\>[29]{}){}\<[29E]%
\>[32]{}\mathrel{=}\Varid{d'}{}\<[E]%
\\
\>[10]{}\Varid{cont}\;(\con{both}\;{}\<[22]%
\>[22]{}\anonymous \;{}\<[25]%
\>[25]{}\Varid{d'}{}\<[29]%
\>[29]{}){}\<[29E]%
\>[32]{}\mathrel{=}\Varid{d'}{}\<[E]%
\ColumnHook
\end{hscode}\resethooks
As we can see, \ensuremath{\Varid{switch}} uses \ensuremath{\select} to check which of \ensuremath{\Varid{xs}} and \ensuremath{\Varid{d}} arrives
first. If \ensuremath{\Varid{xs}} arrives first and produces a new signal \ensuremath{\Varid{xs'}\mathbin{:}\Sig\;\Conid{A}},
then \ensuremath{\Varid{switch}} recursively continues. Otherwise, if \ensuremath{\Varid{d}} arrives either
before or simultaneously with \ensuremath{\Varid{xs}}, then \ensuremath{\Varid{switch}} simply returns the
resulting signal \ensuremath{\Varid{d'}\mathbin{:}\Sig\;\Conid{A}}.

To understand the \ensuremath{\wait} primitive, we must first consider what kind
of events a \Rizzo program can receive and how these events cause a
clock to tick. A \Rizzo program receives inputs from its environment
via \emph{channels}. A channel of type \ensuremath{\Chan\;\Conid{A}} can receive data of
type \ensuremath{\Conid{A}}. We write $\kappa \mapsto t$ for the \emph{event} that the
channel $\kappa$ has received input $t$, which is typically a value
but does not need to be one. For example, if the channel context
$\Delta$ contains a channel $\kappa_{\sym{keyboard}} : \ensuremath{\Chan\;\Conid{Char}}$
that receives the characters typed on the keyboard, then the event
$\kappa_{\sym{keyboard}} \mapsto \ensuremath{\text{\ttfamily 'a'}}$ indicates that the user has
pressed the `a' key. We can wait to receive a value from a channel
using $\ensuremath{\wait\mathbin{:}\Chan\;\Conid{A}\to \DelayE\;\Conid{A}}$, which takes a channel $\ensuremath{\kappa\mathbin{:}\Chan\;\Conid{A}}$ and constructs a delayed computation with a clock that
ticks whenever an event $\kappa \mapsto t$ occurs, which then causes
the delayed computation to produce the value $v$ that $t$ evaluates
to. For example, $\wait\,\kappa_{\sym{keyboard}} : \ensuremath{\DelayE\;\Conid{Char}}$ is a
delayed computation that produces a character value as soon as the
user presses a key. In turn, we can use delayed computations to
construct signals using the following combinator:
\begin{hscode}\SaveRestoreHook
\column{B}{@{}>{\hspre}l<{\hspost}@{}}%
\column{E}{@{}>{\hspre}l<{\hspost}@{}}%
\>[B]{}\Varid{mkSig}\mathbin{:}\DelayE\;\Conid{A}\to \DelayE\;(\Sig\;\Conid{A}){}\<[E]%
\\
\>[B]{}\Varid{mkSig}\;\Varid{da}\mathrel{=}(\lambda \Varid{a}.\Varid{a}\mathbin{::}\Varid{mkSig}\;\Varid{da})\appIE\Varid{da}{}\<[E]%
\ColumnHook
\end{hscode}\resethooks
Fig.~\ref{fig:example-sigs} illustrates the delayed signal \ensuremath{\Varid{mkSig}\;(\wait\;\kappa)} under the assumption that the events in $\tau$ are more specifically $e_1
= \kappa \mapsto 7$ and $e_2 = \kappa \mapsto 8$, while the events
$e_0, e_3, e_4$ do not concern $\kappa$, i.e.\ they are of the form
$\kappa' \mapsto t$ for some $\kappa' \neq \kappa$.

Channels are either drawn from the channel context $\Delta$ or can be
constructed by the program using $\chan : \ensuremath{\Chan\;\Conid{A}}$. A channel $\kappa
: \ensuremath{\Chan\;\Conid{A}} \in \Delta$ is a built-in channel like
$\kappa_{\sym{keyboard}} : \ensuremath{\Chan\;\Conid{Char}}$ or $\kappa_{\sym{mouse}} :
\ensuremath{\Chan\;(\Conid{Int}\times\Conid{Int})}$, whereas dynamically created channels can be used
to construct objects that may receive input. For example, in a GUI, we
may create a channel $\kappa : \ensuremath{\Chan\;\mathrm{1}}$ whenever we
construct a button, so that we may receive an event
$\kappa \mapsto ()$ every time the button is pressed.

Finally, the primitive $\ensuremath{\watch\mathbin{:}\Sig\;(\Conid{A}\mathbin{+}\mathrm{1})\to \DelayE\;\Conid{A}}$ allows us to
treat a signal of type \ensuremath{\Sig\;(\Conid{A}\mathbin{+}\mathrm{1})} as an \emph{internal} source of
events. We typically call a signal of this type a \emph{partial
signal}. For readability, we use the shorthand \ensuremath{\Conid{Maybe}\;\Conid{A}} for \ensuremath{\Conid{A}\mathbin{+}\mathrm{1}},
and we write \ensuremath{\con{just}\;\Varid{t}} and \ensuremath{\con{nothing}} for \ensuremath{\interm_1\;\Varid{t}} and \ensuremath{\interm_2\;()},
respectively. Given a partial signal \ensuremath{\Varid{s}\mathbin{:}\Sig\;(\Conid{Maybe}\;\Conid{A})}, \ensuremath{\watch\;\Varid{s}}
constructs a delayed computation that produces a value of type \ensuremath{\Conid{A}}
whenever \ensuremath{\Varid{s}} is updated to a value of the form \ensuremath{\con{just}\;\Varid{v}} for some \ensuremath{\Varid{v}\mathbin{:}\Conid{A}}. With the help of \ensuremath{\watch}, we can implement \ensuremath{\Varid{filter}}, which takes a
predicate $p$ and a delayed signal $s$ to construct a new delayed
signal that contains only those elements of $s$ that satisfy $p$:
\begin{hscode}\SaveRestoreHook
\column{B}{@{}>{\hspre}l<{\hspost}@{}}%
\column{E}{@{}>{\hspre}l<{\hspost}@{}}%
\>[B]{}\Varid{filter}\mathbin{:}(\Conid{A}\to \Conid{Bool})\to \DelayE\;(\Sig\;\Conid{A})\to \DelayE\;(\Sig\;\Conid{A}){}\<[E]%
\\
\>[B]{}\Varid{filter}\;\Varid{p}\;\Varid{s}\mathrel{=}\Varid{mkSig}\;(\watch\;(\con{nothing}\mathbin{::}(\Varid{mapMaybe}\;\Varid{p}\;\Varid{s}))){}\<[E]%
\\[\blanklineskip]%
\>[B]{}\Varid{mapMaybe}\mathbin{:}(\Conid{A}\to \Conid{Bool})\to \DelayE\;(\Sig\;\Conid{A})\to \DelayE\;(\Sig\;(\Conid{Maybe}\;\Conid{A})){}\<[E]%
\\
\>[B]{}\Varid{mapMaybe}\;\Varid{p}\;\Varid{d}\mathrel{=}\Varid{map}\;(\lambda \Varid{x}.\mathbf{if}\;\Varid{p}\;\Varid{x}\;\mathbf{then}\;\con{just}\;\Varid{x}\;\mathbf{else}\;\con{nothing})\appIE\Varid{d}{}\<[E]%
\ColumnHook
\end{hscode}\resethooks
We use \ensuremath{\Varid{mapMaybe}} to turn a delayed signal \ensuremath{\Varid{d}} into a delayed
\emph{partial} signal that contains only the elements of \ensuremath{\Varid{d}} that
satisfy the predicate \ensuremath{\Varid{p}}. Then \ensuremath{\Varid{filter}} uses \ensuremath{\watch} on the partial
signal to construct the desired (non-partial) signal.

Now that we have seen all primitive operations to construct delayed
computations, we can give a precise account of what clocks are: A
clock $\theta$ is a finite set containing channels \ensuremath{\kappa\mathbin{:}\Chan\;\Conid{A}}
and partial signals \ensuremath{\Varid{s}\mathbin{:}\Sig\;(\Conid{Maybe}\;\Conid{A})}. Given an event $\kappa
\mapsto t$, the clock $\theta$ ticks iff $\kappa \in \theta$ or there
is a partial signal $s \in \theta$ such that $\kappa \mapsto t$ causes
\ensuremath{\Varid{s}} to produce a new value of the form \ensuremath{\con{just}\;\Varid{w}}.
In particular, we define $\cl{\ensuremath{\never}} = \emptyset$, i.e.\
$\cl{\ensuremath{\never}}$ never ticks; $\cl{\ensuremath{\wait\;\kappa}} = \set{\kappa}$;
$\cl{\ensuremath{\watch\;\Varid{v}}} = \set{\ensuremath{\Varid{v}}}$; \ensuremath{\cl{\Varid{v}\appAE\Varid{w}}\mathrel{=}\cl{\Varid{w}}}; and $\cl{\ensuremath{\select\;\Varid{v}\;\Varid{w}}} = \cl{v} \cup
\cl{w}$, i.e.\ $\cl{\ensuremath{\select\;\Varid{v}\;\Varid{w}}}$ ticks whenever $\cl{v}$ or $\cl{w}$
ticks. In addition, we also extend the notion of clocks to non-value
terms $t$ by defining $\cl{t} = \cl{v}$, where $v$ is the result of
evaluating $t$ with the empty sequence of events, i.e.\ $t
\reacts{\cdot} v$.

\subsection{Recursive Types and Guarded Recursion}
\label{sec:recursion}

\Rizzo features recursive types of the form $\mu \alpha. F$, where
$\alpha$ may not appear in the domain of a function type nor in the
scope of \ensuremath{\Chan}, \ensuremath{\DelayE}, or \ensuremath{\DelayA}. This restriction can be seen in the
type formation rules in Fig.~\ref{fig:wftype}, which permit the types $F
\to G$, $\Chan\,F$, $\DelayE F$, and $\DelayA F$ only if there are no
free type variable occurrences in $F$. Supported recursive types
include standard inductive types such as the type of natural numbers
$\ensuremath{\Nat} = \mu \alpha. 1 + \alpha$, lists $\List\, A = \mu \alpha. 1 +
(A \times \alpha)$, binary trees $\mu \alpha. A + (\alpha \times A
\times \alpha)$, and infinitely branching trees $\mu \alpha. A + (A
\times (\Nat \to \alpha))$, but also nested inductive types such as the
type of rose trees $\mu \alpha. A \times \List\, \alpha$. Importantly,
we also support recursive types with signals, such as a type of binary
trees $\mu \alpha. A + \Sig\, (\alpha \times A \times \alpha)$, where
each inner node is a signal and can thus change over time. We will see
an example of this in section~\ref{sec:example}, where we use such
trees to represent graphical user interfaces that can change
dynamically in response to user input.

We can traverse such recursive data structures using the primitive
recursion combinator $\rec*$, but in examples we use standard pattern
matching and recursion syntax, which can be translated into
uses of $\rec*$. For instance, consider the following recursive definitions:
\columnratio{0.5}
\begin{paracol}{2}
  \begin{hscode}\SaveRestoreHook
\column{B}{@{}>{\hspre}l<{\hspost}@{}}%
\column{14}{@{}>{\hspre}c<{\hspost}@{}}%
\column{14E}{@{}l@{}}%
\column{17}{@{}>{\hspre}l<{\hspost}@{}}%
\column{E}{@{}>{\hspre}l<{\hspost}@{}}%
\>[B]{}\mathbf{data}\;\List\;\Conid{A}{}\<[14]%
\>[14]{}\mathrel{=}{}\<[14E]%
\>[17]{}\con{nil}{}\<[E]%
\\
\>[14]{}\mid {}\<[14E]%
\>[17]{}\con{cons}\;\Conid{A}\;(\List\;\Conid{A}){}\<[E]%
\ColumnHook
\end{hscode}\resethooks
\switchcolumn
\begin{hscode}\SaveRestoreHook
\column{B}{@{}>{\hspre}l<{\hspost}@{}}%
\column{21}{@{}>{\hspre}l<{\hspost}@{}}%
\column{E}{@{}>{\hspre}l<{\hspost}@{}}%
\>[B]{}\Varid{length}\mathbin{:}\List\;\Conid{A}\to \Nat{}\<[E]%
\\
\>[B]{}\Varid{length}\;\con{nil}{}\<[21]%
\>[21]{}\mathrel{=}\mathrm{0}{}\<[E]%
\\
\>[B]{}\Varid{length}\;(\con{cons}\;\Varid{x}\;\Varid{xs}){}\<[21]%
\>[21]{}\mathrel{=}\mathrm{1}\mathbin{+}\Varid{length}\;\Varid{xs}{}\<[E]%
\ColumnHook
\end{hscode}\resethooks
\end{paracol}
\noindent
This recursion and pattern matching syntax desugars into the following \Rizzo type and term:
\[
\ensuremath{\List\;\Conid{A}} = \mu\alpha. 1 + (A \times \alpha)\hspace{2cm}
\ensuremath{\Varid{length}} = \lambda l .
\rec{r}{\caseterm{r}{x}{0}{x}{1+\pi_2\,(\pi_2\,x)}}{l}
\]

In addition, \Rizzo features guarded
recursion~\citep{nakano00ModalityRecursion} via the guarded fixed
point combinator \ensuremath{\fix}. Instead of a general fixed point combinator of
type \ensuremath{(\Conid{A}\to \Conid{A})\to \Conid{A}}, the guarded fixed point combinator has type \ensuremath{(\DelayA\;\Conid{A}\to \Conid{A})\to \Conid{A}}. Following \citet{bahr23AsynchronousModalFRP}, we make
the argument have type \ensuremath{\DelayA\;\Conid{A}}, so that the guarded fixed point can
only be unfolded in response to a future event, which in turn allows
us to prove productivity of the language.

We have already seen several examples of functions defined using
guarded recursion, including \ensuremath{\Varid{map}}, $\ensuremath{\Varid{mkSig}}$, and \ensuremath{\Varid{switch}}. If all
recursive calls are guarded by an application of \ensuremath{\delay}, we can
translate a recursive definition from recursion syntax into a guarded
fixed point. For example, the recursive call in the definition of
\ensuremath{\Varid{map}} appears as the first argument to \ensuremath{\appIE} and is thus guarded by a
\ensuremath{\delay}. We can translate this definition into an explicit use of
\ensuremath{\fix} and also translate the pattern matching syntax into explicit
projections with \ensuremath{\head} and \ensuremath{\tail} so that we obtain a term in the
core calculus:
\begin{hscode}\SaveRestoreHook
\column{B}{@{}>{\hspre}l<{\hspost}@{}}%
\column{26}{@{}>{\hspre}l<{\hspost}@{}}%
\column{E}{@{}>{\hspre}l<{\hspost}@{}}%
\>[B]{}\Varid{map}\mathbin{:}(\Conid{A}\to \Conid{B})\to \Sig\;\Conid{A}\to \Sig\;\Conid{B}{}\<[E]%
\\
\>[B]{}\Varid{map}\mathrel{=}\fix\;\Varid{r}.\lambda \Varid{f}.\lambda \Varid{s}.{}\<[26]%
\>[26]{}\mathbf{let}\;\Varid{x}\mathrel{=}\head\;\Varid{s}\;\mathbf{in}\;\mathbf{let}\;\Varid{xs}\mathrel{=}\tail\;\Varid{s}\;\mathbf{in}\;\Varid{f}\;\Varid{x}\mathbin{::}(\delay\;(\lambda \Varid{r'}.\Varid{r'}\;\Varid{f})\appA\Varid{r}\appAE\Varid{xs}){}\<[E]%
\ColumnHook
\end{hscode}\resethooks
For readability, we still keep the standard syntactic sugar \ensuremath{\mathbf{let}\;\Varid{x}\mathrel{=}\Varid{s}\;\mathbf{in}\;\Varid{t}} as notation for \ensuremath{(\lambda \Varid{x}.\Varid{t})\;\Varid{s}}.

This translation from recursive surface syntax to explicit guarded
fixed points follows a general scheme. Given a recursively defined
function $f$ where all recursive calls occur under a $\delay$, we can
desugar a recursive definition of the form
$
f x_1 \dots x_n  = K[\delay\, t_1,\dots,\delay\, t_m]
$,
where $f$ does not occur in \ensuremath{\Conid{K}}, into the following form that uses
$\fix$:
\[
f = \fix\, r . \lambda x_1 \dots \lambda x_n. K[\delay (\lambda r' . t_1[r'/f]) \appA r,\dots,\delay (\lambda r' . t_m[r'/f]) \appA r]
\]
Moreover, pattern matching is desugared into corresponding eliminators
of the calculus in the standard way as we have seen in the example for
\ensuremath{\Varid{map}} above.

\subsection{Equational Reasoning Limitations}
\label{sec:equational-reasoning}

The forgetful semantics of signals means that $\beta$ and $\eta$
conversions are not equivalence preserving in general. For example,
$\ensuremath{(\lambda \Varid{x}.\delay\;\Varid{x})\;(\head\;\Varid{xs})}$ is semantically different from its
$\beta$-contractum $\ensuremath{\delay\;(\head\;\Varid{xs})}$, and $\ensuremath{\Varid{f}\;(\head\;\Varid{xs})}$ is
semantically different from its $\eta$-expansion $\ensuremath{\lambda \Varid{x}.\Varid{f}\;(\head\;\Varid{xs})\;\Varid{x}}$. The head of a signal may be different depending on whether it is
evaluated now or in the future, because evaluating a signal erases its
past to avoid space leaks by construction. As an example, consider a
$\beta$-expanded version of \ensuremath{\Varid{sample}}, defined by 
\ensuremath{\Varid{sample}_\beta\;\Varid{xs}\;\Varid{ys}\mathrel{=}\Varid{map}\;((\lambda \Varid{y}\;\Varid{x}.(\Varid{x},\Varid{y}))\;(\head\;\Varid{ys}))\;\Varid{xs}}.
Given \ensuremath{\Varid{s}_{\mathrm{1}},\Varid{s}_{\mathrm{2}}} from Fig.~\ref{fig:example-sigs}, \ensuremath{\Varid{sample}\;\Varid{s}_{\mathrm{1}}\;\Varid{s}_{\mathrm{2}}}
produces $(1,6) \T{\theta_0} (2,8) \T{\theta_1} (3,8) \T{\theta_2}
(4,8) \T{\theta_3} \ldots$, whereas \ensuremath{\Varid{sample}_\beta\;\Varid{s}_{\mathrm{1}}\;\Varid{s}_{\mathrm{2}}} produces
$(1,6) \T{\theta_0} (2,6) \T{\theta_1} (3,6) \T{\theta_2} (4,6)
\T{\theta_3} \ldots$, i.e.\ the latter does not sample \ensuremath{\Varid{s}_{\mathrm{2}}} but simply
returns the initial value of \ensuremath{\Varid{s}_{\mathrm{2}}}, namely $6$.

A similar situation also appears in Async
RaTT~\citep{bahr23AsynchronousModalFRP} (the only other asynchronous
modal FRP calculus with space leaks guarantees), where the term
$\sym{read}_\kappa : A$ reads the current value of a channel $\kappa$
and thus evaluates to different values depending on \emph{when} it is
evaluated. Therefore, similarly to Async RaTT, equational reasoning in
\Rizzo has to restrict $\beta$- and $\eta$-conversion to values, i.e.\
$(\lambda x. t)\,v = t[v/x]$ and $v = \lambda x . v\,x$.

\section{Programming in \Rizzo}
\label{sec:example}

In this section, we present a collection of example \Rizzo programs to
illustrate the expressiveness of the language and to compare it with
previous work. We start by implementing a small library of signal
combinators and then use some of these combinators to implement a
simple GUI library along with a small GUI application.

\subsection{FRP Library}
\label{sec:frp-library}
\begin{figure}
  \small
  \scalebox{0.95}{
\begin{minipage}{0.40\textwidth}
  \begin{hscode}\SaveRestoreHook
\column{B}{@{}>{\hspre}l<{\hspost}@{}}%
\column{9}{@{}>{\hspre}l<{\hspost}@{}}%
\column{E}{@{}>{\hspre}l<{\hspost}@{}}%
\>[B]{}\Varid{map}{}\<[9]%
\>[9]{}\mathbin{:}(\Conid{A}\to \Conid{B})\to \Sig\;\Conid{A}\to \Sig\;\Conid{B}{}\<[E]%
\\
\>[B]{}\Varid{mkSig}{}\<[9]%
\>[9]{}\mathbin{:}\DelayE\;\Conid{A}\to \DelayE\;(\Sig\;\Conid{A}){}\<[E]%
\\
\>[B]{}\Varid{const}{}\<[9]%
\>[9]{}\mathbin{:}\Conid{A}\to \Sig\;\Conid{A}{}\<[E]%
\\
\>[B]{}\Varid{sample}{}\<[9]%
\>[9]{}\mathbin{:}\Sig\;\Conid{A}\to \Sig\;\Conid{B}\to \Sig\;(\Conid{A}\times\Conid{B}){}\<[E]%
\\
\>[B]{}\Varid{zip}{}\<[9]%
\>[9]{}\mathbin{:}\Sig\;\Conid{A}\to \Sig\;\Conid{B}\to \Sig\;(\Conid{A}\times\Conid{B}){}\<[E]%
\ColumnHook
\end{hscode}\resethooks
\end{minipage}
\begin{minipage}{0.7\textwidth}
  \begin{hscode}\SaveRestoreHook
\column{B}{@{}>{\hspre}l<{\hspost}@{}}%
\column{13}{@{}>{\hspre}c<{\hspost}@{}}%
\column{13E}{@{}l@{}}%
\column{16}{@{}>{\hspre}l<{\hspost}@{}}%
\column{E}{@{}>{\hspre}l<{\hspost}@{}}%
\>[B]{}\Varid{switch}{}\<[13]%
\>[13]{}\mathbin{:}{}\<[13E]%
\>[16]{}\Sig\;\Conid{A}\to \DelayE\;(\Sig\;\Conid{A})\to \Sig\;\Conid{A}{}\<[E]%
\\
\>[B]{}\Varid{interleave}{}\<[13]%
\>[13]{}\mathbin{:}{}\<[13E]%
\>[16]{}(\Conid{A}\to \Conid{A}\to \Conid{A})\to \DelayE\;(\Sig\;\Conid{A})\to \DelayE\;(\Sig\;\Conid{A})\to \DelayE\;(\Sig\;\Conid{A}){}\<[E]%
\\
\>[B]{}\Varid{scan}{}\<[13]%
\>[13]{}\mathbin{:}{}\<[13E]%
\>[16]{}(\Conid{B}\to \Conid{A}\to \Conid{B})\to \Conid{B}\to \Sig\;\Conid{A}\to \Sig\;\Conid{B}{}\<[E]%
\\
\>[B]{}\Varid{scanAwait}{}\<[13]%
\>[13]{}\mathbin{:}{}\<[13E]%
\>[16]{}(\Conid{B}\to \Conid{A}\to \Conid{B})\to \Conid{B}\to \DelayE\;(\Sig\;\Conid{A})\to \Sig\;\Conid{B}{}\<[E]%
\\
\>[B]{}\Varid{filter}{}\<[13]%
\>[13]{}\mathbin{:}{}\<[13E]%
\>[16]{}(\Conid{A}\to \Conid{Bool})\to \DelayE\;(\Sig\;\Conid{A})\to \DelayE\;(\Sig\;\Conid{A}){}\<[E]%
\ColumnHook
\end{hscode}\resethooks
\end{minipage}}
  \vspace{-0.9em}
  \caption{Small FRP library for signal processing.}
  \label{fig:library}
\end{figure}

Fig.~\ref{fig:library} lists the type signatures of a small signal
combinator library, which we have already partly implemented in
section~\ref{sec:rizzo}. We discuss the implementation of the remaining
combinators below.

First, consider the \ensuremath{\Varid{scan}} combinator, which is a variant of the \ensuremath{\Varid{map}} combinator
where the produced output may depend on the previous value of the output
signal:
\begin{hscode}\SaveRestoreHook
\column{B}{@{}>{\hspre}l<{\hspost}@{}}%
\column{E}{@{}>{\hspre}l<{\hspost}@{}}%
\>[B]{}\Varid{scan}\mathbin{:}(\Conid{B}\to \Conid{A}\to \Conid{B})\to \Conid{B}\to \Sig\;\Conid{A}\to \Sig\;\Conid{B}{}\<[E]%
\\
\>[B]{}\Varid{scan}\;\Varid{f}\;\Varid{b}\;(\Varid{a}\mathbin{::}\Varid{as})\mathrel{=}\mathbf{let}\;\Varid{b'}\mathrel{=}\Varid{f}\;\Varid{b}\;\Varid{a}\;\mathbf{in}\;\Varid{b'}\mathbin{::}(\Varid{scan}\;\Varid{f}\;\Varid{b'}\appIE\Varid{as}){}\<[E]%
\ColumnHook
\end{hscode}\resethooks
The example in Fig.~\ref{fig:example-sigs} uses \ensuremath{\Varid{scan}} to sum up the
numbers produced by signal \ensuremath{\Varid{s}_{\mathrm{1}}}.

We can derive a variant of \ensuremath{\Varid{scan}} that takes a delayed signal as
argument:
\begin{hscode}\SaveRestoreHook
\column{B}{@{}>{\hspre}l<{\hspost}@{}}%
\column{E}{@{}>{\hspre}l<{\hspost}@{}}%
\>[B]{}\Varid{scanAwait}\mathbin{:}(\Conid{B}\to \Conid{A}\to \Conid{B})\to \Conid{B}\to \DelayE\;(\Sig\;\Conid{A})\to \Sig\;\Conid{B}{}\<[E]%
\\
\>[B]{}\Varid{scanAwait}\;\Varid{f}\;\Varid{b}\;\Varid{s}\mathrel{=}\Varid{b}\mathbin{::}(\Varid{scan}\;\Varid{f}\;\Varid{b}\appIE\Varid{s}){}\<[E]%
\ColumnHook
\end{hscode}\resethooks

The \ensuremath{\Varid{zip}} combinator has the same type as \ensuremath{\Varid{sample}}, but unlike the
latter, the signal produced by $\ensuremath{\Varid{zip}\;\Varid{as}\;\Varid{bs}}$ updates whenever \ensuremath{\Varid{as}} or
$\ensuremath{\Varid{bs}}$ updates, i.e.\ $\cl{\ensuremath{\tail\;(\Varid{zip}\;\Varid{as}\;\Varid{bs})}} = \cl{\ensuremath{\tail\;\Varid{as}}} \cup
\cl{\ensuremath{\tail\;\Varid{bs}}}$: 
\begin{hscode}\SaveRestoreHook
\column{B}{@{}>{\hspre}l<{\hspost}@{}}%
\column{3}{@{}>{\hspre}l<{\hspost}@{}}%
\column{10}{@{}>{\hspre}l<{\hspost}@{}}%
\column{31}{@{}>{\hspre}l<{\hspost}@{}}%
\column{E}{@{}>{\hspre}l<{\hspost}@{}}%
\>[B]{}\Varid{zip}\mathbin{:}\Sig\;\Conid{A}\to \Sig\;\Conid{B}\to \Sig\;(\Conid{A}\times\Conid{B}){}\<[E]%
\\
\>[B]{}\Varid{zip}\;\Varid{as}\;\Varid{bs}\mathrel{=}(\head\;\Varid{as},\head\;\Varid{bs})\mathbin{::}(\Varid{cont}\appIE\select\;(\tail\;\Varid{as})\;(\tail\;\Varid{bs})){}\<[E]%
\\
\>[B]{}\hsindent{3}{}\<[3]%
\>[3]{}\mathbf{where}\;{}\<[10]%
\>[10]{}\Varid{cont}\;(\con{left}\;\Varid{as'}){}\<[31]%
\>[31]{}\mathrel{=}\Varid{zip}\;\Varid{as'}\;\Varid{bs}{}\<[E]%
\\
\>[10]{}\Varid{cont}\;(\con{right}\;\Varid{bs'}){}\<[31]%
\>[31]{}\mathrel{=}\Varid{zip}\;\Varid{as}\;\Varid{bs'}{}\<[E]%
\\
\>[10]{}\Varid{cont}\;(\con{both}\;\Varid{as'}\;\Varid{bs'}){}\<[31]%
\>[31]{}\mathrel{=}\Varid{zip}\;\Varid{as'}\;\Varid{bs'}{}\<[E]%
\ColumnHook
\end{hscode}\resethooks
That is, at any given time, the signal \ensuremath{\Varid{zip}\;\Varid{as}\;\Varid{bs}} has the current value
$(a,b)$, where $a$ is the current value of \ensuremath{\Varid{as}} and $b$ is the current
value of $\ensuremath{\Varid{bs}}$.

If we have two \emph{delayed} signals of the same type, we can also
combine them by interleaving, so that each update on the resulting
delayed signal is an update from one of the two original delayed
signals. This idea is implemented by the \ensuremath{\Varid{interleave}} combinator,
which also takes a function that serves as a tiebreaker for when both
delayed signals update at the same time:
\begin{hscode}\SaveRestoreHook
\column{B}{@{}>{\hspre}l<{\hspost}@{}}%
\column{3}{@{}>{\hspre}l<{\hspost}@{}}%
\column{10}{@{}>{\hspre}l<{\hspost}@{}}%
\column{23}{@{}>{\hspre}l<{\hspost}@{}}%
\column{35}{@{}>{\hspre}l<{\hspost}@{}}%
\column{48}{@{}>{\hspre}l<{\hspost}@{}}%
\column{57}{@{}>{\hspre}l<{\hspost}@{}}%
\column{78}{@{}>{\hspre}l<{\hspost}@{}}%
\column{E}{@{}>{\hspre}l<{\hspost}@{}}%
\>[B]{}\Varid{interleave}\mathbin{:}(\Conid{A}\to \Conid{A}\to \Conid{A})\to \DelayE\;(\Sig\;\Conid{A})\to \DelayE\;(\Sig\;\Conid{A})\to \DelayE\;(\Sig\;\Conid{A}){}\<[E]%
\\
\>[B]{}\Varid{interleave}\;\Varid{f}\;\Varid{xs}\;\Varid{ys}\mathrel{=}\Varid{cont}\appIE\select\;\Varid{xs}\;\Varid{ys}{}\<[E]%
\\
\>[B]{}\hsindent{3}{}\<[3]%
\>[3]{}\mathbf{where}\;{}\<[10]%
\>[10]{}\Varid{cont}\;(\con{left}\;{}\<[23]%
\>[23]{}(\Varid{x}\mathbin{::}\Varid{xs'})){}\<[48]%
\>[48]{}\mathrel{=}\Varid{x}{}\<[57]%
\>[57]{}\mathbin{::}\Varid{interleave}\;\Varid{f}\;\Varid{xs'}\;{}\<[78]%
\>[78]{}\Varid{ys}{}\<[E]%
\\
\>[10]{}\Varid{cont}\;(\con{right}\;{}\<[35]%
\>[35]{}(\Varid{y}\mathbin{::}\Varid{ys'})){}\<[48]%
\>[48]{}\mathrel{=}\Varid{y}{}\<[57]%
\>[57]{}\mathbin{::}\Varid{interleave}\;\Varid{f}\;\Varid{xs}\;{}\<[78]%
\>[78]{}\Varid{ys'}{}\<[E]%
\\
\>[10]{}\Varid{cont}\;(\con{both}\;{}\<[23]%
\>[23]{}(\Varid{x}\mathbin{::}\Varid{xs'})\;{}\<[35]%
\>[35]{}(\Varid{y}\mathbin{::}\Varid{ys'})){}\<[48]%
\>[48]{}\mathrel{=}\Varid{f}\;\Varid{x}\;\Varid{y}{}\<[57]%
\>[57]{}\mathbin{::}\Varid{interleave}\;\Varid{f}\;\Varid{xs'}\;{}\<[78]%
\>[78]{}\Varid{ys'}{}\<[E]%
\ColumnHook
\end{hscode}\resethooks

\subsection{GUI Application}
\label{sec:gui-app}

\begin{figure}
  \small
\begin{hscode}\SaveRestoreHook
\column{B}{@{}>{\hspre}l<{\hspost}@{}}%
\column{17}{@{}>{\hspre}l<{\hspost}@{}}%
\column{E}{@{}>{\hspre}l<{\hspost}@{}}%
\>[B]{}\mathbf{data}\;\Conid{Button}{}\<[17]%
\>[17]{}\mathrel{=}\con{mkButton}\;(\Sig\;\Conid{String})\;(\Sig\;\Conid{Colour})\;(\Chan\;\mathrm{1}){}\<[E]%
\\
\>[B]{}\mathbf{data}\;\Conid{TextField}{}\<[17]%
\>[17]{}\mathrel{=}\con{mkTextField}\;(\Sig\;\Conid{String})\;(\Sig\;\Conid{Colour})\;(\Chan\;\Conid{String}){}\<[E]%
\\
\>[B]{}\mathbf{data}\;\Conid{Layout}{}\<[17]%
\>[17]{}\mathrel{=}\con{horizontal}\mid \con{vertical}{}\<[E]%
\\
\>[B]{}\mathbf{data}\;\Conid{Widget}{}\<[17]%
\>[17]{}\mathrel{=}\con{button}\;\Conid{Button}\mid \con{textfield}\;\Conid{TextField}\mid \con{stack}\;\Conid{Layout}\;(\Sig\;(\List\;\Conid{Widget})){}\<[E]%
\\[\blanklineskip]%
\>[B]{}\Varid{simpleButton}\mathbin{:}\Conid{String}\to \Conid{Button}{}\<[E]%
\\
\>[B]{}\Varid{simpleButton}\;\Varid{txt}\mathrel{=}\con{mkButton}\;(\Varid{const}\;\Varid{txt})\;(\Varid{const}\;\Conid{Black})\;\chan[]{}\<[E]%
\\[\blanklineskip]%
\>[B]{}\Varid{onClick}\mathbin{:}\Conid{Button}\to \DelayE\;(\Sig\;\mathrm{1}){}\<[E]%
\\
\>[B]{}\Varid{onClick}\;(\con{mkButton}\;\anonymous \;\anonymous \;\Varid{k})\mathrel{=}\Varid{mkSig}\;(\wait\;\Varid{k}){}\<[E]%
\\[\blanklineskip]%
\>[B]{}\Varid{simpleTF}\mathbin{:}\Conid{String}\to \Conid{TextField}{}\<[E]%
\\
\>[B]{}\Varid{simpleTF}\;\Varid{txt}\mathrel{=}\mathbf{let}\;\Varid{k}\mathrel{=}\chan[]\;\mathbf{in}\;\con{mkTextField}\;(\Varid{txt}\mathbin{::}\Varid{mkSig}\;(\wait\;\Varid{k}))\;(\Varid{const}\;\Conid{Black})\;\Varid{k}{}\<[E]%
\ColumnHook
\end{hscode}\resethooks
\vspace{-1.5em}
\caption{Simple GUI library.}
\label{fig:gui-lib}
\end{figure}

Fig.~\ref{fig:gui-lib} implements the interface for a very simple GUI
library. It defines a \ensuremath{\Conid{Widget}} type to represent the hierarchical
structure of a GUI consisting of two basic widget types -- buttons and
text fields -- and stack widgets that
can place several widgets next to each other (horizontally or vertically).

The attributes of widgets, such as their colour, are described using
signals, so that these may change over time. In addition, widgets may
produce events, e.g.\ when a button is pressed or the user types text
into a text field. Such events and their accompanying data are sent on
channels: A button has a channel of type \ensuremath{\Chan\;\mathrm{1}}, so
that when the button is pressed, the button's channel receives a unit
value; and the text field has a channel of type \ensuremath{\Chan\;\Conid{String}}, so that
when the text field's contents change, we receive the new contents on
that channel.

For example, we may construct a button that changes its colour from
black to red when pressed:
\begin{hscode}\SaveRestoreHook
\column{B}{@{}>{\hspre}l<{\hspost}@{}}%
\column{8}{@{}>{\hspre}l<{\hspost}@{}}%
\column{17}{@{}>{\hspre}l<{\hspost}@{}}%
\column{E}{@{}>{\hspre}l<{\hspost}@{}}%
\>[B]{}\Varid{colourButton}\mathbin{:}\Conid{Button}{}\<[E]%
\\
\>[B]{}\Varid{colourButton}\mathrel{=}{}\<[17]%
\>[17]{}\mathbf{let}\;\Varid{c}\mathrel{=}\chan[]\;\mathbf{in}\;\con{mkButton}\;(\Varid{const}\;\text{\ttfamily \char34 click~me\char34})\;(\Conid{Black}\mathbin{::}(\Varid{const}\;\Conid{Red}\mathbin{@}\wait\;\Varid{c}))\;\Varid{c}{}\<[E]%
\\[\blanklineskip]%
\>[B]{}\anonymous \mathbin{@}\anonymous {}\<[8]%
\>[8]{}\mathbin{:}\Conid{B}\to \DelayE\;\Conid{A}\to \DelayE\;\Conid{B}{}\<[E]%
\\
\>[B]{}\Varid{v}\mathbin{@}\Varid{d}\mathrel{=}(\lambda \anonymous .\Varid{v})\appIE\Varid{d}{}\<[E]%
\ColumnHook
\end{hscode}\resethooks
We define the \ensuremath{\mathbin{@}} operator so that we can delay the constant signal \ensuremath{\Varid{const}\;\Conid{Red}}
until $\ensuremath{\wait\;\Varid{c}}$ arrives.

Fig.~\ref{fig:gui-lib} also defines some helper functions to construct
simple widgets: a button with a constant text and a
fixed colour; and a text field that displays exactly the
text that the user has typed. Let's use these simple building blocks to
implement a small GUI with an `Add' button, which allows the user to add
an extra text field to the GUI with each press of the button:
\columnratio{0.4}
\begin{paracol}{2}
\begin{hscode}\SaveRestoreHook
\column{B}{@{}>{\hspre}l<{\hspost}@{}}%
\column{E}{@{}>{\hspre}l<{\hspost}@{}}%
\>[B]{}\Varid{btn}\mathbin{:}\Conid{Button}{}\<[E]%
\\
\>[B]{}\Varid{btn}\mathrel{=}\Varid{simpleButton}\;\text{\ttfamily \char34 Add\char34}{}\<[E]%
\\[\blanklineskip]%
\>[B]{}\Varid{newField}\mathbin{:}\mathrm{1}\to \Conid{Widget}{}\<[E]%
\\
\>[B]{}\Varid{newField}\;\anonymous \mathrel{=}\con{textfield}\;(\Varid{simpleTF}\;\text{\ttfamily \char34 \char34}){}\<[E]%
\ColumnHook
\end{hscode}\resethooks
\switchcolumn
\begin{hscode}\SaveRestoreHook
\column{B}{@{}>{\hspre}l<{\hspost}@{}}%
\column{E}{@{}>{\hspre}l<{\hspost}@{}}%
\>[B]{}\Varid{newWidgets}\mathbin{:}\DelayE\;(\Sig\;\Conid{Widget}){}\<[E]%
\\
\>[B]{}\Varid{newWidgets}\mathrel{=}\Varid{map}\;\Varid{newField}\appIE\Varid{onClick}\;\Varid{btn}{}\<[E]%
\\[\blanklineskip]%
\>[B]{}\Varid{allWidgets}\mathbin{:}\Sig\;(\List\;\Conid{Widget}){}\<[E]%
\\
\>[B]{}\Varid{allWidgets}\mathrel{=}\Varid{scanAwait}\;\Varid{snoc}\;[\mskip1.5mu \con{button}\;\Varid{btn}\mskip1.5mu]\;\Varid{newWidgets}{}\<[E]%
\ColumnHook
\end{hscode}\resethooks
\end{paracol}
\vspace{-0.7em}
\begin{hscode}\SaveRestoreHook
\column{B}{@{}>{\hspre}l<{\hspost}@{}}%
\column{E}{@{}>{\hspre}l<{\hspost}@{}}%
\>[B]{}\Varid{gui}\mathbin{:}\Conid{Widget}{}\<[E]%
\\
\>[B]{}\Varid{gui}\mathrel{=}\con{stack}\;\con{vertical}\;\Varid{allWidgets}{}\<[E]%
\ColumnHook
\end{hscode}\resethooks
This example illustrates the highly dynamic structure that is enabled
by the nesting of signals in the \ensuremath{\Conid{Widget}} type. The GUI's behaviour is
completely described by $\ensuremath{\Varid{gui}}$ of type \ensuremath{\Conid{Widget}}: It describes the
tree-structured hierarchy of widgets that are supposed to be presented
to the user, and how this hierarchy of widgets is supposed to change
in response to input from the user. The user's interaction with the
GUI produces values that are sent on the channels associated with individual
widgets. In particular, each time the button \ensuremath{\Varid{btn}} is pressed, an
event occurs on \ensuremath{\Varid{btn}}'s channel, which causes the delayed signal
\ensuremath{\Varid{newWidgets}\mathbin{:}\DelayE\;(\Sig\;\Conid{Widget})} to produce a new text field widget. In
turn, this delayed signal of new widgets is collected by \ensuremath{\Varid{allWidgets}}
into a signal of all widgets that have been created starting with the
button \ensuremath{\Varid{btn}}. To this end, we use \ensuremath{\Varid{scanAwait}} to add any new widget
from \ensuremath{\Varid{newWidgets}} to the initial list of widgets \ensuremath{[\mskip1.5mu \con{button}\;\Varid{btn}\mskip1.5mu]} via
\ensuremath{\Varid{snoc}\mathbin{:}\List\;\Conid{A}\to \Conid{A}\to \List\;\Conid{A}} that appends an element to the end of a
list.

To demonstrate the highly dynamic nature of widgets enabled by their
nested-signal structure, we revise the above example GUI to also allow
the removal of widgets. To this end, we change the definition of the
$\ensuremath{\Varid{newField}}$ function so that in addition to a text field, it also
adds a `Remove' button. In turn, the `Remove' button allows the user
to remove the text field and the `Remove' button itself.
\begin{hscode}\SaveRestoreHook
\column{B}{@{}>{\hspre}l<{\hspost}@{}}%
\column{15}{@{}>{\hspre}l<{\hspost}@{}}%
\column{19}{@{}>{\hspre}l<{\hspost}@{}}%
\column{29}{@{}>{\hspre}l<{\hspost}@{}}%
\column{37}{@{}>{\hspre}l<{\hspost}@{}}%
\column{E}{@{}>{\hspre}l<{\hspost}@{}}%
\>[B]{}\Varid{newField}\mathbin{:}\mathrm{1}\to \Conid{Widget}{}\<[E]%
\\
\>[B]{}\Varid{newField}\;\anonymous \mathrel{=}{}\<[15]%
\>[15]{}\mathbf{let}\;\Varid{remove}{}\<[29]%
\>[29]{}\mathrel{=}\Varid{simpleButton}\;\text{\ttfamily \char34 Remove\char34}\;\mathbf{in}{}\<[E]%
\\
\>[15]{}\mathbf{let}\;\Varid{tfAndBtn}{}\<[29]%
\>[29]{}\mathrel{=}[\mskip1.5mu \con{textfield}\;(\Varid{simpleTF}\;\text{\ttfamily \char34 \char34}),\con{button}\;\Varid{remove}\mskip1.5mu]{}\<[E]%
\\
\>[15]{}\mathbf{in}\;{}\<[19]%
\>[19]{}\con{stack}\;\con{horizontal}\;{}\<[37]%
\>[37]{}(\Varid{tfAndBtn}\mathbin{::}(\Varid{const}\;\con{nil}\mathbin{@}\Varid{onClick}\;\Varid{remove})){}\<[E]%
\ColumnHook
\end{hscode}\resethooks
The function adds a stack containing a text field and a `Remove'
button, but as soon as the `Remove' button is clicked, the stack
contents are changed to the empty list of widgets \ensuremath{\con{nil}}.

\subsection{Comparison to Async RaTT}
\label{sec:async-ratt}

As mentioned in the introduction, \Rizzo takes inspiration from the
Async RaTT calculus of \citet{bahr23AsynchronousModalFRP} but is
simpler and more expressive. With the help of the examples in the
previous sections, we can explore the differences in more detail.

\subsubsection{Expressiveness}
\label{sec:expressiveness}
Like \Rizzo, Async RaTT also features two later modalities (\ensuremath{\DelayE} and
\ensuremath{\DelayA}) and a notion of clocks. However, in order to obtain guarantees
about space leaks, Async RaTT also features the stable modality \ensuremath{\Box}
and a Fitch-style type system that restricts the scope of variables
unless they are of a stable type. A type $A$ is stable if all
occurrences of $\to$, \ensuremath{\Sig}, and \ensuremath{\DelayE} in $A$ are guarded by \ensuremath{\Box}. As
a result, functions, signals, and delayed computations cannot be moved
across time. For example, \ensuremath{\Varid{map}} is of type \ensuremath{\Box(\Conid{A}\to \Conid{B})\to \Sig\;\Conid{A}\to \Sig\;\Conid{B}} in Async RaTT, i.e.\ the function must be boxed with the \ensuremath{\Box}
modality so that \ensuremath{\Varid{map}} can use it at any time in the future.
This limits what kind of functions we may give \ensuremath{\Varid{map}}. For example, the
function we pass to \ensuremath{\Varid{map}} may not contain a signal in its closure, as
it does in our definition of \ensuremath{\Varid{sample}} in \Rizzo. 

All combinators in Fig.~\ref{fig:library} apart from \ensuremath{\Varid{filter}} and
\ensuremath{\Varid{sample}} are expressible in Async RaTT as well, but most of them with
additional restrictions: All function arguments have to be boxed with \ensuremath{\Box}
similarly to the function argument to \ensuremath{\Varid{map}}. Additionally, \ensuremath{\Varid{zip}}
requires \ensuremath{\Conid{A}} and \ensuremath{\Conid{B}} to be stable, and \ensuremath{\Varid{scan}} and \ensuremath{\Varid{scanAwait}} require
\ensuremath{\Conid{B}} to be stable. This means that these combinators cannot be used when
dealing with nested signals, which are ubiquitous when implementing
GUIs as we have seen with the \ensuremath{\Conid{Widget}} type in
section~\ref{sec:gui-app}. Indeed, our use of \ensuremath{\Varid{scanAwait}} to implement
\ensuremath{\Varid{allWidgets}} would not be allowed in Async RaTT as \ensuremath{\List\;\Conid{Widget}} is
not a stable type.

The combinators \ensuremath{\Varid{sample}} and \ensuremath{\Varid{filter}} are not expressible in Async
RaTT at all due to two fundamental differences in the notion of
clocks: First, \Rizzo clocks only indicate timing dependencies but not
data dependencies, while Async RaTT clocks appear to indicate both
timing and data dependencies. For example, $\ensuremath{\Varid{sample}\;\Varid{xs}\;\Varid{ys}}$ has a data
dependency on both $\ensuremath{\Varid{xs}}$ and $\ensuremath{\Varid{ys}}$ as the values produced by
$\ensuremath{\Varid{sample}\;\Varid{xs}\;\Varid{ys}}$ depend on the values produced by $\ensuremath{\Varid{xs}}$ and $\ensuremath{\Varid{ys}}$.
However, $\ensuremath{\Varid{sample}\;\Varid{xs}\;\Varid{ys}}$ has a timing dependency on only $\ensuremath{\Varid{xs}}$ as
$\ensuremath{\cl{\tail\;(\Varid{sample}\;\Varid{xs}\;\Varid{ys})}\mathrel{=}\cl{\tail\;\Varid{xs}}}$. Async RaTT's apparent conflation of
timing and data dependencies means that Async RaTT cannot express
dataflow graphs where data and timing dependencies do not coincide.
Second, \Rizzo has a more expressive notion of clocks, which consist
of channels (just like Async RaTT) and partial signals (not available
in Async RaTT clocks). The latter enables the implementation of \ensuremath{\Varid{filter}}.

The less restrictive type system of \Rizzo also makes some
programs more modular and easier to express than in Async RaTT. For
example, consider the following combinator in \Rizzo:
\begin{hscode}\SaveRestoreHook
\column{B}{@{}>{\hspre}l<{\hspost}@{}}%
\column{E}{@{}>{\hspre}l<{\hspost}@{}}%
\>[B]{}\Varid{switchAt}\mathbin{:}\Sig\;\Conid{C}\to \Sig\;\Conid{C}\to \DelayE\;\mathrm{1}\to \Sig\;\Conid{C}{}\<[E]%
\\
\>[B]{}\Varid{switchAt}\;\Varid{xs}\;\Varid{ys}\;\Varid{d}\mathrel{=}\Varid{switch}\;\Varid{xs}\;(\Varid{ys}\mathbin{@}\Varid{d}){}\<[E]%
\ColumnHook
\end{hscode}\resethooks
It produces a signal that first behaves like \ensuremath{\Varid{xs}} but switches to \ensuremath{\Varid{ys}}
as soon as \ensuremath{\cl{\Varid{d}}} ticks. While \ensuremath{\Varid{switchAt}} can be implemented in Async
RaTT, it is restricted to stable types \ensuremath{\Conid{C}}. Moreover, the Async RaTT
implementation is much more complicated: Async RaTT can express \ensuremath{\mathbin{@}\mathbin{:}\Conid{A}\to \DelayE\;\Conid{B}\to \DelayE\;\Conid{A}}, but only if \ensuremath{\Conid{A}} is stable, and thus \ensuremath{\mathbin{@}} cannot be
used for \ensuremath{\Varid{ys}\mathbin{:}\Sig\;\Conid{C}}, even if \ensuremath{\Conid{C}} is stable. Therefore, \ensuremath{\Varid{switchAt}}
has to be implemented as a guarded recursive function that performs a
three-way synchronisation of \ensuremath{\tail\;\Varid{xs}}, \ensuremath{\tail\;\Varid{ys}}, and \ensuremath{\Varid{d}} that
requires complicated nested uses of \ensuremath{\select} (see
Appendix~\ref{sec:switchAt} in the supplementary material). This is an
example of a more general pattern often encountered in Async RaTT: By
disallowing direct access to signals in any future computational
context (e.g.\ access to \ensuremath{\Varid{ys}\mathbin{:}\Sig\;\Conid{A}} when \ensuremath{\Varid{d}\mathbin{:}\DelayE\;\mathrm{1}} has ticked),
Async RaTT forces us to manually keep track of the updates of each
signal we want to access in a future context.  In general, keeping
track of $n$ such signals can lead to $O(2^n)$ case distinctions due
to nested uses of \ensuremath{\select} and also restricts signals to stable
payload types, which in turn precludes some use cases like GUIs as
mentioned above.

Finally, while Async RaTT does not feature \Rizzo's first-class
channels, it would not be difficult to add them. Likewise, Async RaTT
distinguishes between push and pull channels, which would
not be difficult to add to \Rizzo. Async RaTT features general guarded
recursive types, while \Rizzo only has \ensuremath{\Sig}. However, Async RaTT's
operational semantics only produces outputs for signals (whereas
\Rizzo produces outputs for any type), and in fact none of the
examples presented by \citet{bahr23AsynchronousModalFRP} use a guarded
recursive type other than $\Sig$.

\subsubsection{Preventing Space Leaks}

The difference in expressiveness between \Rizzo and Async RaTT is due
to their differing strategies to avoid space leaks. As discussed
above, Async RaTT requires the \ensuremath{\Varid{map}} function to take a boxed
function. Otherwise, Async RaTT could express the following:
\begin{hscode}\SaveRestoreHook
\column{B}{@{}>{\hspre}l<{\hspost}@{}}%
\column{E}{@{}>{\hspre}l<{\hspost}@{}}%
\>[B]{}\Varid{addHead}\mathbin{:}\Sig\;\Nat\to \Sig\;\Nat{}\<[E]%
\\
\>[B]{}\Varid{addHead}\;\Varid{xs}\mathrel{=}\Varid{map}\;(\lambda \Varid{x}.\Varid{x}\mathbin{+}\head\;\Varid{xs})\;\Varid{xs}{}\<[E]%
\ColumnHook
\end{hscode}\resethooks
If allowed in Async RaTT, \ensuremath{\Varid{addHead}} would have a space leak, because it has
to keep the entire history of the input signal \ensuremath{\Varid{xs}} in memory. Each
time a new number is received on the input signal $\ensuremath{\Varid{xs}}$, that
number has to be stored indefinitely. The problem is that \ensuremath{\Varid{map}}
recursively traverses the signal $\ensuremath{\Varid{xs}}$, i.e.\ it traverses forward in
the signal as new inputs arrive, while the function closure for $\ensuremath{\lambda \Varid{x}.\Varid{x}\mathbin{+}\head\;\Varid{xs}}$ contains $\ensuremath{\Varid{xs}}$, which points to the
very beginning of the entire history of the input signal. To
illustrate this, let's assume an input
signal $\ensuremath{\Varid{xs}}$ of consecutive numbers starting from $0$. After
receiving $7$ values on $\ensuremath{\Varid{xs}}$, the memory
representation of the input signal would look like this:
\begin{center}
  \begin{tikzpicture}[
  cell/.style={
    rectangle split, rectangle split horizontal, rectangle split parts=2,
    draw, minimum width=5mm, minimum height=4.5mm, 
    font=\small, align=center
  },
  ptr/.style={-Stealth, thick},
  node distance=10mm and 8mm
]

\node[cell] (n0)  {0\nodepart{two} $\bullet$};
\node[cell] (n1)  [right=of n0] {1\nodepart{two}$\bullet$};
\node[cell] (n2)  [right=of n1] {2\nodepart{two}$\bullet$};
\node[cell] (n3)  [right=of n2] {3\nodepart{two}$\bullet$};
\node[cell] (n4)  [right=of n3] {4\nodepart{two}$\bullet$};
\node[cell] (n5)  [right=of n4] {5\nodepart{two}$\bullet$};
\node[cell] (n6)  [right=of n5] {6\nodepart{two}$\bullet$};

\foreach \i/\j in {n0/n1,n1/n2,n2/n3,n3/n4,n4/n5,n5/n6}
  \draw ($0.5*(\i.two east) + 0.5*(\i.two split)+(0,0.012)$) edge[ptr,bend left=20] ($(\j.west)+(0,0.012)$);

\begin{scope}[every node/.style={inner sep=0.15em}]
  \node[node distance=2em,below=of n2] (map) {$\ensuremath{\Varid{map}\;(\lambda \Varid{x}.\Varid{x}\mathbin{+}\head}$};
  \node[anchor=base west,at=(map.base east)] (xs) {$xs$};
  \node[anchor=base west,at=(xs.base east)] (p) {$\!\!)$};
  \node[anchor=base west,at=(p.base east)] (xs') {$\ensuremath{\Varid{xs}_{\mathrm{6}}}$};
\end{scope}
  \draw[thick,-Stealth] 
    (xs) edge[out=130,in=-50] (n0)
    (xs') edge[out=70,in=-130] (n6);
\end{tikzpicture}
\vspace{-0.7em}
\end{center}

The recursively defined \ensuremath{\Varid{map}} function traverses the signal $\ensuremath{\Varid{xs}}$ by
one step each time $\ensuremath{\Varid{xs}}$ receives a new value. After $\ensuremath{\Varid{xs}}$ has
received the number $6$, we are at the $6$th recursive call of \ensuremath{\Varid{map}}
(illustrated above) on the signal $\ensuremath{\Varid{xs}_{\mathrm{6}}}$ that starts with the number
$6$. However, the variable $\ensuremath{\Varid{xs}}$ in the function closure still points
to the initial signal starting at $0$. Hence, we have to keep the
entire prefix of the signal from \ensuremath{\mathrm{0}} to \ensuremath{\mathrm{6}} in memory. The longer the
program runs, the more memory it will use. 

Async RaTT avoids the above situation by rejecting \ensuremath{\Varid{addHead}} as
ill-typed. By contrast, \Rizzo avoids this situation by choosing a
different semantics that does not allow signals to store their past.
If \ensuremath{\Varid{addHead}} were implemented in \Rizzo, \ensuremath{\Varid{xs}} would instead
point at the same signal as \ensuremath{\Varid{xs}_{\mathrm{6}}}.

\section{Operational Semantics and Operational Properties}
\label{sec:operational-semantics}

The purpose of this section is to give a precise account of the
operational guarantees provided by \Rizzo. To this end, we first give
a formal definition of the operational semantics of \Rizzo and then
state the operational guarantees in terms of this formal operational
semantics. Specifically, we give a formal account of the
\emph{reactive evaluation semantics} $\reacts\tau$ that we used
informally in section~\ref{sec:rizzo}. Given a closed \Rizzo term
$\hastype{}{t}{A}$, a sequence of events $\tau$ causes the pair
$\state{t;\Delta}$ to evaluate to the pair $\state{v;\Delta'}$,
denoted $\state{t;\Delta}\reacts\tau \state{v;\Delta'}$. Since $t$ may
allocate channels via $\chan$, we may obtain a new channel context
$\Delta'$ that extends the original $\Delta$. When the channel context
is not relevant (e.g.\ when $t$ does not use $\chan$), we write $t
\reacts\tau v$ to mean that $\state{t;\Delta}\reacts\tau
\state{v;\Delta'}$ for some $\Delta'$.

We give an operational semantics with the following three
characteristics: First, it has \emph{by construction} no space leaks,
i.e.\ old inputs are not kept in memory implicitly. Second, it
computes its result incrementally, i.e.\ given
$\state{t;\Delta}\reacts{\tau} \state{v;\Delta'}$ and an additional
event $e$, the operational semantics effectively computes
$\state{t;\Delta}\reacts{\tau,e} \state{w;\Delta''}$ from
$\state{v;\Delta'}$ and $e$ alone. Third, the operational
semantics computes $w$ by performing in-place updates of the values
that changed within $v$.

To obtain such an operational semantics, we define a machine that
represents the value $v$ produced by $\reacts{\tau}$ as a pair
$\state{p;\eta}$, where $p$ is a \emph{machine value} that represents
each signal in $v$ as a pointer to a location $l$ in the \emph{heap}
$\eta$ that stores the signal. These machine values are defined as
follows:
\[
    \arraycolsep=2pt
    \begin{array}{llcl} 
&p,q &::= & \unit \mid \lambda x.t \mid \pair{p}{q} \mid \interm_i\, p \mid  l \mid\tail\,l\mid
\kappa\mid \wait\,\kappa\mid \watch\,l \mid p \appAE q \mid\delay\, t\\
&&\mid & \never \mid
\select\,p\,q\mid\cons[\mu\alpha.F]\,q
\end{array}
\]
That is, machine values $p,q$ only differ from values $v,w$ (as defined in
Fig.~\ref{fig:syntax}) in the inclusion of $l$ and $\tail\, l$ and the
exclusion of $v ::_A w$. The intuition is that if the heap $\eta$
stores the representation of a signal $v ::_A w$ at $l$, then the
machine value $l$ represents the signal $v ::_A w$ and the machine
value $\tail\,l$ represents the delayed signal $w$.

Using this representation of values, we define a machine that performs
the reactive evaluation $\state{t;\Delta}\reacts{\tau}
\state{v;\Delta'}$ in incremental steps, with each such step
performing an update in response to a single event in $\tau$. Given
$\tau = e_0,e_1,\dots,e_{n-1}$, the machine first performs an
initialisation step $\initstep$ and then successively performs a
reactive step $\fullstep{e_i}$ for each event $e_i$:
\[
\state{t;\Delta} \initstep
\state{p;\eta_0\sep\Delta_0} \fullstep{e_0} \state{p;\eta_1\sep\Delta_1} 
\fullstep{e_1}\quad \dots\quad \fullstep{e_{n-1}} \state{p;\eta_{n}\sep\Delta_{n}} 
\]
The machine's state $\state{p;\eta_i\sep\Delta_i}$ consists of the
value $p$ that $t$ has been evaluated to, a heap $\eta_i$ that stores
signals, and a channel context $\Delta_i$. The value $p$ contains
references into the heap $\eta_i$ and does not change after the
initialisation step. In particular, the machine evaluates terms of the
form \ensuremath{\Varid{s}\mathbin{::}_{\Conid{A}}\Varid{s'}} by evaluating the terms $s$ and $s'$ to values $q$ and
$q'$, respectively, and then returning a reference $l$ that points to
a location in $\eta_i$ where the signal \ensuremath{\Varid{q}\mathbin{::}_{\Conid{A}}\Varid{q'}} is stored. Signals
are updated in place since the machine only changes the heap $\eta_i$.
Each pair $\state{p;\eta_i}$ represents the value $v_i$ defined by
$v_i = p[\eta_i]$, which replaces each reference $l$ in $p$ with the
signal stored at $l$ in $\eta_i$. In particular,
$\state{t;\Delta}\reacts{\tau} \state{p[\eta_n];\Delta_n}$. In the
sections that follow, we give a more precise definition of the
reactive evaluation semantics $\reacts{\tau}$ and the underlying
machine.

\subsection{Evaluation Semantics}
\label{sec:evaluation-semantics}

\begin{figure}[t]
  \small
 \textsc{Evaluation semantics \quad $\heval{t}{\varepsilon}{p}{\varepsilon'}$}\\
\begin{mathpar}
  \mprset{andskip=1em}
  \mprset{sep=1em}
  \inferrule*%
  {~}%
  {\heval{p}{\varepsilon}{p}{\varepsilon}}%
  \and%
  \inferrule*%
  {\heval {s} {\varepsilon} {p} {\varepsilon'}\\
    \heval {t} {\varepsilon'} {q} {\varepsilon''}}%
  {\heval {\pair{s}{t}} {\varepsilon} {\pair{p}{q}} {\varepsilon''}}%
  \and%
  \inferrule*%
  {\heval {t} {\varepsilon} {\pair{p_1}{p_2}} {\varepsilon'} \\ i \in \{1,2\}}%
  {\heval {\pi_i\,t} {\varepsilon} {p_i} {\varepsilon'}}%
  \and%
  \inferrule*%
  {\heval t {\varepsilon} p {\varepsilon'} \\ i \in \{1,2\}}%
  {\heval {\interm_i\,t} {\varepsilon} {\interm_i\,p} {\varepsilon'}}%
  \and%
  \inferrule*%
  {\heval {t} {\varepsilon} {\interm_i\,p} {\varepsilon'}\\
    \heval {t_i[p/x]} {\varepsilon'} {p_i} {\varepsilon''} \\ i \in \{1,2\}}%
  {\heval {\caseterm{t}{x}{t_1}{x}{t_2}} {\varepsilon} {p_i} {\varepsilon''}}%
  \and%
  \inferrule*%
  {\heval{t}{\varepsilon}{\lambda x.s}{\varepsilon'} \\
    \heval{t'}{\varepsilon'}{p}{\varepsilon''}\\
    \heval {s[p/x]}{\varepsilon''}{q}{\varepsilon'''}}%
  {\heval{t\,t'}{\varepsilon}{q}{\varepsilon'''}}%
  \and%
  \inferrule*%
  {\heval{t}{\varepsilon}{p}{\varepsilon'}}%
  {\heval{\cons[\mu\,\alpha.F] \, t}{\varepsilon}{\cons[\mu\,\alpha.F] \, p}{\varepsilon'}}%
  \and%
  \inferrule*%
  {\heval{t}{\varepsilon}{\cons[\mu\,\alpha.F]\,p}{\varepsilon'}\\\heval{\fmap{F}\,(\lambda\,y.\pair
    y {\rec{x}{s}{y}})\,p}{\varepsilon'}{q}{\varepsilon''}
  \\\heval{s[q/x]}{\varepsilon''}{q'}{\varepsilon'''}}%
  {\heval{\rec{x}{s}{t}}{\varepsilon}{q'}{\varepsilon'''}}%
  \and%
  \inferrule*%
  {\heval
    {s}{\varepsilon} {\delay\,s'}{\varepsilon'}\\ \heval{t}{\varepsilon'}{\delay\, t'}{\varepsilon''}}%
  {\heval {s \appA t}{\varepsilon}{\delay\,(s'\,t')}{\varepsilon''}}%
  \and%
  \inferrule*%
  {\heval {s} {\varepsilon} {p} {\varepsilon'}\\
    \heval {t} {\varepsilon'} {q} {\varepsilon''}}%
  {\heval {s \appAE t} {\varepsilon} {p \appAE q} {\varepsilon''}}%
  \and%
  \inferrule*%
  {\heval{t}{\varepsilon}{\kappa}{\varepsilon'}}%
  {\heval{\wait \, t}{\varepsilon}{\wait \, \kappa}{\varepsilon'}}%
  \and%
  \inferrule*%
  {\heval{t}{\varepsilon}{l}{\varepsilon'}}%
  {\heval{\watch \, t}{\varepsilon}{\watch \, l}{\varepsilon'}}%
    \and
  \inferrule*%
  {\heval {s} {\varepsilon} {p} {\varepsilon'}\\
    \heval {t} {\varepsilon'} {q} {\varepsilon''}}%
  {\heval {\select\,{s}\,{t}} {\varepsilon} {\select\,{p}\,{q}} {\varepsilon''}}%
  \and%
  \inferrule*%
  {\heval {t} {\varepsilon} {l} {\varepsilon'}}%
  {\heval {\tail\, t} {\varepsilon} {\tail\,l } {\varepsilon'}}%
  \and%
  \inferrule*%
  {\heval {t} {\varepsilon} {l} {\eta_N\tick\eta_E\sep\Delta}\\\eta_N(l) = \sig p q }%
  {\heval {\head\, t} {\varepsilon} {p} {\eta_N\tick\eta_E\sep\Delta}}%
  \and%
  \inferrule*%
  {\kappa = \allocate{\Delta}}%
  {\heval{\chan}{\eta_N\tick\eta_E\sep\Delta}{\kappa}{\eta_N\tick\eta_E\sep\Delta,\kappa \chantype A}}%
  \and%
  \inferrule*%
  {\heval {s} {\varepsilon} {p} {\varepsilon'}\\
    \heval {t} {\varepsilon'} {q} {\eta_N\tick\eta_E\sep\Delta}\\
    l = \allocate{\eta_N\tick\eta_E}}%
  {\heval {s ::_A t} {\varepsilon} {l} {\eta_N,l \sigtype A \mapsto \sig[\unchanged] p q\tick\eta_E\sep\Delta}}%
  \and%
    \inferrule*%
  {\heval
    {t[\delay\,(\fix\,x.t)/x]}{\varepsilon} {p}{\varepsilon'}}%
  {\heval {\fix\,x. t}{\varepsilon}{p}{\varepsilon'}}%
\end{mathpar}
~\\
\text{The term $\fmap{F}\; \vec f : F[\vec A/\vec \alpha] \to F[\vec
B/\vec\alpha]$ takes $n$ functions $\vec f = f_1,\dots,f_n$ where each
$f_i : A_i \to B_i$:}\\[0.5em]
\begin{minipage}{0.34\textwidth}
\begin{align*}
  &\begin{aligned}
    &\fmap{\alpha_i}& \hspace{-1em}\vec f\, x &= f_i\; x\\
    &\fmap{C \to F}&\hspace{-1em}\vec f\, x &= \lambda y.\,\fmap{F}\,\vec f\,(x\,y)\\
    &\fmap{\Sig\,F}&\hspace{-1em} \vec f\, x &= \ensuremath{\Varid{map}}\, (\fmap{F}\,\vec f)\,x\\
  \end{aligned}
\end{align*}
\end{minipage}
\begin{minipage}{0.65\textwidth}
\begin{align*}
  &\begin{aligned}
    &\fmap{C}& \, \vec f\, x &= x\quad \text{if }C \text{ is of the form } \Unit,\DelayE D, \DelayA D, \Chan\, D\\
    &\fmap{F \times G}& \vec f\, x &= \pair{\fmap F\,\vec f\, (\pi_1\,x)}{\fmap G\,\vec f\, (\pi_2\,x)}\\
    &\fmap{F + G}&\vec f\, x &= \caseterm x y
    {\interm_1(\fmap F \, \vec f\, y)} z {\interm_2(\fmap G
      \, \vec f\, z)}\\
    &\fmap{\mu\alpha_{n+1}. F}\hspace{-10em}&\vec f\, x &= \rec{z}{\cons[{\mu\alpha_{n+1}.F[\vec
B/\vec\alpha]}]\,(\fmap{F}\,\vec f\,(\lambda x . \pi_2\,x)\,z)}{x}.
  \end{aligned}
\end{align*}

\end{minipage}
  \vspace{-0.5em}
  \caption{Evaluation semantics.}
  \label{fig:machine}
\end{figure}

The machine performs a step $\state{p;\eta\sep\Delta} \fullstep{e}
\state{p;\eta'\sep\Delta'}$ in reaction to an event $e$ and updates
the signals stored in the heap $\eta$ resulting in a new heap $\eta'$.
Along the way, the machine needs to be able to evaluate a term $t$ to a
value $p$. To do so, the machine needs to keep track of which part of the
heap has already been updated and which part still needs updating.
Hence, when evaluating a term, the machine uses a \emph{store} $\sigma$ that
consists of two heaps: a \emph{\now heap} $\eta_N$ that stores signals
that have already been updated and an \emph{\earlier heap} $\eta_E$
that still requires updating. We write $\eta_N\tick \eta_E$ to denote
a store with \now heap $\eta_N$ and \earlier heap $\eta_E$. The
\emph{evaluation semantics}, denoted
$\heval{t}{\varepsilon}{p}{\varepsilon'}$ and defined in
Fig.~\ref{fig:machine}, describes how the machine evaluates a term $t$
to a value $p$. This evaluation semantics uses an \emph{environment}
$\varepsilon$ that consists of a store $\sigma$ and a channel context
$\Delta$, written $\sigma\sep\Delta$.

A \emph{heap} is a sequence of assignments of the form $l \sigtype A
\mapsto \sig p q$, each of which maps a \emph{heap location} $l$ of
type $\ensuremath{\Sig\;\Conid{A}}$ to a \emph{stored signal} $\sig p q$ consisting of a
head value $p$ of type $A$ and a tail value $q$ of type $\DelayE
(\Sig\, A)$. In addition, the stored signal contains a flag $U$
indicating whether the signal has been updated in the current $\fullstep{e}$
step. This flag $U$ has no bearing on the evaluation semantics, but it
is important for the step semantics of the machine discussed in
section~\ref{sec:step-semantics}. We write $\dom{\eta}$ for the domain
of a heap $\eta$, i.e.\ the set of heap locations stored in $\eta$; and
we write $\dom{\sigma}$ for the domain of a store $\sigma$, i.e.\
$\dom{\eta_N\tick \eta_E} =\dom{\eta_N} \cup \dom{\eta_E}$.

Intuitively speaking, signals stored in the \now heap $\eta_N$ are up
to date and can be safely dereferenced, whereas signals in the
\earlier heap $\eta_E$ are stale and need updating before they can be
dereferenced. The machine maintains the invariant that, for each store
$\eta_N\tick\eta_E$, we have $\dom{\eta_N} \cap \dom{\eta_E} =
\emptyset$, i.e.\ for each heap location $l$, there is at most one
mapping for $l$ in $\eta_N\tick\eta_E$. In order to allocate fresh
channels and heap locations, we assume a function $\allocate{\cdot}$
so that $\allocate{\Delta}$ produces a fresh $\kappa\nin
\dom{\Delta}$, and $\allocate{\sigma}$ produces a fresh $l \nin
\dom{\sigma}$.

The evaluation semantics of the lambda calculus fragment of \Rizzo is
the standard call-by-value semantics and makes up the top half of the
rules in Fig.~\ref{fig:machine}. This includes the semantics for
product, sum, function, and recursive types. The semantics of the
primitive recursion combinator $\rec*$ on recursive types $\mu \alpha.
F$ is defined using a term $\fmap F : (A \to B) \to F[A/\alpha]\to
F[B/\alpha]$ where $\istype{\alpha}{F}$. Since we allow nested
recursive types, $\fmap F$ has to be defined more generally for
$\istype{\alpha_1,\dots,\alpha_n}{F}$ as shown at the bottom of
Fig.~\ref{fig:machine}.


The primitives that produce values in the \ensuremath{\DelayE} modality -- i.e.\ \ensuremath{\appAE},
\ensuremath{\wait}, \ensuremath{\watch}, \ensuremath{\select}, and \ensuremath{\tail} -- behave similarly to strict
constructors of an algebraic data type: They eagerly evaluate their
arguments to values and produce values of the form \ensuremath{\Varid{p}\appAE\Varid{q}}, \ensuremath{\wait\;\kappa}, \ensuremath{\watch\;\Varid{l}}, \ensuremath{\select\;\Varid{p}\;\Varid{q}}, and \ensuremath{\tail\;\Varid{l}}, respectively. Such values
of type \ensuremath{\DelayE\;\Conid{A}} are then later evaluated to values of type \ensuremath{\Conid{A}} by the
\emph{advance semantics} of the machine presented in
section~\ref{sec:step-semantics}. Note in particular that the
evaluation semantics does \emph{not} evaluate  \ensuremath{\tail\;\Varid{l}} further by
looking up the tail of $l$ stored in the \now heap. Instead, the
semantics relies on the fact that the machine's advance semantics will
evaluate \ensuremath{\tail\;\Varid{l}} to \ensuremath{\Varid{l}} once the signal at $l$ itself has already
been updated to its new value by the \emph{update semantics} of the
machine (also presented in section~\ref{sec:step-semantics}). This
ensures that the delayed computation stored in the tail of a signal is
only performed once and at the appropriate time, namely when its clock has ticked.

By contrast, \ensuremath{\head\;\Varid{t}} does indeed look up the current value of the
signal referred to by \ensuremath{\Varid{t}} and returns it. The signal constructor $::_A$
simply evaluates its arguments and stores them as a stored signal at a
freshly allocated location on the \now heap; and $\ensuremath{\chan[\Conid{A}]}$ returns a
freshly allocated channel of type \ensuremath{\Chan\;\Conid{A}}. Finally, \ensuremath{\fix} has the
standard guarded fixed point semantics: \ensuremath{\fix\;\Varid{x}.\Varid{t}} evaluates \ensuremath{\Varid{t}} with
\ensuremath{\Varid{x}} replaced by \ensuremath{\delay\;(\fix\;\Varid{x}.\Varid{t})}. The latter is a value and thus
does not evaluate further itself, which (in concert with the type
system) ensures termination of the evaluation semantics.

\subsection{Step Semantics of the Machine and Reactive Evaluation Semantics}
\label{sec:step-semantics}

\begin{figure}
  \small
   \textsc{Advance semantics \quad $\evalAdv{p}{\varepsilon}{q}{\varepsilon'}$}\\
  \begin{mathpar}
    \inferrule*%
    {\evalAdv{p_1}{\varepsilon}{p_2}{\varepsilon'}\\\heval{t\,p_2}{\varepsilon'}{q}{\varepsilon''}}%
    {\evalAdv{\delay\,t \appAE p_1}{\varepsilon}{q}{\varepsilon''}}%
    \and
    \inferrule*%
    {\eta_N(l) = \sig[\updated]{\interm_1\,p_1}{p_2}}
    {\evalAdv{\watch\,l}{\eta_N\tick\eta_E\sep\Delta}{p_1}{\eta_N\tick\eta_E\sep\Delta}}%
    \and%
    \inferrule*%
    {\heval{t}{\varepsilon}{q}{\varepsilon'}}
    {\evalAdv[\kappa\mapsto t]{\wait\,\kappa}{\varepsilon}{q}{\varepsilon'}}%
    \and%
    \inferrule*%
    {\evalAdv[\kappa\mapsto t]{p_i}{\varepsilon}{q}{\varepsilon'}\\
    \ticked{\varepsilon}{p_i}\\ \neg\ticked{\varepsilon}{p_{3-i}}}
    {\evalAdv[\kappa \mapsto t]{\select\,p_1\,p_2}{\varepsilon}{\interm_1\,(\interm_i\,q)}{\varepsilon'}}%
    \and
    \inferrule*%
    {~}
    {\evalAdv{\tail\,l}{\varepsilon}{l}{\varepsilon}}%
    \and
    \inferrule*%
    {\evalAdv[\kappa\mapsto t]{p_1}{\varepsilon}{q_1}{\varepsilon'}\\
    \evalAdv[\kappa\mapsto t]{p_2}{\varepsilon'}{q_2}{\varepsilon''}\\\\
    \ticked{\varepsilon}{p_1}\\ \ticked{\varepsilon}{p_2}}
    {\evalAdv[\kappa \mapsto t]{\select\,p_1\,p_2}{\varepsilon}{\interm_2\,\pair{q_1}{q_2}}{\varepsilon''}}%
  \end{mathpar}
  ~\\[0.2em]
  \textsc{Update semantics \quad $\state{\varepsilon}\partstep{e} \state{\varepsilon'}$}\\
  \begin{mathpar}
    \inferrule*%
    {\neg\ticked{\eta_N}{p_2}}
    {\state{\eta_N\tick l \sigtype A  \mapsto \sig{p_1}{p_2}, \eta_E\sep\Delta}\partstep{\kappa\mapsto t} \state{\eta_N,l \sigtype A \mapsto \sig[\unchanged]{p_1}{p_2}\tick \eta_E\sep\Delta}}
    \and
    \inferrule*%
    {\ticked{\eta_N}{p_2}\\\evalAdv[\kappa\mapsto t]{p_2}{\eta_N\tick l \sigtype A \mapsto \sig{p_1}{p_2}, \eta_E\sep\Delta}{l'}{\eta'_N\tick l \sigtype A \mapsto \sig{p_1}{p_2}, \eta_E\sep\Delta'}\\\eta'_N(l') = \sig[V]{q_1}{q_2}}
    {\state{\eta_N\tick l \sigtype A \mapsto \sig{p_1}{p_2}, \eta_E\sep\Delta}\partstep{\kappa\mapsto t} \state{\eta'_N,l \sigtype A \mapsto \sig[\updated]{q_1}{q_2}\tick \eta_E\sep\Delta'}}
\end{mathpar}
  ~\\[0.2em]
  \textsc{Step semantics \quad $\fullstep{e}$ \quad and \quad $\initstep$}\\
  \begin{mathpar}
    \inferrule*%
    {\state{\emptyheap\tick\eta\sep\Delta}\partstep*{e} \state{\eta'\tick\emptyheap\sep\Delta'}}
    {\state{p;\eta\sep\Delta}\fullstep{e} \state{p;\eta'\sep\Delta'}}
    \and
    \inferrule*%
    {\heval{t}{\emptyheap\tick\emptyheap\sep\Delta}{p}{\eta\tick\emptyheap\sep\Delta'}}
    {\state{t;\Delta}\initstep \state{p;\eta\sep\Delta'}}
  \end{mathpar}
  ~\\[0.2em]
  \begin{minipage}[t]{0.63\textwidth}
    \centering
  \textsc{Reactive evaluation semantics \quad $\state{t;\Delta} \reacts{\tau} \state{v;\Delta'}$}\\
  \begin{mathpar}
\inferrule{
  \isev[\Delta_i]{e_i}\quad \text{ for all } 0 \le i < n\\\\
\state{t;\Delta} \initstep
\state{p;\eta_0\sep\Delta_0} \fullstep{e_0} \state{p;\eta_1\sep\Delta_1} 
\fullstep{e_1}\quad \dots\quad \fullstep{e_{n-1}} \state{p;\eta_{n}\sep\Delta_{n}}}
{\state{t;\Delta} \reacts{e_0,\dots,e_{n-1}} \state{p[\eta_n];\Delta_n}}
  \end{mathpar}
  \end{minipage}
  \begin{minipage}[t]{0.35\textwidth}
    \centering
   \textsc{Events \quad $\isev{e}$}\\[1.6em]
  \begin{mathpar}
\inferrule
   {\kappa : \Chan\,A\in\Delta \\ \hastype{}{t}{A}}{\isev{\kappa \mapsto t}}
  \end{mathpar}
  \end{minipage}  
  ~\\[1em]
  where $s[\eta]$ is defined by \quad $s[\cdot] = s$ \quad and
   \quad $s[\eta,l \sigtype A \mapsto \sig{p}{q}] = ((s[q/\tail\, l])[p ::_A q/l])[\eta]$
   \\[1.2em]
  \textsc{Ticked predicate \quad $\ticked{\eta}{p}$}\\
  \begin{equation*}
\begin{aligned}
  \ticked{\eta}{\never} &\Leftrightarrow \bot\\
  \ticked{\eta}{p \appAE q} &\Leftrightarrow \ticked{\eta} q\\
  \ticked{\eta}{\wait\,\kappa'} &\Leftrightarrow \kappa = \kappa'\\
  \ticked{\eta_N\tick\eta_E\sep\Delta}{p} &\Leftrightarrow \ticked{\eta_N\tick\eta_E}{p}
\end{aligned}
\hspace{1cm}
\begin{aligned}
  \ticked{\eta}{\watch\,l} &\Leftrightarrow \exists p, q. \eta(l) = \sig[\updated] {\interm_1\, p} q\\
  \ticked{\eta}{\tail\,l} &\Leftrightarrow \exists p,q. \eta(l) = \sig[\updated] p q\\
  \ticked{\eta}{\select\; p\; q} &\Leftrightarrow \ticked{\eta}{p} \lor \ticked{\eta}{q}\\
\ticked{\eta_N\tick\eta_E}{p} &\Leftrightarrow \ticked{\eta_N}{p}
\end{aligned}
  \end{equation*}
  \vspace{-0.5em}
  \caption{Reactive evaluation semantics and step semantics of the machine.}
  \label{fig:reactiveSemantics}
\end{figure}

The step semantics $\fullstep{e}$, defined in
Fig.~\ref{fig:reactiveSemantics}, describes how the machine reacts to
each event $e$. To this end, the step semantics uses the evaluation
semantics from section~\ref{sec:evaluation-semantics} as well as two
additional components: The \emph{advance semantics}, denoted
$\evalAdv{p}{\varepsilon}{q}{\varepsilon'}$, describes how a value $p$
of type $\DelayE A$ is advanced to a value $q$ of type $A$ given that
the event $e$ has caused the clock of $p$ to tick. The advance
semantics is used by the \emph{update semantics}, denoted
$\state{\varepsilon}\partstep{e} \state{\varepsilon'}$, to update the
left-most signal in the \earlier heap of the environment
$\varepsilon$, so that it can be moved to the \now heap, which results
in the new environment $\varepsilon'$. Finally, the \emph{step
semantics} describes a complete computation step that is performed in
reaction to an event, denoted ${\state{p;\eta\sep\Delta}\fullstep{e}
\state{p;\eta'\sep\Delta'}}$, as well as the very first computation
step that initialises a program, denoted $\state{t;\Delta}\initstep
\state{p;\eta\sep\Delta'}$.

We first look at the step semantics: Given a well-typed term
$\hastype{}{t}{A}$, the step $\state{t;\Delta}\initstep
\state{p;\eta\sep\Delta'}$ initialises the program by evaluating $t$
in the context of $\Delta$ to a value $p$ using the evaluation
semantics from Fig.~\ref{fig:machine}. Evaluating $t$ may cause
signals to be stored in the \now heap and additional channels to be
allocated, resulting in a heap $\eta$ and a new channel context
$\Delta'$.

After this initialisation, the state of the machine is represented by
a tuple $\state{p;\eta\sep\Delta}$ consisting of a value $p$, a heap
$\eta$, and a channel context $\Delta$. Throughout the runtime of the
machine, $p$ will not change, but $p$ may refer to signals stored in
$\eta$ that \emph{are} subject to change. Each occurrence of an event
$e$ may cause some signals stored in $\eta$ to be updated. More
precisely, for each stored signal $l \sigtype A \mapsto \sig
{p_1}{p_2}$ in $\eta$, the clock of $p_2$ may tick in reaction to $e$.
In that case, the signal at $l$ must be updated with the new value of
type \ensuremath{\Sig\;\Conid{A}} produced by $p_2$ using the advance semantics. To update
all such stored signals in a systematic manner, the step semantics
takes the current heap $\eta$ and designates it as the \earlier heap.
This means that all signals in $\eta$ are considered stale and
possibly need updating. The actual update process is performed -- one
signal at a time -- by the update semantics $\partstep{e}$, which
takes the left-most signal from the \earlier heap and either moves it
to the \now heap unchanged, if the clock of the signal's tail did not
tick, or it performs an update and then moves the updated signal to
the \now heap, if the clock did tick. Whether an event $\kappa \mapsto
t$ causes the clock of a delayed computation $p$ to tick is indicated
by the $\ticked{\eta}{p}$ predicate defined at the bottom of
Fig.~\ref{fig:reactiveSemantics}. Since $p$ may contain references to
a heap $\eta$, it is included in the predicate.

To precisely relate $\ticked{\eta}{p}$ to the clock $\cl{v}$ of a value $v$,
as defined in section~\ref{sec:delayed-computations}, we define the
clock $\cl[\eta]{p}$ of a machine value $p$ such that it satisfies the correspondence
$(\cl[\eta]{p})[\eta] = \cl{p[\eta]}$:
\begin{equation*}
\begin{aligned}[t]
  \cl[\eta]{\never} &= \emptyset\\
  \cl[\eta]{\watch\,l} &= \set{l}\\
  \cl[\eta]{p \appAE q} &= \cl[\eta] q
\end{aligned}
\hspace{1cm}
\begin{aligned}[t]
  \cl[\eta]{\wait\,\kappa} &= \set{\kappa}\\
  \cl[\eta]{\select\; p\; q} &= \cl[\eta]{p} \cup \cl[\eta]{q}\\
  \cl[\eta]{\tail\,l} &= \cl[\eta_1]{q} \quad\text{ if } \eta = \eta_1,l \sigtype A \mapsto \sig p q, \eta_2
\end{aligned}
\end{equation*}
During the execution of a reactive step
$\state{p;\eta\sep\Delta}\fullstep{\kappa \mapsto t}
\state{p;\eta'\sep\Delta'}$, the machine maintains the following
invariant for any environment $\eta_N\tick\eta_E\sep\Delta''$
throughout the update process:
\begin{equation}
\text{
  $\ticked{\eta_N}{q}$ iff $\kappa \in \cl[\eta]{q}$ or there is some $l \in
  \cl[\eta]{q}$ with $\eta_N(l)= \sig[\updated]{\interm_1\, q_1}{q_2}$.}
  \label{eq:ticked-clock-correspondence}
\end{equation}
That is, the clock of a delayed computation $q$ ticks in response to
an event on channel $\kappa$ (and therefore can be advanced) iff the
clock contains $\kappa$ or a reference $l$ to a  partial signal that
has produced a fresh value during this reactive step. Note that the
relevant clock $\cl[\eta]{q}$ is with respect to $\eta$, i.e.\ the
heap from \emph{before} the start of the reactive step, because the
timing information in $\eta_N$ is for the \emph{next} step. We state
and prove the above observations more formally in
section~\ref{sec:ticked-clock-correspondence}.

Whenever the clock $\cl[\eta]{p_2}$ of a stored signal
$\sig{p_1}{p_2}$ ticks, the tail $p_2$ is advanced to produce the new
state of the signal. In general, the advance semantics takes a delayed
computation $p:\DelayE A$ whose clock has ticked, i.e.\
$\ticked{\eta_N}{p}$, and performs the delayed computation to obtain a
value $q$ of type $A$, denoted $\evalAdv[\kappa\mapsto
t]{p}{\eta_N\tick\eta_E\sep\Delta}{q}{\eta'_N\tick\eta_E\sep\Delta'}$.
Similarly to the evaluation semantics, the advance semantics may
allocate new signals on the \now heap and create fresh channels, but
it leaves the \earlier heap unchanged. In the cases for \ensuremath{\Varid{p}\mathrel{=}\watch\;\Varid{l}}
and \ensuremath{\Varid{p}\mathrel{=}\wait\;\kappa}, the advance semantics can simply look up the
desired terms in $\eta_N$ and $\kappa \mapsto t$, respectively,
because we know that $\ticked{\eta_N}{p}$ holds. In the case of $\ensuremath{\Varid{p}\mathrel{=}\select\;\Varid{p}_{\mathrm{1}}\;\Varid{p}_{\mathrm{2}}}$, the advance semantics consults
$\ticked{\eta_N}{\cdot}$ to decide which of the two delayed
computations to perform. Since $\ticked{\eta_N}{p}$ holds, we know
that at least one of the two has ticked. In the case of \ensuremath{\Varid{p}\mathrel{=}\delay\;\Varid{t}\appAE\Varid{p}_{\mathrm{1}}}, the advance semantics performs the delayed computation $p_1$
to obtain a value $p_2$ and then evaluates the function application \ensuremath{\Varid{t}\;\Varid{p}_{\mathrm{2}}}. Finally, in the case of \ensuremath{\Varid{p}\mathrel{=}\tail\;\Varid{l}}, the advance semantics
simply produces \ensuremath{\Varid{l}}, because it can rely on the fact that $l$ points
into the \now heap, which the update semantics has already updated.

We define the reactive evaluation semantics
$\state{t;\Delta}\reacts{\tau} \state{v;\Delta_n}$ in
Fig.~\ref{fig:reactiveSemantics} as the result of the machine reacting
to the events in $\tau$ where each event $e_i$ in $\tau$ is well-typed
according to the judgement $\isev[\Delta_i]{e_i}$ also defined in
Fig.~\ref{fig:reactiveSemantics}. The resulting machine value $p$
together with the heap $\eta_n$ represent the value $v = p[\eta_n]$,
which we obtain by taking each stored signal $l : \Sig\,A \mapsto \sig
{p_1} {p_2}$ in $\eta_n$ and replacing $\tail\,l$ with $p_2$ and $l$
with $p_1 ::_A p_2$ in $p$.

\subsection{Practical Considerations}
\label{sec:practical-considerations}

While the operational semantics provides a precise account of the
operational behaviour of \Rizzo and thus allows us to give formal
statements of the operational guarantees provided by the type system,
it does not try to give an \emph{efficient} implementation strategy.
How to efficiently implement algebraic concepts like variable
substitution and recursive data types is well studied in the
literature. In the following, we focus on how the semantics of signals
may be implemented efficiently.

Instead of two separate heaps to store signals, an implementation could use a
single, global heap for \now signals, \earlier signals, and any
other data types that need heap allocation such as product, sum and
recursive types. The \now and \earlier heaps are instead jointly
represented by a linked list data structure by making each signal
store a reference to its two neighbours. The $\tick$ divider between
the \now and \earlier heap is then represented by a reference into
the linked list.

In a practical implementation, we have some memory management strategy
that deallocates unused data from the global heap -- including stored
signals. However, signals stored on the heap take up space \emph{and
time}, because each stored signal is checked by the update semantics
and updated if necessary. It is thus important that unused memory is
freed in a timely fashion so that as few unused signals as possible
are updated. That makes \emph{automatic reference counting}, which has
seen renewed attention for efficient implementations of functional
languages~\citep{ullrich21CountingImmutableBeans,reinking21PerceusGarbageFree,lorenzen22ReferenceCountingFrame},
an ideal memory management strategy for \Rizzo. It frees unused memory
immediately, and thus no unused signals are ever updated. Moreover,
since the type system rules out circularity (cf.\
Theorem~\ref{thm:productivity} in section~\ref{sec:results}), programs
cannot produce circular reference chains, which automatic reference
counting would not be able to detect.

Finally, there are further small optimisations that a practical
implementation could perform in order to reduce the number of
allocations of stored signals. For example, when the update semantics updates a
signal at location $l$, it advances the tail of that signal. This will
produce a new signal at $l'$, whose data is copied into $l$. An
efficient implementation can avoid the additional allocation of $l'$
and the subsequent copying by simply writing the new signal directly
into $l$ to begin with. This is safe, because we can rule out that the process of
advancing the tail of $l$ will itself read from $l$, due to the simple
fact that $l$ resides in the \earlier heap.

\subsection{Example Run of the Machine}
\label{sec:opsem-examples}

To illustrate the machine, we show how it executes a small example
program involving \ensuremath{\Varid{sample}}. This and two more examples are developed
in more detail in Appendix~\ref{sec:more-examples} and the
accompanying Lean formalisation (both found in the supplementary
material).

Our example program \ensuremath{\Varid{t}} assumes a channel context $\Delta = \set{\kappa_1 :
\ensuremath{\Chan\;\Nat}, \kappa_2 : \ensuremath{\Chan\;\Conid{Char}}}$:
\[
t = \ensuremath{\mathbf{let}\;\Varid{xs}\mathrel{=}\mathrm{0}\mathbin{::}\Varid{mkSig}\;(\wait\;\kappa_1)\;\mathbf{in}\;\mathbf{let}\;\Varid{ys}\mathrel{=}\text{\ttfamily 'a'}\mathbin{::}\Varid{mkSig}\;(\wait\;\kappa_2)\;\mathbf{in}\;\Varid{sample}\;\Varid{xs}\;\Varid{ys}}
\]
This program first constructs two signals $\ensuremath{\Varid{xs}}$ and $\ensuremath{\Varid{ys}}$ from the
two channels $\kappa_1$ and $\kappa_2$ and then samples from $\ensuremath{\Varid{ys}}$
using $\ensuremath{\Varid{xs}}$.
To avoid clutter, we elide heap locations that are not referenced
anywhere, leave out the type annotations of all heap locations, and
do not spell out the values $p_i$:
\begin{eqnarray*}
  \state{t;\Delta} &\initstep 
  &\state{l_3;l_1 \mapsto \sig[\unchanged]{0}{p_1},
  l_2 \mapsto \sig[\unchanged]{\ensuremath{\text{\ttfamily 'a'}}}{p_2},
  l_3 \mapsto \sig[\unchanged]{(0,\ensuremath{\text{\ttfamily 'a'}})}{p_3}\sep\Delta}\\
  &\fullstep{\kappa_1 \mapsto 1} 
  &\state{l_3;l_1 \mapsto \sig[\updated]{1}{p_1},
  l_2 \mapsto \sig[\unchanged]{\ensuremath{\text{\ttfamily 'a'}}}{p_2},
  l_3 \mapsto \sig[\updated]{(1,\ensuremath{\text{\ttfamily 'a'}})}{p_3}\sep\Delta}\\
  &\fullstep{\kappa_2 \mapsto \ensuremath{\text{\ttfamily 'b'}}} 
  &\state{l_3;l_1 \mapsto \sig[\unchanged]{1}{p_1},
  l_2 \mapsto \sig[\updated]{\ensuremath{\text{\ttfamily 'b'}}}{p_2},
  l_3 \mapsto \sig[\unchanged]{(1,\ensuremath{\text{\ttfamily 'a'}})}{p_3}\sep\Delta}\\
  &\fullstep{\kappa_1 \mapsto 2} 
  &\state{l_3;l_1 \mapsto \sig[\updated]{2}{p_1},
  l_2 \mapsto \sig[\unchanged]{\ensuremath{\text{\ttfamily 'b'}}}{p_2},
  l_3 \mapsto \sig[\updated]{(2,\ensuremath{\text{\ttfamily 'b'}})}{p_3}\sep\Delta}
\end{eqnarray*}
The initialisation step allocates three signals on the heap. Each of
the first two, at $l_1$ and $l_2$, updates whenever the corresponding
channel $\kappa_i$ receives an input since $\cl[\eta]{\ensuremath{\tail\;\Varid{l}_{\Varid{i}}}} =
\cl[\eta]{p_i} = \cl[\eta]{\ensuremath{\wait\;\kappa_i}} = \set{\kappa_i}$. The
signal at $l_3$, however, only updates whenever $\kappa_1$ receives an
input since $\cl[\eta]{\ensuremath{\tail\;\Varid{l}_{\mathrm{3}}}} = \cl[\eta]{p_3} = \cl[\eta]{\ensuremath{\tail\;\Varid{l}_{\mathrm{1}}}} = \set{\kappa_1}$. That is, given $\tau = \kappa_1 \mapsto
1,\kappa_2 \mapsto \ensuremath{\text{\ttfamily 'b'}},\kappa_1 \mapsto 2$, we have the reactive
evaluation $\state{t;\Delta}\reacts{\tau} \state{\ensuremath{(\mathrm{2},\text{\ttfamily 'b'})\mathbin{::}\Varid{v}};\Delta}$ for some \ensuremath{\Varid{v}\mathbin{:}\DelayE\;(\Sig\;(\Nat\times\Conid{Char}))}.

As mentioned above, we have elided heap locations that are
not referenced anywhere. For example, after the first reactive step
$\fullstep{\kappa_1 \mapsto 1}$, the machine also allocates a heap
location $l_4$ with the same value as $l_1$ and inserts it just to the
left of $l_1$. Similarly, the machine allocates a location $l_5$
just to the left of $l_3$ with the same value as $l_3$. In a practical
implementation, these intermediate heap locations can be garbage
collected or avoided altogether as described in
section~\ref{sec:practical-considerations}.

\subsection{Extended Type System and Type Preservation}
\label{sec:extended-type}

As discussed in section~\ref{sec:evaluation-semantics}, the evaluation
semantics creates fresh channels in the channel context $\Delta$ and
allocates new signals in the \now heap $\eta_N$. To give a precise account
of the type preservation property satisfied by \Rizzo, we generalise
the typing judgement $\hastype\Gamma t A$ to an extended typing
judgement of the form $\hastype*\Gamma t A$ so that it also has a
\emph{heap context} $H$. The latter is a sequence of heap location
typings of the form $l \sigtype A$, each of which indicates that there
is a signal of type $\Sig\,A$ at heap location $l$. All typing rules
from Fig.~\ref{fig:typing} carry over to the extended typing judgement
with the heap context $H$ being the same in premises and conclusions
(just like $\Delta$). For example, the rule for lambda abstractions is
as follows in the extended system:
\[
  \inferrule*%
  {\hastype*{\Gamma,x:A}{t}{B}}%
  {\hastype*{\Gamma}{\lambda x.t}{A \to B}}%
\]
In addition, we have one new typing rule that simply considers heap
locations from $H$ as well-typed:
\[
  \inferrule*%
  {l \sigtype A \in H}%
  {\hastype*{\Gamma}{l}{\Sig\,A}}%
\]
This extended type system is only needed to formulate and to prove the
operational properties of \Rizzo. Programs are still expected to be
typed according to the ``surface'' type system from
Fig.~\ref{fig:typing}. However, this ``surface'' type system is a
special case of the extended type system where $H$ is empty.

The evaluation semantics preserves typing with respect to the
extended typing judgement. In order to state this precisely, we define the
heap context $\heaptype\eta$ of a heap $\eta$ as the sequence of heap
locations and their typing from $\eta$, i.e.\
$\heaptype{\cdot} = \cdot$ and $\heaptype{\eta, l \sigtype A \mapsto \sig p q} = \heaptype{\eta},l \sigtype A$,
where we write $\cdot$ for the empty heap.

Finally, to precisely state the operational guarantees of \Rizzo, we
also give typing judgements for \now and \earlier heaps
($\isheap{\eta}$ and $\isearlier{H}{\eta}$) as well as environments
($\isenv{\varepsilon}$) in Fig.~\ref{fig:isheap}. In addition, we also
define the shorthand $\hastype*[\eta]{}{t}{A}$, which states that both
$t$ and $\eta$ are well-typed.

\begin{figure}
  \small
  \begin{mathpar}
    \inferrule*{~}{\isheap{\cdot}}
    \and
    \inferrule*
    {\isheap{\eta}\\ \hastype*[{\heaptype{\eta}}]{}{p}{A}\\ \hastype*[{\heaptype{\eta}}]{}{q}{\DelayE (\Sig A)}}
    {\isheap{\eta,l \sigtype A \mapsto \sig{p}{q}}}
  \\
    \inferrule*{~}{\isearlier{H}{\cdot}}
    \and
    \inferrule*
    {\isearlier{H,l \sigtype A}{\eta}\\ \hastype*[H]{}{p}{A}\\ \hastype*[H]{}{q}{\DelayE (\Sig A)}}
    {\isearlier{H}{l \sigtype A \mapsto \sig{p}{q},\eta}}
    \\
    \inferrule*
    {\hastype*[\heaptype{\eta}]{}{t}{A} \\ \isheap{\eta}}
    {\hastype*[\eta]{}{t}{A}}
    \and
    \inferrule*
    {\isheap{\eta_N}\\\isearlier{\heaptype{\eta_N}}{\eta_E}}
    {\isenv{\eta_N\tick\eta_E\sep\Delta}}
  \end{mathpar}
  \vspace{-1.5em}
  \caption{Well-typed \now heaps ($\isheap{\eta}$), \earlier heaps
  ($\isearlier{H}{\eta}$), and environments ($\isenv{\varepsilon}$).}
  \label{fig:isheap}
\end{figure}

\subsection{Main Metatheoretical Results}
\label{sec:results}

Using the operational semantics, we can now give a precise account of
the operational guarantees that \Rizzo's type system provides, namely
productivity, causality, and the absence of space leaks.

\subsubsection{Productivity}
Recall the definition of the well-typed event judgement
$\isev[\Delta]{e}$ from Fig.~\ref{fig:reactiveSemantics}. A
\emph{well-formed reactive sequence} is of the form:
\begin{equation}
\state{t;\Delta} \initstep
\state{p;\eta_0\sep\Delta_0} \fullstep{e_0} \state{p;\eta_1\sep\Delta_1} 
\fullstep{e_1} \dots \fullstep{e_{n-1}} \state{p;\eta_n\sep\Delta_n}
\text{ with $\isev[\Delta_i]{e_i}$ for $0 \le i < n$}
\label{eq:wf-sequence}
\end{equation}
Productivity states that the machine can always extend such a sequence
from a well-typed term:
\begin{theorem}[productivity]
  \label{thm:productivity}
  Let $\hastype{}{t}{A}$ be a well-typed term.
  \begin{enumerate}[(i)]
    \item There is a step $\state{t;\Delta} \initstep \state{p;\eta_0\sep\Delta_0}$.
    \label{item:productivityI}
    \item Given a well-formed reactive sequence of the form
  \eqref{eq:wf-sequence} and an event $\isev[\Delta_n] {e_n}$, there
  is a step $\state{p;\eta_n\sep\Delta_n}
  \fullstep{e_n}
  \state{p;\eta_{n+1}\sep\Delta_{n+1}}$.
    \label{item:productivityII}
  \end{enumerate}
\end{theorem}
\noindent
This means that, given a well-typed term $\hastype*[][\Delta]{}{t}{A}$ and
an infinite sequence of well-typed events
$(\isev[\Delta_i]e_i)_{i\in\nats}$, we obtain an
infinite well-formed reactive sequence
\begin{equation}
\state{t;\Delta} \initstep
\state{p;\eta_0\sep\Delta_0} \fullstep{e_0} \state{p;\eta_1\sep\Delta_1} 
\fullstep{e_1} \dots 
\label{eq:inf-sequence}
\end{equation}

Moreover, the output produced by this reactive sequence,
namely the value $p$ along with any signals it may refer to in the
heaps $\eta_i$, is well-typed:
\begin{theorem}[type preservation]
  \label{thm:preservation}
  Given a well-typed term $\hastype*[][\Delta]{}{t}{A}$ and an
  infinite well-formed reactive sequence of the form
  \eqref{eq:inf-sequence}, we have
  $\hastype*[\eta_i][\Delta_i]{}{p}{A}$ for all $i \ge 0$.
\end{theorem}
\noindent
Here we use the judgement $\hastype*[\eta_i][\Delta_i]{}{p}{A}$ from
Fig.~\ref{fig:isheap} stating that both $p$ and $\eta_i$ are well-typed.

We can state the above productivity and the type preservation
properties in terms of the reactive evaluation relation
$\reacts{\tau}$. The productivity property follows from
Theorem~\ref{thm:productivity} and the fact that $\state{t;\Delta}
\reacts{\tau} \state{t';\Delta'}$ implies that $t'$ is a value
according to the grammar in Fig.~\ref{fig:syntax}:
\begin{corollary}[productivity]
  \label{cor:productivity}
  Let $\hastype[\Delta]{}{t}{A}$ be a well-typed term.
  \begin{enumerate}[(i)]
    \item There is a reactive evaluation $\state{t;\Delta} \reacts{\cdot}
    \state{v;\Delta'}$, where $\cdot$ is the empty sequence of events.
    \item If $\state{t;\Delta} \reacts{\tau} \state{v;\Delta'}$ and
  $\isev[\Delta'] {e}$, then there is a reactive evaluation $\state{t;\Delta} \reacts{\tau,e}
  \state{w;\Delta''}$.
  \end{enumerate} 
\end{corollary}
In turn, type preservation follows from
Theorem~\ref{thm:preservation} and the fact that
$\hastype*[\eta][\Delta]{}{p}{A}$ implies
$\hastype[\Delta]{}{p[\eta]}{A}$:
\begin{corollary}[type preservation]
  \label{cor:preservation}
  If $\hastype[\Delta]{}{t}{A}$ and $\state{t;\Delta} \reacts{\tau}
  \state{v;\Delta'}$, then $\hastype[\Delta']{}{v}{A}$.
\end{corollary}

\subsubsection{Causality}
\label{sec:causality}

A \Rizzo term $\hastype*[][\Delta]{}{t}{A}$ is called \emph{causal} if, for any
infinite well-formed reactive sequence of the form
\eqref{eq:inf-sequence}, each state $\state{p;\eta_n\sep\Delta_n}$
in the sequence only depends on the initial state $\state{t;\Delta}$ and all events
$e_i$ with $i<n$.
\begin{theorem}[causality]
  \label{thm:causality}
  Suppose the infinite well-formed reactive sequences
  \eqref{eq:inf-sequence} and 
  \[
    \state{t;\Delta} \initstep
    \state{p';\eta'_0\sep\Delta'_0} \fullstep{e'_0} \state{p';\eta'_1\sep\Delta'_1} 
    \fullstep{e'_1} \dots 
  \]
  Let $n \in \nats$ and suppose $e_i = e'_i$ for all $i < n$. Then
  $\state{p;\eta_n\sep\Delta_n} = \state{p';\eta'_n\sep\Delta'_n}$.
\end{theorem}
\subsubsection{No Space Leaks}

The absence of space leaks is a direct consequence of the
productivity property in Theorem~\ref{thm:productivity}: The
operational semantics is by definition free of space leaks since after
each step $\state{p;\eta\sep\Delta} \fullstep{e}
\state{p;\eta'\sep\Delta'}$, all old signals have been overwritten
with their new value. Hence, the program cannot retain old input data.
Theorem~\ref{thm:productivity} shows that this in-place update
semantics is safe, and that the program never tries to dereference old
data. The old data is kept in the \earlier heap and any attempt by the
program to dereference a signal from that heap would result in a stuck
execution, which Theorem~\ref{thm:productivity} rules out.

\section{Metatheory}
\label{sec:metatheory}
In this section, we give an overview of how we have proved the main
results presented in section~\ref{sec:results} as well as the timing
correspondence \eqref{eq:ticked-clock-correspondence}. The complete
formalisation of these proofs in Lean can be found in the \ifanon%
supplementary material to this article.
\else%
supplementary material to this article~\citep{supplementary}.
\fi%
The three main theorems we prove
are Theorem~\ref{thm:productivity} (productivity),
Theorem~\ref{thm:preservation} (type preservation), and
Theorem~\ref{thm:causality} (causality), while the absence of space
leaks follows by definition of the operational semantics and
Theorem~\ref{thm:productivity}.

\subsection{Causality}

The causality property of Theorem~\ref{thm:causality} is an immediate
consequence of the fact that the operational semantics is
deterministic in each of its components, in particular the step
semantics:
\begin{lemma}[determinism]~
  \label{lem:deterministic}
  \begin{enumerate}[(i)]
    \item If $\state{t;\Delta}\initstep\state{p_1;\eta_1\sep\Delta_1}$
    and $\state{t;\Delta}\initstep\state{p_2;\eta_2\sep\Delta_2}$,
    then $\state{p_1;\eta_1\sep\Delta_1} =
    \state{p_2;\eta_2\sep\Delta_2}$.
    \label{item:deterministicI}
    \item If $\state{p;\eta\sep\Delta}\fullstep{e}\state{p_1;\eta_1\sep\Delta_1}$
    and $\state{p;\eta\sep\Delta}\fullstep{e}\state{p_2;\eta_2\sep\Delta_2}$,
    then $\state{p_1;\eta_1\sep\Delta_1} =
    \state{p_2;\eta_2\sep\Delta_2}$.
    \label{item:deterministicII}
  \end{enumerate}
\end{lemma}
\noindent
This determinism property is proved by showing that all parts of the
operational semantics (evaluation, advance, update, and step
semantics) are deterministic using a straightforward argument by
structural induction on the rules of the operational semantics.


\subsection{Productivity \& Type Preservation}
\label{sec:productivity-preservation}
Theorem~\ref{thm:productivity} (productivity) and
Theorem~\ref{thm:preservation} (type preservation) are a consequence
of two standard properties~\citep{wright94SyntacticApproachType}:
\emph{progress}, i.e.\ given a well-typed starting point, the
operational semantics produces a result, and \emph{preservation},
i.e.\ the produced result is itself well-typed again. However, instead
of a syntactic proof, we use a logical relations argument for the
progress property.

In order to make the progress and preservation properties precise, we
have introduced the extended typing judgement
$\hastype*{\Gamma}{t}{A}$ in section~\ref{sec:extended-type} along
with typing judgements for heaps and environments in
Fig.~\ref{fig:isheap}. Thus, we
can formally describe the type preservation properties as follows:
\begin{proposition}[type preservation]
  \label{prop:preservation}
  ~
  \begin{enumerate}[(i)]
  \item If $\hastype*[\eta_N]{}{t}{A}$ and
    $\heval{t}{\eta_N\tick\eta_E\sep\Delta}{p}{\eta'_N\tick\eta'_E\sep\Delta'}$,
    then $\hastype*[\eta'_N][\Delta']{}{p}{A}$.
  \label{item:preservationI}
  \item If $\hastype*[\eta_N]{}{p}{\DelayE A}$, $\isev e$,  and
    $\evalAdv[e]{p}{\eta_N\tick\eta_E\sep\Delta}{q}{\eta'_N\tick\eta'_E\sep\Delta'}$,
    then $\hastype*[\eta'_N][\Delta']{}{q}{A}$.
  \label{item:preservationII}
  \item If $\isenv{\varepsilon}$, $\isev e$, and
  $\state{\varepsilon}\partstep{e}
  \state{\varepsilon'}$, then $\isenv{\varepsilon'}$.
   \label{item:preservationIII}

  \item If $\hastype*[][\Delta]{}{t}{A}$
  and $\state{t;\Delta}\initstep
  \state{p;\eta\sep\Delta'}$, then  $\hastype*[\eta][\Delta']{}{p}{A}$.
  \label{item:preservationIV}
  \item If $\isheap{\eta}$, $\isev e$,
  and $\state{p;\eta\sep\Delta}\fullstep{e}
  \state{p;\eta'\sep\Delta'}$, then  $\isheap[\Delta']{\eta'}$.
  \label{item:preservationV}
  \end{enumerate}
\end{proposition}
All five type preservation properties can be proved by structural
induction on the definition of the operational semantics. To obtain
Theorem~\ref{thm:preservation}, we need two additional properties:
\begin{lemma}[weakening]
  \label{lem:weakening}
  If $\hastype*{\Gamma}{t}{A}$, $H \subseteq
  H'$, and $\Delta \subseteq \Delta'$, then $\hastype*[H'][\Delta']{\Gamma}{t}{A}$.
\end{lemma}%
\begin{lemma}
  \label{lem:step-incr}
  If $\isheap{\eta}$, $\isev e$, and $\state{p;\eta\sep\Delta}\fullstep{e}
  \state{p;\eta'\sep\Delta'}$, then $\heaptype{\eta} \subseteq
  \heaptype{\eta'}$ and $\Delta \subseteq \Delta'$.
\end{lemma}
\noindent
Lemma~\ref{lem:weakening} states the weakening of the term typing
judgement and follows by an induction on the derivation of
$\hastype*{\Gamma}{t}{A}$, while Lemma~\ref{lem:step-incr} states that
the reactive step semantics grows the heap context and the channel context, which
follows by an induction on $\state{p;\eta\sep\Delta}\fullstep{e}
\state{p;\eta'\sep\Delta'}$.

\begin{proof}[Proof of Theorem~\ref{thm:preservation} (type preservation)]
We prove $\hastype*[\eta_i][\Delta_i]{}{p}{A}$ by induction on $i$. If
$i = 0$, then $\hastype*[\eta_0][\Delta_0]{}{p}{A}$ follows by
Proposition~\ref{prop:preservation}~\ref{item:preservationIV}. If $i =
j + 1$, then we have $\hastype*[\eta_j][\Delta_j]{}{p}{A}$ by
induction hypothesis, which means
$\hastype*[\heaptype{\eta_j}][\Delta_j]{}{p}{A}$ and
$\isheap[\Delta_j]{\eta_j}$. By Lemma~\ref{lem:step-incr}, we have
$\heaptype{\eta_j} \subseteq \heaptype{\eta_i}$ and $\Delta_j
\subseteq \Delta_i$. Thus
$\hastype*[\heaptype{\eta_i}][\Delta_i]{}{p}{A}$ follows by
Lemma~\ref{lem:weakening}, and $\isheap[\Delta_i]{\eta_i}$ follows by
Proposition~\ref{prop:preservation}~\ref{item:preservationV}.
\end{proof}

For each preservation property in Proposition~\ref{prop:preservation},
we have a corresponding progress property, which states that the
computation step in the preservation property exists:
\begin{proposition}[progress]
  \label{prop:progress}
  ~
  \begin{enumerate}[(i)]
  \item If $\hastype*[\eta_N]{}{t}{A}$, then there are $p,\eta'_N,\eta'_E, \Delta'$ with 
    $\heval{t}{\eta_N\tick\eta_E\sep\Delta}{p}{\eta'_N\tick\eta'_E\sep\Delta'}$.
  \label{item:progressI}
  \item If $\hastype*[\eta_N]{}{p}{\DelayE A}$, $\isev {\kappa \mapsto t}$, and $\ticked{\eta_N}{p}$, then there are
    $q,\eta'_N,\eta'_E, \Delta'$ with $\evalAdv[\kappa\mapsto
    t]{p}{\eta_N\tick\eta_E\sep\Delta}{q}{\eta'_N\tick\eta'_E\sep\Delta'}$.
  \label{item:progressII}
  \item If $\isenv{\eta_N\tick\eta_E\sep\Delta}$, $\isev e$, and $\eta_E$ is non-empty, then there are
  $\sigma', \Delta'$ with
  $\state{\eta_N\tick\eta_E\sep\Delta}\partstep{e}
  \state{\sigma'\sep\Delta'}$.
  \label{item:progressIII}
  \item If $\hastype*[][\Delta]{}{t}{A}$, then
  there are $p, \eta, \Delta'$ with $\state{t;\Delta}\initstep
  \state{p;\eta\sep\Delta'}$.
  \label{item:progressIV}
  \item If $\isheap{\eta}$, $\isev e$,
  then there are $\eta', \Delta'$ with $\state{p;\eta\sep\Delta}\fullstep{e}
  \state{p;\eta'\sep\Delta'}$.
  \label{item:progressV}
  \end{enumerate}
\end{proposition}
\noindent
The proof of \ref{item:progressI} uses a logical relations argument,
which we sketch in section~\ref{sec:logical-relation} below, whereas
\ref{item:progressII} follows by induction on the machine value of type
$\DelayE A$, appealing to \ref{item:progressI} in the cases for \ensuremath{\wait}
and \ensuremath{\appAE}. The remaining properties \ref{item:progressIII} to
\ref{item:progressV} follow from \ref{item:progressI},
\ref{item:progressII}, and the corresponding type preservation
properties in Proposition~\ref{prop:preservation}.

\begin{proof}[Proof of Theorem~\ref{thm:productivity} (productivity)]
Part \ref{item:productivityI} follows immediately from
Proposition~\ref{prop:progress}~\ref{item:progressIV}.
To show \ref{item:productivityII}, assume a well-formed reactive
sequence of the form \eqref{eq:wf-sequence}. By induction on the
length $n$ of the sequence, we can show that
$\isheap[\Delta_n]{\eta_n}$ using
Proposition~\ref{prop:preservation}~\ref{item:preservationIV}-\ref{item:preservationV}.
We can then apply
Proposition~\ref{prop:progress}~\ref{item:progressV} to obtain the
desired step $\state{p;\eta_n\sep\Delta_n} \fullstep{e_n}
\state{p;\eta_{n+1}\sep\Delta_{n+1}}$.
\end{proof}

\subsection{Logical Relation}
\label{sec:logical-relation}

\begin{figure}[t]
  \small
  \begin{align*}
  \vinterp{1}{\varepsilon}
  &= \set{\unit},
  \\
  \vinterp{F \times G}{\varepsilon}
  &= \setcom{\pair{p_1}{p_2}}{p_1 \in \vinterp{F}{\varepsilon}
    \land p_2 \in \vinterp{G}{\varepsilon}},
  \\
  \vinterp{F + G}{\varepsilon}
  &= \setcom{\interm_1 \, p}{p \in
    \vinterp{F}{\varepsilon}} \cup
  \setcom{\interm_2 \, p}{p \in
    \vinterp{G}{\varepsilon}}
  \\
  \vinterp{\DelayA A}{\varepsilon} &=
  \setcom{p}{\Hastype{\varepsilon}{p}{\DelayA A}}
  \qquad\text{ where }\quad\Hastype{\eta_N\tick\eta_E\sep\Delta}{t}{A} \quad\text{iff}\quad \hastype*[\heaptype{\eta_N}]{}{t}{A}\\
  \vinterp{\DelayE A}{\varepsilon} &=
  \setcom{p}{\Hastype{\varepsilon}{p}{\DelayE A}}\\
  \vinterp{A \to F}{\varepsilon}
  &= \setcom{ \lambda x.t}{\Hastype{\varepsilon}{\lambda x.t}{A \to F\rhosubs} \land
    \forall \varepsilon' \heapge \varepsilon.p \in \vinterp[]{A}{\varepsilon'}. 
    t[p/x] \in \tinterp{F}{\varepsilon'}}
    \\
  \vinterp{\Chan\,A}{\sigma\sep\Delta}
  &= \setcom{\kappa}{\kappa \chantype A \in \Delta}
  \\
  \vinterp{\Sig\, F}{\eta_N\tick\eta_E\sep\Delta}
  &= \setcom{l}{l \sigtype F\rhosubs \mapsto \sig p q \in \eta_N, p \in \vinterp{F}{\eta_N\tick\eta_E\sep\Delta}}
  \\
  \vinterp{\mu\alpha. F}{\varepsilon} &= \lfp\bigl(\Psi^{\mu\alpha.F}_\rho\bigr)(\varepsilon),
  \quad
  \begin{aligned}[t]
    \text{where}\;
  \lfp(\Psi) &= \bigcap \setcom{X}{X \text{ is Kripke monotone and } \Psi(X) \subseteq X}\\
        \Psi^{\mu\alpha.F}_\rho(X)(\varepsilon)
  &= \setcom{\cons[{(\mu\alpha.F)\rhosubs}]\,p}{p \in \vinterp[{\rho[\alpha\mapsto (X,(\mu\alpha.F)\rhosubs)]}]{F}{\varepsilon}}
  \end{aligned}
\\
  \vinterp{\alpha}{\varepsilon} &=\setcom{p \in X(\varepsilon)}{\Hastype{\varepsilon}{p}{A}} \quad \text{ if } \rho(\alpha) = (X,A)
  \\[1em]
  \vinterp[]{A}{\varepsilon} &= \vinterp[\emptyset]{A}{\varepsilon}
  \hspace{2cm} 
  \tinterp{F}{\varepsilon}
  = \setcom{t}{
    \exists p, \varepsilon'
    . \heval{t}{\varepsilon}{p}{\varepsilon'}
    \land p \in \vinterp{F}{\varepsilon'}}
  \\[1em]
      \hinterp{\varepsilon}  &= 
     \setcom{\eta}{\forall l \sigtype A \mapsto \sig{p}{q} \in \eta. p \in \vinterp[]{A}{\varepsilon}}   
  \\[1em]
    \cinterp{\cdot}{\varepsilon}  &= \set{\ast}
  \hspace{2.3cm}
  \cinterp{\Gamma,x:A}{\varepsilon}
  = \setcom{\gamma[x \mapsto p]}
  {\gamma \in \cinterp{\Gamma}{\varepsilon},p
    \in\vinterp[]{A}{\varepsilon}}
\end{align*}
  \vspace{-1.5em}
\caption{Logical relation. Recall that $A,B$ range over closed types, and $F,G$ range over possibly open types.}
  \label{fig:log-rel}
\end{figure}

We sketch the logical relations argument that underpins the proof of
Proposition~\ref{prop:progress}~\ref{item:progressI}. This proof
proceeds by constructing for each closed type $A$ and each environment
$\varepsilon$ a set of terms $\tinterp[]{A}{\varepsilon}$ that are
semantically of type $A$ in the context of the environment
$\varepsilon$. By construction, each term in
$\tinterp[]{A}{\varepsilon}$ evaluates to a value.
Proposition~\ref{prop:progress}~\ref{item:progressI} then follows from
the fundamental property of the logical relation
$\tinterp[]{A}{\varepsilon}$, which states that all well-typed terms
of type $A$ are in $\tinterp[]{A}{\varepsilon}$.

The proof of the fundamental property relies on the fact that the
logical relation satisfies \emph{Kripke monotonicity}, which means
that it is closed under a suitable ordering $\heaple$ on
$\varepsilon$. This ordering captures the fact that the evaluation
semantics may allocate additional signals and channels. On heaps and
on channel contexts, $\heaple$ is defined as the subsequence relation.
That is, $\Delta \heaple \Delta'$ iff $\Delta$ can be obtained from
$\Delta'$ by removing some of its elements. On heaps, $\heaple$ is
defined analogously. We lift this ordering to environments by
pointwise ordering on the \now heap and the channel context
components:
\[
\eta_N\tick\eta_E\sep\Delta \heaple \eta'_N\tick\eta'_E\sep\Delta' \qquad \text{iff} \qquad
\eta_N \heaple \eta'_N \land \Delta \heaple \Delta' \land \eta_E = \eta'_E
\]
Note that the \earlier heap remains constant. This matches the
following property of the evaluation semantics, which can be proved by a
straightforward induction on the evaluation semantics:
\begin{lemma}
If $\heval{t}{\varepsilon}{p}{\varepsilon'}$, then $\varepsilon
\heaple \varepsilon'$.
\end{lemma}

Fig.~\ref{fig:log-rel} defines the logical relations used for proving
the progress property. In addition to the term relation
$\tinterp{F}{\varepsilon}$, we also define a corresponding value
relation $\vinterp{F}{\varepsilon}$. Both $\tinterp{F}{\varepsilon}$
and $\vinterp{F}{\varepsilon}$ are defined more generally for open
types $F$, i.e.\ $F$ may contain free type variables. To
account for this, both logical relations have the additional component
$\rho$, which is a finite map from type variables to pairs $(X,A)$
consisting of a semantic type $X$ and a closed type $A$. More
precisely, $X$ is a function from environments to sets of values of
type $A$. In the definition of the value relation, we use the notation
$F[\rho]$ to apply the syntactic substitution $\rhosyn$ to a type $F$,
where $\rhosyn(\alpha) = A$ iff $\rho(\alpha)=(X,A)$. The definition
of $\vinterp{\mu\alpha. F}{\varepsilon}$ uses a standard least
pre-fixed point construction on semantic types, denoted by $\lfp$, on
a function $\Psi^{\mu\alpha.F}_\rho$. To simplify the proof, $\lfp$ is
limited to semantic types $X$ that are Kripke monotone, i.e.\
$X(\varepsilon) \subseteq X(\varepsilon')$ whenever $\varepsilon
\heaple \varepsilon'$. The least pre-fixed point is well-defined since
$\Psi^{\mu\alpha.F}_\rho$ is monotone, i.e.\
$\Psi^{\mu\alpha.F}_\rho(X)(\varepsilon) \subseteq
\Psi^{\mu\alpha.F}_\rho(Y)(\varepsilon)$ whenever $X(\varepsilon)
\subseteq Y(\varepsilon)$.

Since the proof of the fundamental property proceeds by induction on the
typing judgement, we also have to consider open
terms. To this end, we also define a corresponding context relation
$\cinterp{\Gamma}{\varepsilon}$, which contains term substitutions
$\gamma$ such that $\gamma(x) \in \vinterp[]{A}{\varepsilon}$ iff $x :
A \in \Gamma$; and a heap relation
$\hinterp{\varepsilon}$ to capture semantically well-typed \now heaps
with respect to an environment $\varepsilon$.

\begin{proposition}[fundamental property]
  \label{prop:fundamental}
  If $\hastype*[\heaptype{\eta_N}]{\Gamma}{t}{A}$,
  $\eta_N\tick\eta_E\sep\Delta \heaple \varepsilon$, $\eta_N \in
  \hinterp{\varepsilon}$, and $\gamma \in \cinterp{\Gamma}{\varepsilon}$,
  then $t\gamma \in \tinterp[]{A}{\varepsilon}$.
\end{proposition}
\noindent
The fundamental property is proved by a lengthy induction on
$\hastype*[\heaptype{\eta_N}]{\Gamma}{t}{A}$. The proof relies on the
fact that the value, context, and heap relations are Kripke monotone:
\begin{lemma}
  Let $\varepsilon \heaple \varepsilon'$. Then
  $\vinterp[]{A}{\varepsilon} \subseteq \vinterp[]{A}{\varepsilon'}$,
  $\cinterp{\Gamma}{\varepsilon} \subseteq \cinterp{\Gamma}{\varepsilon'}$, and
  $\hinterp{\varepsilon} \subseteq \hinterp{\varepsilon'}$
\end{lemma}


Proposition~\ref{prop:fundamental} is much more general than what we
need to prove Proposition~\ref{prop:progress}~\ref{item:progressI}, so
that the induction hypothesis is strong enough for the induction
argument to succeed.
Proposition~\ref{prop:progress}~\ref{item:progressI} follows
immediately from the following instantiation of
Proposition~\ref{prop:fundamental} to closed terms:
\begin{corollary}[fundamental property]
  If $\hastype*[\eta_N]{}{t}{A}$, then $t \in \tinterp[]{A}{\eta_N\tick\eta_E\sep\Delta}$.
\end{corollary}

\subsection{Timing of Delayed Computations}
\label{sec:ticked-clock-correspondence}
Finally, we make the correspondences between $\ticked{\eta}{q}$,
$\cl[\eta]{q}$, and $\cl{q[\eta]}$ precise:
\begin{proposition}[$\sym{ticked}$-$\sym{cl}$ correspondence]
  \label{prop:ticked-clock-correspondence}
  If $\state{\emptyheap\tick\eta\sep\Delta}\partstep*{\kappa \mapsto
  t} \state{\varepsilon}$ and
  $\heval{s}{\varepsilon}{p'}{\eta_N\tick\eta_E\sep\Delta''}$ with
  $\isheap{\eta}$, then the correspondence \eqref{eq:ticked-clock-correspondence} holds
  for any $\hastype*[{\heaptype{\eta_N}}]{}{q}{\DelayE A}$.
\end{proposition}
\noindent
That is, after any sequence of updates
$\state{\emptyheap\tick\eta\sep\Delta}\partstep*{\kappa \mapsto t}
\state{\varepsilon}$ and any further allocations in $\varepsilon$
performed by arbitrary evaluation
$\heval{s}{\varepsilon}{p'}{\eta_N\tick\eta_E\sep\Delta''}$, the
correspondence~\eqref{eq:ticked-clock-correspondence} between
$\ticked{\eta_N}{\cdot}$ and $\cl[\eta]{\cdot}$ holds. This
characterisation covers all environments
$\eta_N\tick\eta_E\sep\Delta''$ that the machine may encounter during
a reactive step $\state{p;\eta\sep\Delta}\fullstep{\kappa \mapsto t}
\state{p;\eta'\sep\Delta'}$. According to
Proposition~\ref{prop:preservation} the machine preserves the typing
of terms and heaps, and thus the above correspondence holds at any
time of the machine's execution.
Proposition~\ref{prop:ticked-clock-correspondence} confirms that the
machine faithfully implements the timing of delayed computations
captured by the notion of clocks. In addition, the two definitions of
clocks ($\cl[\eta]{p}$ and $\cl{v}$) agree for well-typed values and
heaps:
\begin{proposition}[clock correspondence]
  \label{prop:clock-correspondence}
  If $\hastype*[\eta]{}{p}{\DelayE A}$, then $(\cl[\eta]{p})[\eta] =
  \cl{p[\eta]}$.
\end{proposition}

\section{Related Work}
\label{sec:related-work}

\paragraph{Modal FRP}
The use of modal types for FRP was first proposed by
\citet{krishnaswami11UltrametricSemanticsReactive,krishnaswami11SemanticModelGraphical}, and its connection
to linear temporal logic was independently discovered by
\citet{jeltsch12CommonCategoricalSemantics} and
\citet{jeffrey12LTLTypesFRP}.
\citet{krishnaswami13HigherorderFunctionalReactive} was the first to
exploit modal types to construct a higher-order FRP language that
provably does not suffer from space leaks. This work was later
extended by \citet{bahr19SimplyRaTTFitchstyle} to simplify the type
system using Fitch-style tokens in the typing
context~\citep{clouston18FitchStyleModalLambda}. However, unlike
\Rizzo, these and subsequent
languages~\citep{bahr22ModalFRPAll,livelyRaTT} with formal guarantees
about space leaks are \emph{synchronous}, i.e.\ there is a single
global clock and all signals update according to this global clock.
Asynchronous FRP for GUIs was pioneered by Elm~\citep{czaplicki13Elm},
which processes signals asynchronously but gives no static guarantees
ruling out space leaks. Modal type systems were later brought to the
asynchronous setting, including languages aimed specifically at
GUIs~\citep{graulund21AdjointReactiveGUI,disch26FunctionalReactiveGUI}
as well as the Async RaTT
calculus~\citep{bahr23AsynchronousModalFRP}. To our
knowledge, only Async RaTT makes formal operational guarantees
including productivity and the absence of space leaks. We refer back
to section~\ref{sec:async-ratt} for a detailed comparison of
\Rizzo and Async RaTT.

\paragraph{Signals as Mutable Cells}
The usual proof technique to establish formal guarantees about space
leaks in modal FRP (first introduced by
\citet{krishnaswami13HigherorderFunctionalReactive}, see above) is
based on an operational semantics that has a special heap to store
delayed computations and evicts such delayed computations from the
heap after each tick of the global clock. By contrast, we use a
special heap for signals rather than delayed computations, and our
operational semantics ensures that we can only store the most recent
value of each signal. The representation of FRP signals as mutable
cells in a dataflow graph is found in other FRP implementations such
as FrTime~\citep{FrTime} and Flapjax~\citep{flapjax}, as well as in
the original work of
\citet{krishnaswami11UltrametricSemanticsReactive,krishnaswami11SemanticModelGraphical}
on modal FRP: Although their denotational semantics defines signals as
infinite streams, which may suffer from space leaks,
\citeauthor{krishnaswami11UltrametricSemanticsReactive} also give a
DSL that implements signals as mutable cells in a dataflow graph --
similar to \Rizzo's heap of stored signals, albeit in a synchronous
setting. Later modal FRP languages, which establish guarantees about
space leaks in the style of
\citet{krishnaswami13HigherorderFunctionalReactive}, did not pick up
this representation. But it is standard in imperative reactive
programming approaches and now underpins mainstream reactive
user-interface frameworks such as SolidJS, Vue, Angular, Svelte, and
the TC39 Signals proposal for JavaScript~\citep{tc39signals}. Unlike
FRP languages that use mutable cells merely as an efficient underlying
implementation of read-only signals, these imperative frameworks
expose signals as \emph{writable} cells that can be mutated through
setters -- typically from event callbacks.

\paragraph{FRP as a Library}
Operational properties of reactive programs can also be ensured by
devising a library with a carefully restricted interface. For example,
Yampa~\citep{nilsson02yampa} uses arrows~\citep{hughes2000arrows} to
ensure causality of signal functions, and
FRPNow!~\citep{ploeg15FRPNow} uses a monadic interface to avoid some
sources of space leaks. A more recent refinement of
Yampa~\citep{barenz18Rhine} annotates signal functions with type-level
clocks, which allows the construction of dataflow graphs that combine
subsystems running at different clock speeds. These type-level clocks
are \emph{statically} determined at compile time with the aim
of providing efficient resampling between subsystems, whereas \Rizzo's
clocks may \emph{dynamically} change over time to allow for
programs with dynamically changing dataflows.


\paragraph{Synchronous Languages} The set of crucial operational
properties we expect from reactive programs depends on the application
domain. For example, most programming languages do not enforce
termination, and similarly we may sacrifice the productivity guarantee
in order to simplify \Rizzo by allowing general recursion. However,
other application domains require \emph{stronger} guarantees such as bounded
memory usage and real-time guarantees for each computation step. For
example, \citet{krishnaswami2012higher} use a linear typing
discipline to obtain static memory bounds for FRP programs. Moreover,
synchronous dataflow languages such as
Esterel~\citep{berry85esterel,berry92esterel},
Lustre~\citep{caspi1987lustre}, and Lucid
Synchrone~\citep{pouzet2006lucid} provide static bounds on runtime.
However, these languages are all limited to a synchronous setting
where all signals are updated according to a global clock. Moreover,
synchronous dataflow languages obtain their strong static guarantees
by enforcing strict limits on the dynamic behaviour, disallowing both
time-varying values of arbitrary types (e.g.\ a signal
of signals is not allowed) and dynamic switching (i.e.\ no functionality equivalent
to the \ensuremath{\Varid{switch}} combinator).  Both Lustre and Lucid Synchrone have a
notion of a local clock, which is a stream of Booleans that indicates,
at each tick of the \emph{global} clock, whether the local clock ticks
as well. In \Rizzo, this notion of local clocks is subsumed by partial
signals, which give rise to clocks via $\ensuremath{\watch\mathbin{:}\Sig\;(\Conid{A}\mathbin{+}\mathrm{1})\to \DelayE\;\Conid{A}}$.

\section{Conclusion and Future Work}
\label{sec:conclusion}

\Rizzo takes a different approach to asynchronous modal FRP. Similarly
to the asynchronous modal FRP languages Async
RaTT~\citep{bahr23AsynchronousModalFRP} and
$\lambda_{\sym{Widget}}$~\citep{graulund21AdjointReactiveGUI}, it uses
modal types to keep track of time. However, unlike these languages,
\Rizzo combines modal types with a forgetful semantics so that signals
can be updated in place. This simplifies the type system and extends
the expressiveness of the language, all while maintaining the
operational guarantees of causality, productivity, and absence of
space leaks.

As discussed in section~\ref{sec:equational-reasoning}, $\beta$- and
$\eta$-equality do not hold universally in Async RaTT and \Rizzo.
This suggests future work in studying the equational theory of
asynchronous FRP and devising a logic for convenient equational
reasoning.

\ifanon
\else
\section*{Acknowledgements}

We thank Patrick Bohn Matthiesen and Tobias Skovbæk Brund for useful
comments and suggestions that improved this article.
\fi

\bibliography{paper}

\appendix
\newpage

\section{Example Runs of the Machine}
\label{sec:more-examples}

\subsection{Sample}
Our example program \ensuremath{\Varid{t}} assumes a channel context $\Delta = \set{\kappa_1 :
\ensuremath{\Chan\;\Nat}, \kappa_2 : \ensuremath{\Chan\;\Conid{Char}}}$:
\[
t = \ensuremath{\mathbf{let}\;\Varid{xs}\mathrel{=}\mathrm{0}\mathbin{::}\Varid{mkSig}\;(\wait\;\kappa_1)\;\mathbf{in}\;\mathbf{let}\;\Varid{ys}\mathrel{=}\text{\ttfamily 'a'}\mathbin{::}\Varid{mkSig}\;(\wait\;\kappa_2)\;\mathbf{in}\;\Varid{sample}\;\Varid{xs}\;\Varid{ys}}
\]
This program first constructs two signals $\ensuremath{\Varid{xs}}$ and $\ensuremath{\Varid{ys}}$ from the
two channels $\kappa_1$ and $\kappa_2$ and then samples from $\ensuremath{\Varid{ys}}$
using $\ensuremath{\Varid{xs}}$. Before we run $t$, we translate the definitions of
$\ensuremath{\Varid{mkSig}}$ and \ensuremath{\Varid{sample}} (as given in
section~\ref{sec:delayed-computations}) from the surface syntax into
the core calculus, similarly to how we translate \ensuremath{\Varid{map}} at the end of
section~\ref{sec:recursion}:
\begin{align*}
\ensuremath{\Varid{sample}} &= \ensuremath{\lambda \Varid{xs}.\lambda \Varid{ys}.\Varid{map}\;(\lambda \Varid{x}.(\Varid{x},\head\;\Varid{ys}))\;\Varid{xs}}\\
\ensuremath{\Varid{mkSig}} &= \ensuremath{\fix\;\Varid{r}.\lambda \Varid{d}.\delay\;(\lambda \Varid{r'}.\lambda \Varid{x}.\Varid{x}\mathbin{::}\Varid{r'}\;\Varid{d})\appA\Varid{r}\appAE\Varid{d}}
\end{align*}
Assuming two signals $l_1$ and $l_2$, the term $\ensuremath{\Varid{sample}\;\Varid{l}_{\mathrm{1}}\;\Varid{l}_{\mathrm{2}}}$
evaluates to a signal with head $(p_1,p_2)$, where each $p_i$ is the
head of $l_i$, and the following tail $\ensuremath{\Varid{sp}}$:
\[
\ensuremath{\Varid{sp}} = \ensuremath{\delay\;((\lambda \Varid{r'}.\Varid{r'}\;(\lambda \Varid{x}.(\Varid{x},\head\;\Varid{l}_{\mathrm{2}})))\;\Varid{map})\appAE\tail\;\Varid{l}_{\mathrm{1}}} 
\]
By definition, $\cl[\eta]{\ensuremath{\Varid{sp}}} = \cl[\eta]{\ensuremath{\tail\;\Varid{l}_{\mathrm{1}}}}$, and thus we
know that the signal produced by \ensuremath{\Varid{sample}} updates whenever the signal
$l_1$ updates.

In turn, assuming a value \ensuremath{\Varid{d}\mathbin{:}\DelayE\;\Conid{A}}, the term $\ensuremath{\Varid{mkSig}\;\Varid{d}}$ evaluates to
the following value $\delaySig{d}$:
\[
\delaySig{d} = \ensuremath{\delay\;((\lambda \Varid{r'}.\lambda \Varid{x}.\Varid{x}\mathbin{::}\Varid{r'}\;\Varid{d})\;\Varid{mkSig})\appAE\Varid{d}}
\]
The notation $\delaySig{d}$ is meant to indicate that $d$ is a
placeholder. For example, $\ensuremath{\Varid{mkSig}\;(\wait\;\kappa_1)}$ evaluates to
$\delaySig{\ensuremath{\wait\;\kappa_1}} = \ensuremath{\delay\;((\lambda \Varid{r'}.\lambda \Varid{x}.\Varid{x}\mathbin{::}\Varid{r'}\;(\wait\;\kappa_1))\;\Varid{mkSig})\appAE\wait\;\kappa_1}$.

Note that, according to the evaluation semantics for \ensuremath{\appA},
the function application that appears as an argument to \ensuremath{\delay} in
both $\ensuremath{\Varid{sp}}$ and $\delaySig{d}$ has not been reduced any further. However, we
can $\beta$-reduce both and obtain semantically equivalent terms that
are easier to read:
\begin{align*}
\ensuremath{\Varid{sp}} &\equiv \ensuremath{\delay\;(\Varid{map}\;(\lambda \Varid{x}.(\Varid{x},\head\;\Varid{l}_{\mathrm{2}})))\appAE\tail\;\Varid{l}_{\mathrm{1}}}\\
\delaySig{d} &\equiv \ensuremath{\delay\;(\lambda \Varid{x}.\Varid{x}\mathbin{::}\Varid{mkSig}\;\Varid{d})\appAE\Varid{d}}
\end{align*}

With these observations, we can now see how \ensuremath{\Varid{t}} reacts to alternating
events from the two channels in $\Delta$. To avoid clutter, we elide
heap locations that are not referenced anywhere and also leave out the
type annotations of all heap locations:
\begin{eqnarray*}
  \state{t;\Delta} &\initstep 
  &\state{l_3;l_1 \mapsto \sig[\unchanged]{0}{\delaySig{\ensuremath{\wait\;\kappa_1}}},
  l_2 \mapsto \sig[\unchanged]{\ensuremath{\text{\ttfamily 'a'}}}{\delaySig{\ensuremath{\wait\;\kappa_2}}},
  l_3 \mapsto \sig[\unchanged]{(0,\ensuremath{\text{\ttfamily 'a'}})}{\ensuremath{\Varid{sp}}}\sep\Delta}\\
  &\fullstep{\kappa_1 \mapsto 1} 
  &\state{l_3;l_1 \mapsto \sig[\updated]{1}{\delaySig{\ensuremath{\wait\;\kappa_1}}},
  l_2 \mapsto \sig[\unchanged]{\ensuremath{\text{\ttfamily 'a'}}}{\delaySig{\ensuremath{\wait\;\kappa_2}}},
  l_3 \mapsto \sig[\updated]{(1,\ensuremath{\text{\ttfamily 'a'}})}{\ensuremath{\Varid{sp}}}\sep\Delta}\\
  &\fullstep{\kappa_2 \mapsto \ensuremath{\text{\ttfamily 'b'}}} 
  &\state{l_3;l_1 \mapsto \sig[\unchanged]{1}{\delaySig{\ensuremath{\wait\;\kappa_1}}},
  l_2 \mapsto \sig[\updated]{\ensuremath{\text{\ttfamily 'b'}}}{\delaySig{\ensuremath{\wait\;\kappa_2}}},
  l_3 \mapsto \sig[\unchanged]{(1,\ensuremath{\text{\ttfamily 'a'}})}{\ensuremath{\Varid{sp}}}\sep\Delta}\\
  &\fullstep{\kappa_1 \mapsto 2} 
  &\state{l_3;l_1 \mapsto \sig[\updated]{2}{\delaySig{\ensuremath{\wait\;\kappa_1}}},
  l_2 \mapsto \sig[\unchanged]{\ensuremath{\text{\ttfamily 'b'}}}{\delaySig{\ensuremath{\wait\;\kappa_2}}},
  l_3 \mapsto \sig[\updated]{(2,\ensuremath{\text{\ttfamily 'b'}})}{\ensuremath{\Varid{sp}}}\sep\Delta}
\end{eqnarray*}
The initialisation step allocates three signals on the heap. Each of
the first two, at $l_1$ and $l_2$, updates whenever the corresponding
channel $\kappa_i$ receives an input since $\cl[\eta]{\ensuremath{\tail\;\Varid{l}_{\Varid{i}}}} =
\cl[\eta]{\delaySig{\ensuremath{\wait\;\kappa_i}}} =  \cl[\eta]{\ensuremath{\wait\;\kappa_i}} =
\set{\kappa_i}$. The signal at $l_3$, however, only updates whenever
$\kappa_1$ receives an input since $\cl[\eta]{\ensuremath{\tail\;\Varid{l}_{\mathrm{3}}}} = \cl[\eta]{\ensuremath{\Varid{sp}}} =
\cl[\eta]{\ensuremath{\tail\;\Varid{l}_{\mathrm{1}}}} = \set{\kappa_1}$. That is, we have the following
reactive evaluation for the sequence of events $\tau = \kappa_1
\mapsto 1,\kappa_2 \mapsto \ensuremath{\text{\ttfamily 'b'}},\kappa_1 \mapsto 2$:
\[
\state{t;\Delta}\reacts{\tau}
\state{\ensuremath{(\mathrm{2},\text{\ttfamily 'b'})\mathbin{::}\delay\;(\Varid{map}\;(\lambda \Varid{x}.(\Varid{x},\head\;(\text{\ttfamily 'b'}\mathbin{::}\delaySig{\wait\;\kappa_2}))))\appAE\delaySig{\wait\;\kappa_1}};\Delta}
\]

\subsection{Filter}

We assume a channel context $\Delta = \set{\kappa_1 : \ensuremath{\Chan\;\Nat}}$
and a function $\ensuremath{\Varid{isEven}\mathbin{:}\Nat\to \Conid{Bool}}$ that checks whether a number
is even:
\[
t = \ensuremath{\mathbf{let}\;\Varid{xs}\mathrel{=}\Varid{mkSig}\;(\wait\;\kappa_1)\;\mathbf{in}\;\mathrm{0}\mathbin{::}\Varid{filter}\;\Varid{isEven}\;\Varid{xs}}
\]
That is, $t$ constructs a delayed signal from $\kappa_1$ and then
filters that signal using the $\ensuremath{\Varid{isEven}}$ predicate.
To construct its result, \ensuremath{\Varid{filter}} takes $\ensuremath{\Varid{isEven}}$ and constructs a new
function \ensuremath{\Varid{p}}:
\[
p = \ensuremath{\lambda \Varid{x}.\mathbf{if}\;\Varid{isEven}\;\Varid{x}\;\mathbf{then}\;\con{just}\;\Varid{x}\;\mathbf{else}\;\con{nothing}}
\]
Assuming the channel $\kappa_1$ receives increasing numbers, $t$
produces the following behaviour:
\begin{eqnarray*}
\state{t;\Delta} &\initstep 
  &\state{l_2;l_1 \mapsto \sig[\unchanged]{\ensuremath{\con{nothing}}}{\ensuremath{\delay\;(\Varid{map}\;\Varid{p})\appAE} \delaySig{\ensuremath{\wait\;\kappa_1}}},
  l_2 \mapsto \sig[\unchanged]{0}{\delaySig{\ensuremath{\watch\;\Varid{l}_{\mathrm{1}}}}}\sep\Delta} \\
   &\fullstep{\kappa_1\mapsto 1} 
  & 
  \begin{aligned}[t]
  \langle l_2;
    &l_3 \mapsto \sig[\unchanged]{1}{\delaySig{\ensuremath{\wait\;\kappa_1}}},
     l_1 \mapsto \sig[\updated]{\ensuremath{\con{nothing}}}{\ensuremath{\delay\;((\lambda \Varid{r'}.\Varid{r'}\;\Varid{p})\;\Varid{map})\appAE\tail\;\Varid{l}_{\mathrm{3}}}},\\ 
    &l_2 \mapsto \sig[\unchanged]{0}{\delaySig{\ensuremath{\watch\;\Varid{l}_{\mathrm{1}}}}}\sep\Delta\rangle
  \end{aligned} \\
   &\fullstep{\kappa_1\mapsto 2} 
  & 
  \begin{aligned}[t]
  \langle l_2;
  &l_3 \mapsto \sig[\updated]{2}{\delaySig{\ensuremath{\wait\;\kappa_1}}},
   l_1 \mapsto \sig[\updated]{\ensuremath{\con{just}\;\mathrm{2}}}{\ensuremath{\delay\;((\lambda \Varid{r'}.\Varid{r'}\;\Varid{p})\;\Varid{map})\appAE\tail\;\Varid{l}_{\mathrm{3}}}},\\ 
  &l_2 \mapsto \sig[\updated]{2}{\delaySig{\ensuremath{\watch\;\Varid{l}_{\mathrm{1}}}}}\sep\Delta\rangle
  \end{aligned}
\end{eqnarray*}
The initialisation step allocates two signals. The signal of type
$\ensuremath{\Sig\;(\Conid{Maybe}\;\Nat)}$ constructed by $\ensuremath{\Varid{filter}\;\Varid{isEven}\;\Varid{xs}}$ using $\ensuremath{\con{nothing}\mathbin{::}(\Varid{map}\;\Varid{p}\appIE\Varid{xs})}$ is stored at $l_1$. The signal of type $\ensuremath{\Sig\;\Nat}$
stored at $l_2$ is the result of evaluating $\ensuremath{\mathrm{0}\mathbin{::}\Varid{mkSig}\;(\watch\;\Varid{l}_{\mathrm{1}})}$. When the first input on $\kappa_1$ arrives, signal $l_1$ needs
updating since $\cl[\eta]{\ensuremath{\tail\;\Varid{l}_{\mathrm{1}}}} = \cl[\eta]{\delaySig{\ensuremath{\wait\;\kappa_1}}} =
\cl[\eta]{\ensuremath{\wait\;\kappa_1}} =  \set{\kappa_1}$. This requires the delayed
computation $\delaySig{\ensuremath{\wait\;\kappa_1}}$ to be performed  by the advance
semantics, which in turn results in the allocation of the signal at
$l_3$. However, the signal at $l_2$ is not updated since $\cl[\eta]{\ensuremath{\tail\;\Varid{l}_{\mathrm{2}}}} = \cl[\eta]{\delaySig{\ensuremath{\watch\;\Varid{l}_{\mathrm{1}}}}} =  \cl[\eta]{\ensuremath{\watch\;\Varid{l}_{\mathrm{1}}}} =  \set{l_1}$,
and while the signal at $l_1$ was updated (indicated by $\updated$),
its current value is not of the form $\ensuremath{\con{just}\;\Varid{q}}$. Only after receiving
the next input on channel $\kappa_1$, the signal at $l_1$ is updated
to have current value $\ensuremath{\con{just}\;\mathrm{2}}$ and thus also $l_2$ is updated. That
is, the events $\tau = \kappa_1 \mapsto 1,\kappa_1 \mapsto 2$
cause the following reactive evaluation:
\[
\state{t;\Delta}\reacts{\tau}
\state{\ensuremath{\mathrm{2}\mathbin{::}\delaySig{\watch\;(\con{just}\;\mathrm{2}\mathbin{::}\delay\;((\lambda \Varid{r'}.\Varid{r'}\;\Varid{p})\;\Varid{map})\appAE\delaySig{\wait\;\kappa_1})}};\Delta}
\]

\subsection{Switch}
We assume a channel context $\Delta = \set{\kappa_1 : \ensuremath{\Chan\;\Nat},
\kappa_2 : \ensuremath{\Chan\;\Nat}}$ to implement the following example involving
\ensuremath{\Varid{switch}}:
\[
t = \ensuremath{\mathbf{let}\;\Varid{xs}\mathrel{=}\mathrm{0}\mathbin{::}\Varid{mkSig}\;(\wait\;\kappa_1)\;\mathbf{in}\;\mathbf{let}\;\Varid{ys}\mathrel{=}\Varid{mkSig}\;(\wait\;\kappa_2)\;\mathbf{in}\;\Varid{switch}\;\Varid{xs}\;\Varid{ys}}
\]
That is, the signal produced by \ensuremath{\Varid{t}} first behaves like the signal
produced by $\kappa_1$, but then behaves as the signal produced by
$\kappa_2$ as soon as $\kappa_2$ receives its first input.

We again translate the definition of \ensuremath{\Varid{switch}} into the core calculus
of \Rizzo, similarly to how we translate \ensuremath{\Varid{map}} and \ensuremath{\Varid{sample}}:
\[
\begin{aligned}
\ensuremath{\Varid{switch}} &= \ensuremath{\fix\;\Varid{r}.\lambda \Varid{s}.\lambda \Varid{d}.\mathbf{let}\;\Varid{x}\mathrel{=}\head\;\Varid{s}\;\mathbf{in}\;\mathbf{let}\;\Varid{xs}\mathrel{=}\tail\;\Varid{s}\;\mathbf{in}\;\Varid{x}\mathbin{::}(\delay\;\Varid{cont}\appA\Varid{r}\appAE\select\;\Varid{xs}\;\Varid{d})}\\
\ensuremath{\Varid{cont}} &= \ensuremath{\lambda \Varid{r'}.\lambda \Varid{y}.\mathbf{case}\;\Varid{y}\;\mathbf{of}\;\{\mskip1.5mu \con{left}\;\Varid{z}.\Varid{r'}\;\Varid{z}\;\Varid{d};\con{right}\;\Varid{d'}.\Varid{d'};\con{both}\;\Varid{u}\;\Varid{d'}.\Varid{d'}\mskip1.5mu\}}
\end{aligned}
\]
Assuming a signal \ensuremath{\Varid{l}_{\mathrm{1}}} and a delayed signal \ensuremath{\Varid{d}\mathbin{:}\DelayE\;(\Sig\;\Nat)}, the
term \ensuremath{\Varid{switch}\;\Varid{l}_{\mathrm{1}}\;\Varid{d}} evaluates to a signal whose head is the head of
\ensuremath{\Varid{l}_{\mathrm{1}}} and whose tail \ensuremath{\Varid{sw}} is
\[
\ensuremath{\Varid{sw}} = \ensuremath{\delay\;(\Varid{cont}\;\Varid{switch})\appAE\select\;(\tail\;\Varid{l}_{\mathrm{1}})\;\Varid{d}},
\]
which we may $\beta$-reduce to the more readable
\[
\begin{aligned}
\ensuremath{\Varid{sw}} &\equiv \ensuremath{\delay\;\Varid{cont'}\appAE\select\;(\tail\;\Varid{l}_{\mathrm{1}})\;\Varid{d}}\\
\ensuremath{\Varid{cont'}} &= \ensuremath{\lambda \Varid{y}.\mathbf{case}\;\Varid{y}\;\mathbf{of}\;\{\mskip1.5mu \con{left}\;\Varid{z}.\Varid{switch}\;\Varid{z}\;\Varid{d};\con{right}\;\Varid{d'}.\Varid{d'};\con{both}\;\Varid{u}\;\Varid{d'}.\Varid{d'}\mskip1.5mu\}}
\end{aligned}
\]
By definition, $\cl[\eta]{\ensuremath{\Varid{sw}}} = \cl[\eta]{\ensuremath{\select\;(\tail\;\Varid{l}_{\mathrm{1}})\;\Varid{d}}} =
\cl[\eta]{\ensuremath{\tail\;\Varid{l}_{\mathrm{1}}}} \cup \cl[\eta]{\ensuremath{\Varid{d}}}$. If \ensuremath{\Varid{d}} is the result of
evaluating \ensuremath{\Varid{mkSig}\;(\wait\;\kappa_2)}, namely $d =\delaySig{\ensuremath{\wait\;\kappa_2}}$, then $\cl[\eta]{d} = \set{\kappa_2}$, and
$\cl[\eta]{\ensuremath{\tail\;\Varid{l}_{\mathrm{1}}}} = \set{\kappa_1}$; hence $\cl[\eta]{\ensuremath{\Varid{sw}}} =
\set{\kappa_1,\kappa_2}$, i.e.\ the switched signal updates whenever
\emph{either} \ensuremath{\Varid{l}_{\mathrm{1}}} (the current signal) or \ensuremath{\Varid{d}} (the switch trigger)
ticks.

As in previous examples, we elide heap locations (except $l_1$) that
are not referenced anywhere as well as the type annotations of heap
locations. The program \ensuremath{\Varid{t}} then reacts to the events $\kappa_1 \mapsto
1, \kappa_1 \mapsto 2, \kappa_2 \mapsto 5, \kappa_1 \mapsto 3$ as
follows:
\begin{eqnarray*}
  \state{t;\Delta} &\initstep
  &\state{l_2;l_1 \mapsto \sig[\unchanged]{0}{\delaySig{\ensuremath{\wait\;\kappa_1}}},
  l_2 \mapsto \sig[\unchanged]{0}{\ensuremath{\Varid{sw}}}\sep\Delta}\\
  &\fullstep{\kappa_1 \mapsto 1}
  &\state{l_2;l_1 \mapsto \sig[\updated]{1}{\delaySig{\ensuremath{\wait\;\kappa_1}}},
  l_2 \mapsto \sig[\updated]{1}{\ensuremath{\Varid{sw}}}\sep\Delta}\\
  &\fullstep{\kappa_1 \mapsto 2}
  &\state{l_2;l_1 \mapsto \sig[\updated]{2}{\delaySig{\ensuremath{\wait\;\kappa_1}}},
  l_2 \mapsto \sig[\updated]{2}{\ensuremath{\Varid{sw}}}\sep\Delta}\\
  &\fullstep{\kappa_2 \mapsto 5}
  &\state{l_2;l_1 \mapsto \sig[\unchanged]{2}{\delaySig{\ensuremath{\wait\;\kappa_1}}},
  l_2 \mapsto \sig[\updated]{5}{\delaySig{\ensuremath{\wait\;\kappa_2}}}\sep\Delta}\\
  &\fullstep{\kappa_1 \mapsto 3}
  &\state{l_2;l_1 \mapsto \sig[\updated]{3}{\delaySig{\ensuremath{\wait\;\kappa_1}}},
  l_2 \mapsto \sig[\unchanged]{5}{\delaySig{\ensuremath{\wait\;\kappa_2}}}\sep\Delta}
\end{eqnarray*}
The initialisation step allocates two signals: the initial signal \ensuremath{\Varid{xs}}
at \ensuremath{\Varid{l}_{\mathrm{1}}} and the \ensuremath{\Varid{switch}} result at \ensuremath{\Varid{l}_{\mathrm{2}}} (the result of evaluating
\ensuremath{\Varid{t}}). The signal at \ensuremath{\Varid{l}_{\mathrm{1}}} updates whenever $\kappa_1$ receives an
input, since $\cl[\eta]{\ensuremath{\tail\;\Varid{l}_{\mathrm{1}}}} = \set{\kappa_1}$, whereas the signal at
\ensuremath{\Varid{l}_{\mathrm{2}}} updates whenever $\kappa_1$ \emph{or} $\kappa_2$ receives an
input, since $\cl[\eta]{\ensuremath{\tail\;\Varid{l}_{\mathrm{2}}}} = \cl[\eta]{\ensuremath{\Varid{sw}}} = \set{\kappa_1,\kappa_2}$.

For the first two events $\kappa_1 \mapsto 1$ and $\kappa_1 \mapsto
2$, both signals tick: the \ensuremath{\select} in \ensuremath{\Varid{sw}} fires on its left branch,
so \ensuremath{\Varid{switch}} recurses and \ensuremath{\Varid{l}_{\mathrm{2}}} keeps the tail \ensuremath{\Varid{sw}} while its head
follows \ensuremath{\Varid{l}_{\mathrm{1}}}. The event $\kappa_2 \mapsto 5$ then \emph{fires the
switch}: \ensuremath{\Varid{l}_{\mathrm{2}}} ticks (as $\kappa_2 \in \cl[\eta]{\ensuremath{\Varid{sw}}}$) but \ensuremath{\Varid{l}_{\mathrm{1}}}
does not (as $\kappa_2 \notin \set{\kappa_1}$), and so the \ensuremath{\select}
fires on its right branch, where \ensuremath{\Varid{cont'}} returns the delivered signal
\ensuremath{\Varid{d'}}. From this point on \ensuremath{\Varid{l}_{\mathrm{2}}} \emph{is} the signal \ensuremath{\Varid{ys}}: its head
becomes $5$ and its tail becomes $\delaySig{\ensuremath{\wait\;\kappa_2}}$, whose
clock is $\set{\kappa_2}$. Consequently, the final event $\kappa_1
\mapsto 3$ updates only \ensuremath{\Varid{l}_{\mathrm{1}}}; the switched signal at \ensuremath{\Varid{l}_{\mathrm{2}}} no longer
reacts to $\kappa_1$ and retains the value $5$.

That is, the events $\tau = \kappa_1 \mapsto 1, \kappa_1 \mapsto 2,
\kappa_2 \mapsto 5, \kappa_1 \mapsto 3$ 
cause the following reactive evaluation:
\[
\state{t;\Delta}\reacts{\tau}
\state{\ensuremath{\mathrm{5}\mathbin{::}\delaySig{\wait\;\kappa_2}};\Delta}.
\]
\section{\ensuremath{\Varid{switchAt}} in Async RaTT}
\label{sec:switchAt}

We implement the \ensuremath{\Varid{switchAt}} combinator introduced in
section~\ref{sec:expressiveness} using Async
RaTT~\citep{bahr23AsynchronousModalFRP}. To this end, we implement two
helper functions. The first one \ensuremath{\Varid{del}} takes a signal \ensuremath{\Varid{s}\mathbin{:}\Sig\;\Conid{C}} and a
delayed unit \ensuremath{\Varid{d}\mathbin{:}\DelayE\;\mathrm{1}} and tries to delay \ensuremath{\Varid{s}} until after \ensuremath{\Varid{d}} ticks.
If this fails, i.e.\ if the tail of \ensuremath{\Varid{s}} arrives before \ensuremath{\Varid{d}}, then the
returned \ensuremath{\Conid{Maybe}} value contains \ensuremath{\con{just}\;\Varid{d}} (but moved into the future
where the tail of \ensuremath{\Varid{s}} ticked). Otherwise, the \ensuremath{\Conid{Maybe}} value is
\ensuremath{\con{nothing}}:
\begin{hscode}\SaveRestoreHook
\column{B}{@{}>{\hspre}l<{\hspost}@{}}%
\column{28}{@{}>{\hspre}l<{\hspost}@{}}%
\column{31}{@{}>{\hspre}l<{\hspost}@{}}%
\column{38}{@{}>{\hspre}l<{\hspost}@{}}%
\column{43}{@{}>{\hspre}l<{\hspost}@{}}%
\column{47}{@{}>{\hspre}l<{\hspost}@{}}%
\column{E}{@{}>{\hspre}l<{\hspost}@{}}%
\>[B]{}\Varid{del}\mathbin{:}\Varid{stable}\;\Conid{C}\Rightarrow \Sig\;\Conid{C}\to \DelayE\;\mathrm{1}\to \DelayE\;(\Conid{Maybe}\;(\DelayE\;\mathrm{1})\times\Sig\;\Conid{C}){}\<[E]%
\\
\>[B]{}\Varid{del}\;(\Varid{x}\mathbin{::}\Varid{xs})\;\Varid{d}\mathrel{=}\delay\;({}\<[28]%
\>[28]{}\mathbf{case}\;\select\;\Varid{xs}\;\Varid{d}\;\mathbf{of}{}\<[E]%
\\
\>[28]{}\hsindent{3}{}\<[31]%
\>[31]{}\con{left}\;{}\<[38]%
\>[38]{}\Varid{xs'}\;{}\<[43]%
\>[43]{}\Varid{d'}{}\<[47]%
\>[47]{}.(\con{just}\;\Varid{d'},\Varid{xs'}){}\<[E]%
\\
\>[28]{}\hsindent{3}{}\<[31]%
\>[31]{}\con{right}\;{}\<[38]%
\>[38]{}\Varid{xs'}\;{}\<[43]%
\>[43]{}\anonymous {}\<[47]%
\>[47]{}.(\con{nothing},\Varid{x}\mathbin{::}\Varid{xs'}){}\<[E]%
\\
\>[28]{}\hsindent{3}{}\<[31]%
\>[31]{}\con{both}\;{}\<[38]%
\>[38]{}\Varid{xs'}\;{}\<[43]%
\>[43]{}\anonymous {}\<[47]%
\>[47]{}.(\con{nothing},\Varid{xs'})){}\<[E]%
\ColumnHook
\end{hscode}\resethooks
We follow the convention of \citet{bahr23AsynchronousModalFRP} and
write \ensuremath{\Varid{stable}\;\Conid{C}} to indicate that the above function requires that the
type \ensuremath{\Conid{C}} is stable. Note also that, unlike \Rizzo's \ensuremath{\select}, the
\ensuremath{\select} of Async RaTT returns in its \ensuremath{\con{left}} (resp.\ \ensuremath{\con{right}}) branch
not only the value produced by the first (resp.\ second) argument but
also the still-delayed computation of the other argument. This is why,
for example, the \ensuremath{\con{left}} branch of \ensuremath{\Varid{del}} binds both the new signal
\ensuremath{\Varid{xs'}} and the still-pending delay \ensuremath{\Varid{d'}} of \ensuremath{\Varid{d}}.

The function \ensuremath{\Varid{switchN}} implements the main logic of \ensuremath{\Varid{switchAt}}. It
takes a signal \ensuremath{\Varid{xs}\mathbin{:}\Sig\;\Conid{C}} and a delayed value \ensuremath{\Varid{e}\mathbin{:}\DelayE\;(\Conid{Maybe}\;(\DelayE\;\mathrm{1})\times\Sig\;\Conid{C})} (which is the result of \ensuremath{\Varid{del}\;\Varid{ys}\;\Varid{d}}) and produces a signal that
first behaves like \ensuremath{\Varid{xs}} but switches to \ensuremath{\Varid{ys}} as soon as \ensuremath{\Varid{d}} arrives.
To this end, it repeatedly calls \ensuremath{\Varid{del}} if the delayed value \ensuremath{\Varid{e}}
produces a \ensuremath{\con{just}} value, which indicates that \ensuremath{\Varid{d}} has not arrived yet:
\begin{hscode}\SaveRestoreHook
\column{B}{@{}>{\hspre}l<{\hspost}@{}}%
\column{37}{@{}>{\hspre}l<{\hspost}@{}}%
\column{39}{@{}>{\hspre}l<{\hspost}@{}}%
\column{46}{@{}>{\hspre}l<{\hspost}@{}}%
\column{65}{@{}>{\hspre}l<{\hspost}@{}}%
\column{E}{@{}>{\hspre}l<{\hspost}@{}}%
\>[B]{}\Varid{switchN}\mathbin{:}\Varid{stable}\;\Conid{C}\Rightarrow \Sig\;\Conid{C}\to \DelayE\;(\Conid{Maybe}\;(\DelayE\;\mathrm{1})\times\Sig\;\Conid{C})\to \Sig\;\Conid{C}{}\<[E]%
\\
\>[B]{}\Varid{switchN}\;(\Varid{x}\mathbin{::}\Varid{xs})\;\Varid{e}\mathrel{=}\Varid{x}\mathbin{::}\delay\;({}\<[37]%
\>[37]{}\mathbf{case}\;\select\;\Varid{xs}\;\Varid{e}\;\mathbf{of}{}\<[E]%
\\
\>[37]{}\hsindent{2}{}\<[39]%
\>[39]{}\con{left}\;{}\<[46]%
\>[46]{}\Varid{xs'}\;\Varid{e'}{}\<[65]%
\>[65]{}.\Varid{switchN}\;\Varid{xs'}\;\Varid{e'}{}\<[E]%
\\
\>[37]{}\hsindent{2}{}\<[39]%
\>[39]{}\con{right}\;{}\<[46]%
\>[46]{}\Varid{xs'}\;(\con{just}\;\Varid{d'},\Varid{ys'}){}\<[65]%
\>[65]{}.\Varid{switchN}\;(\Varid{x}\mathbin{::}\Varid{xs'})\;(\Varid{del}\;\Varid{ys'}\;\Varid{d'}){}\<[E]%
\\
\>[37]{}\hsindent{2}{}\<[39]%
\>[39]{}\con{right}\;{}\<[46]%
\>[46]{}\Varid{xs'}\;(\con{nothing},\Varid{ys'}){}\<[65]%
\>[65]{}.\Varid{ys'}{}\<[E]%
\\
\>[37]{}\hsindent{2}{}\<[39]%
\>[39]{}\con{both}\;{}\<[46]%
\>[46]{}\Varid{xs'}\;(\con{just}\;\Varid{d'},\Varid{ys'}){}\<[65]%
\>[65]{}.\Varid{switchN}\;\Varid{xs'}\;(\Varid{del}\;\Varid{ys'}\;\Varid{d'}){}\<[E]%
\\
\>[37]{}\hsindent{2}{}\<[39]%
\>[39]{}\con{both}\;{}\<[46]%
\>[46]{}\Varid{xs'}\;(\con{nothing},\Varid{ys'}){}\<[65]%
\>[65]{}.\Varid{ys'}){}\<[E]%
\\[\blanklineskip]%
\>[B]{}\Varid{switchAt}\mathbin{:}\Varid{stable}\;\Conid{C}\Rightarrow \Sig\;\Conid{C}\to \Sig\;\Conid{C}\to \DelayE\;\mathrm{1}\to \Sig\;\Conid{C}{}\<[E]%
\\
\>[B]{}\Varid{switchAt}\;\Varid{xs}\;\Varid{ys}\;\Varid{d}\mathrel{=}\Varid{switchN}\;\Varid{xs}\;(\Varid{del}\;\Varid{ys}\;\Varid{d}){}\<[E]%
\ColumnHook
\end{hscode}\resethooks
Note that this implementation has a less general type than the
corresponding \Rizzo implementation in
section~\ref{sec:expressiveness} as it requires that \ensuremath{\Conid{C}} be stable.

\end{document}